\documentclass[prodmode,acmtocl]{acmsmall}

\acmVolume{0}
\acmNumber{0}
\acmArticle{0}
\acmYear{0}
\acmMonth{0}
\doi{0}

\usepackage{amsmath,amssymb,amsfonts,mathrsfs,latexsym,stmaryrd}
\usepackage{xspace}
\usepackage[table]{xcolor}
\usepackage{pseudocode}
\usepackage{tikz}
\usetikzlibrary{snakes,arrows,shapes,er,positioning}
\usetikzlibrary{decorations,decorations.pathmorphing}
\usetikzlibrary{patterns,calc}

\usepackage[newitem,newenum,neverdecrease]{paralist}
\usepackage{color}
\usepackage{multirow}
\usepackage{graphicx}

\usepackage{url}
\usepackage{pifont}

\renewcommand{\qed}{\ding{113}}

\newtheorem{claim}{{\sc Claim}}[section]
\newtheorem{question}{{\sc Question}}

 \newcommand{\Z}{\mathbb{Z}}

\newcommand{\auf}{\langle}
\newcommand{\zu}{\rangle}

\newcommand{\avec}[1]{\boldsymbol{#1}}

\newcommand{\OWL}{\textsl{OWL\,2}}
\newcommand{\OWLQL}{\textsl{OWL\,2\,QL}}
\newcommand{\OWLEL}{\textsl{OWL\,2\,EL}}

\renewcommand{\int}{\mathsf{int}}
\newcommand{\hinte}{\textit{h\_int}}
\newcommand{\vinte}{\textit{v\_int}}
\newcommand{\sqe}{\textit{square}}

\newcommand{\nwplane}{\textit{nw}_{\omega\times\omega}}

\newcommand{\R}{\ensuremath{\mathsf{R}}}
\newcommand{\after}{\ensuremath{\mathsf{A}}}
\newcommand{\afterb}{\ensuremath{\mathsf{\bar{A}}}}
\newcommand{\later}{\ensuremath{\mathsf{L}}}
\newcommand{\laterb}{\ensuremath{\mathsf{\bar{L}}}}
\newcommand{\begins}{\ensuremath{\mathsf{B}}}
\newcommand{\beginsb}{\ensuremath{\mathsf{\bar{B}}}}
\newcommand{\eds}{\ensuremath{\mathsf{E}}}
\newcommand{\edsb}{\ensuremath{\mathsf{\bar{E}}}}
\newcommand{\during}{\ensuremath{\mathsf{D}}}
\newcommand{\duringb}{\ensuremath{\mathsf{\bar{D}}}}
\newcommand{\overlaps}{\ensuremath{\mathsf{O}}}
\newcommand{\overlapsb}{\ensuremath{\mathsf{\bar{O}}}}
\newcommand{\univ}{\mathsf{U}}

\newcommand{\CC}{\mathcal{C}}
\newcommand{\F}{\mathfrak{F}}

\newcommand{\Fin}{\mathsf{Fin}}
\newcommand{\Lin}{\mathsf{Lin}}
\newcommand{\Dis}{\mathsf{Dis}}
\newcommand{\Disinf}{\mathsf{Dis}^\infty}
\newcommand{\Den}{\mathsf{Den}}

\newcommand{\T}{\mathfrak{T}}
\newcommand{\M}{\mathfrak{M}}

\newcommand{\C}{\ensuremath{\mathfrak{K}}}

\newcommand{\D}{\Diamond}

\newcommand{\U}{[\univ]}
\newcommand{\LD}{\langle \later\rangle}

\newcommand{\EbD}{\langle \edsb\rangle}
\newcommand{\EbB}{[\edsb]\,}

\newcommand{\AD}{\langle \after\rangle}
\newcommand{\AbD}{\langle \afterb\rangle}
\newcommand{\AbB}{[\afterb]\,}
\newcommand{\AB}{[\after]\,}
\newcommand{\BD}{\langle \begins\rangle}

\newcommand{\BB}{[\begins]}
\newcommand{\DD}{\langle \during\rangle}
\newcommand{\DB}{[\during]}
\newcommand{\OD}{\langle \overlaps\rangle}
\newcommand{\DbD}{\langle \duringb\rangle}
\newcommand{\ObD}{\langle \overlapsb\rangle}
\newcommand{\Bh}{[\eds]\,}
\newcommand{\Dh}{\langle \eds\rangle}
\newcommand{\Bv}{[\beginsb]}
\newcommand{\Dv}{\langle\beginsb\rangle}

\newcommand{\evar}{{\sf end}}

\newcommand{\lvar}{{\sf line}}
\newcommand{\diag}{{\sf diag}}
\newcommand{\notdiag}{\overline{\sf diag}}
\newcommand{\plane}{{\sf unit}}

\newcommand{\ndiag}{{\sf next}}

\newcommand{\uvar}{{\sf up}}
\newcommand{\wvar}{{\sf wall}}

\newcommand{\uvaru}{{\sf up}_{\uparrow}}
\newcommand{\uvarr}{{\sf up}_{\to}}
\newcommand{\uvarn}{{\sf up}^{+}}
\newcommand{\now}{{\sf now}}

\newcommand{\pred}{{\sf P}}
\newcommand{\predp}{{\sf P}'}

\newcommand{\Htick}{{\sf Htick}}
\newcommand{\Vtick}{{\sf Vtick}}
\newcommand{\nHtick}{\overline{{\sf Htick}}}
\newcommand{\nVtick}{\overline{{\sf Vtick}}}

\newcommand{\mpred}{{\sf M}_{\lambda}}
\newcommand{\xpred}{{\sf X}_{\lambda}}
\newcommand{\ypred}{{\sf Y}_{\lambda}}
\newcommand{\qpred}{{\sf Q}}
\newcommand{\qpredp}{{\sf Q}'}

\newcommand{\init}{{\sf init}}
\newcommand{\uvarlast}{{\sf last\_up}}
\newcommand{\gridp}{{\sf grid\_proper}}
\newcommand{\mvar}{{\sf mirror}}
\newcommand{\gridc}{{\sf grid\_copy}}
\newcommand{\gridm}{{\sf first\_mirror}}
\newcommand{\mvarlast}{{\sf last\_mirror}}

\newcommand{\cellvar}{{\sf cell}}

\newcommand{\RhM}{\prec_\to^{\M}}

\newcommand{\RvM}{\prec_\uparrow^{\M}}

\newcommand{\imph}{\Rightarrow_{\!H}}
\newcommand{\impv}{\Rightarrow_V}
\newcommand{\imphd}{\Rightarrow^d_{\!H}\,} 

\newcommand{\hcoverf}{{\sf Cover}_\leftrightarrow}
\newcommand{\vcoverf}{{\sf Cover}_\updownarrow}
\newcommand{\succh}{{\sf grid\_succ}_\to}
\newcommand{\succv}{{\sf grid\_succ}_\uparrow}

\newcommand{\cbr}{{\sf Chessboard}}
\newcommand{\hsuc}{{\sf succ\_sq}_\to}
\newcommand{\vsuc}{{\sf succ\_sq}_\uparrow}
\newcommand{\ffill}{{\sf fill}}

\newcommand{\fdiag}{\phi_{\textit{enum}}}
\newcommand{\fdiagr}{\phi_{\textit{enum}}^r}
\newcommand{\fdiagb}{\phi_{\textit{enum}}^\Box}
\newcommand{\fgridcore}{\phi_{\textit{grid}}^{\textit{core}}}
\newcommand{\fgrid}{\phi_{\textit{grid}}}
\newcommand{\fgridr}{\phi_{\textit{grid}}^r}
\newcommand{\fgridb}{\phi_{\textit{grid}}^\Box}
\newcommand{\ftmpspace}{\Phi_{\!\A}}
\newcommand{\ftmpspaced}{\Phi_{\!\A}^d}
\newcommand{\ftm}{\Psi_{\!\A}}
\newcommand{\ftmr}{\Psi_{\!\A}^r}
\newcommand{\ftmb}{\Psi_{\!\A}^\Box}

\newcommand{\nwnxt}{\Rightarrow\!\begin{tikzpicture}[>=latex,yscale=0.5,xscale=0.5]\footnotesize
\draw (0,0) circle (3mm);
\draw[->] (0.2,-0.2) -- (-0.2,0.23);
\end{tikzpicture}\,}%

\newcommand{\HS}{\ensuremath{\mathcal{HS}}}
\newcommand{\HSh}{\ensuremath{\mathcal{HS}_\textit{horn}}}
\newcommand{\HSc}{\ensuremath{\mathcal{HS}_\textit{core}}}
\newcommand{\HShb}{\ensuremath{\mathcal{HS}^\Box_\textit{horn}}}
\newcommand{\HShd}{\ensuremath{\mathcal{HS}^\Diamond_\textit{horn}}}
\newcommand{\HScb}{\ensuremath{\mathcal{HS}^\Box_\textit{core}}}
\newcommand{\HScd}{\ensuremath{\mathcal{HS}^\Diamond_\textit{core}}}

\newcommand{\NLogSpace}{\textsc{NLogSpace}}
\newcommand{\PTime}{\textsc{P}}
\newcommand{\ExpTime}{\textsc{ExpTime}}
\newcommand{\NExpTime}{\textsc{NExpTime}}
\newcommand{\NP}{\textsc{NP}}
\newcommand{\PSpace}{\textsc{PSpace}}
\newcommand{\ExpSpace}{\textsc{ExpSpace}}

\newcommand{\A}{\ensuremath{\mathcal{A}}}

\newcommand{\LEnd}{\textit{LEnd}}
\newcommand{\GA}{\Gamma_{\!\A}}
\newcommand{\WA}{W_{\!\A}}

\newcommand{\CDT}{\ensuremath{\mathcal{CDT}}}
\newcommand{\SL}{\ensuremath{\mathcal{SL}}}

\newcommand{\chase}{\mathfrak V}

\newcommand{\cl}{\mathsf{cl}}

\newcommand{\tml}{\boldsymbol{l}}
\newcommand{\tmr}{\boldsymbol{r}}

\begin{document}

\markboth{D.~Bresolin, A. Kurucz, E.~Mu\~noz, V.~Ryzhikov, G.~Sciavicco, M.~Zakharyaschev}{Horn Fragments of the Halpern-Shoham Interval Temporal Logic}

\title{Horn Fragments of the Halpern-Shoham Interval Temporal Logic}

\author{DAVIDE BRESOLIN
\affil{University of Padova, Italy}
AGI KURUCZ
\affil{King's College London, UK}
EMILIO MU\~NOZ-VELASCO
\affil{University of Malaga, Spain}
VLADISLAV RYZHIKOV
\affil{Free University of Bozen-Bolzano, Italy}
GUIDO SCIAVICCO
\affil{University of Ferrara, Italy}
MICHAEL ZAKHARYASCHEV
\affil{Birkbeck, University of London, UK}
}

\begin{abstract}
We investigate the satisfiability problem for Horn fragments of the Halpern-Shoham interval temporal logic depending on the type (box or diamond) of the interval modal operators, the type of the underlying linear order (discrete or dense), and the type of semantics for the interval relations (reflexive or irreflexive). For example, we show that satisfiability of Horn formulas with diamonds  is undecidable for any type of linear orders and semantics. On the contrary, satisfiability of Horn formulas with boxes is tractable over both discrete and dense orders under the reflexive semantics and over dense orders under the irreflexive semantics, but becomes undecidable over discrete orders under the irreflexive semantics. Satisfiability of  binary Horn formulas with both boxes and diamonds is always undecidable under the irreflexive semantics.
\end{abstract}

%





\acmformat{Davide Bresolin, Agi Kurucz, Emilio Mu\~noz-Velasco, Vladislav Ryzhikov, Guido Sciavicco, and Michael Zakharyaschev. 2017. Horn Fragments of the Halpern-Shoham Interval Temporal Logic.}

\def\copyrightline{}

\makeatletter
\def\endbottomstuff{\vspace*{30mm}\par
\doiline
\vskip-13pt
\strut
\end@float}
\makeatother

\begin{bottomstuff}
The authors acknowledge the support from the Italian INDAM-GNCS project 2017 \emph{`Logics and Automata for Interval Model Checking'} (D.~Bresolin, G.~Sciavicco), the Spanish project \emph{TIN15-70266-C2-P-1} (E.~Mu\~noz-Velasco), and the EPSRC UK project EP/M012670 `\emph{iTract: Islands of Tractability in Ontology-Based Data Access}' (M.~Zakharyaschev).
\end{bottomstuff}

\maketitle


\section{Introduction}

Our concern in this paper is the satisfiability problem for Horn fragments of the interval temporal (or modal) logic introduced by Halpern and Shoham~\citeyear{HalpernS91} and known since then under the moniker $\HS$. Syntactically, \HS{} is a classical  propositional logic with modal diamond operators of the form $\langle \R \rangle$, where $\R$ is one of Allen's~\citeyear{allen83} twelve interval relations: \emph{After}, \emph{Begins}, \emph{Ends}, \emph{During}, \emph{Later}, \emph{Overlaps} and their inverses. The propositional variables of \HS{} are interpreted by sets of  closed intervals $[i,j]$ of some flow of time (such as $\mathbb Z$, $\mathbb R$, etc.), and a formula $\langle \R\rangle \varphi$ is regarded to be true in $[i,j]$ if and only if $\varphi$ is true in some interval $[i',j']$ such that $[i,j] \R [i',j']$ in Allen's interval algebra.

The elegance and expressive power of \HS{} have attracted attention of the temporal and modal communities, as well as many other areas of computer science, AI, philosophy and linguistics, e.g.,~\cite{AllenAI84,CauHDZMMS02,IMM2004-02867,CimattiRT15,DMGMS11,pratt}.
However, promising applications have been hampered by the fact, already discovered by Halpern and Shoham~\citeyear{HalpernS91}, that \HS{} is highly undecidable (for example, validity over $\mathbb Z$ and $\mathbb R$ is $\Pi^1_1$-hard).

A quest for `tame' fragments of \HS{} began in the 2000s, and has resulted in a substantial body of literature that identified a number of ways of reducing the expressive power of $\HS$:
\begin{itemize}
\item \emph{Constraining the underlying temporal structures.} Montanari et al.~\citeyear{Montanari02} interpreted their Split Logic \SL{} over structures where every interval can be chopped into at most a constant number of subintervals. \SL\ shares the syntax with \HS\ and \CDT~\cite{chopping_intervals} and can be seen as their decidable variant.

\item \emph{Restricting the set of modal operators.} Complete  classifications of decidable and undecidable fragments of \HS{} have been obtained for finite linear orders (62 decidable fragments), discrete linear orders (44), $\mathbb N$ (47), $\mathbb Z$ (44), and dense linear orders (130). For example, over finite linear orders, there are two maximal decidable fragments with the relations $\after,\afterb,\begins,\beginsb$ and $\after,\afterb, \eds,\edsb$, both of which are non-primitive recursive. Smaller fragments may have lower complexity: for example, the $\begins,\beginsb, \later,\laterb$  fragment is \NP-complete, $\after,\afterb$ is \NExpTime-complete, while $\after,\begins,\beginsb,\laterb$ is \ExpSpace-complete. For more details, we refer the reader
to~\cite{undecidability_BE,abba_finite,DBLP:conf/ecai/BresolinMMSS12,DBLP:journals/corr/abs-1210-2479,lata15} and references therein.

\item \emph{Softening semantics}. Allen~\citeyear{allen83} and  Halpern and Shoham~\citeyear{HalpernS91} defined the semantics of interval relations using the irreflexive $<$: for example, $[x,y]\later [x',y']$ if and only if $y < x'$. By `softening' $<$ to reflexive $\le$ one can make the undecidable $\during$ fragment of \HS~\cite{DBLP:journals/fuin/MarcinkowskiM14} decidable and \PSpace-complete~\cite{DBLP:conf/time/MontanariPS10}.

\item \emph{Relativisations.} The results of Schwentick and Zeume~\citeyear{Schwenticketal10} imply that some undecidable fragments of $\HS$ become decidable if one allows models in which not all the possible intervals of the underlying linear order are present.

\item \emph{Restricting the nesting of modal operators.} Bresolin et al.~\citeyear{light2013} defined a decidable fragment of \CDT\ that mimics the behaviour of the (\NP-complete) Bernays-Sch\"oenfinkel fragment of first-order logic, and one can define a similar fragment of \HS.

\item \emph{Coarsening relations}. Inspired by Golumbic and Shamir's~\citeyear{GolumbicS93} coarser interval algebra, Mu\~{n}oz-Velasco et al.~\citeyear{Munoz-VelascoPS15} reduce the expressive power of \HS\ by defining interval relations that correspond to (relational) unions of Allen's relations. They proposed two coarsening schemata, one of which turned out to be \PSpace-complete.
\end{itemize}

In this article, we analyse a different way of taming the expressive power of logic formalisms while retaining their usefulness for applications, viz., taking Horn fragments. Universal first-order Horn sentences $\forall \avec{x} (A_1 \land \ldots \land A_n \to A_0)$ with atomic $A_i$ are rules (or clauses) of the programming language Prolog. Although Prolog itself is undecidable due to the availability of functional symbols, its function-free subset Datalog, designed for interacting with databases, is \ExpTime-complete for combined complexity, even \PSpace-complete when restricted to predicates of bounded arity, and \PTime-complete in the propositional case~\cite{dan01}. Horn fragments of the Web Ontology Language \OWL~\cite{owl} such as the tractable profiles \OWLQL{} and \OWLEL{} were designed for ontology-based data access via query rewriting and applications that require ontologies with very large numbers of properties and classes (e.g., SNOMED CT). More expressive decidable Horn knowledge representation formalisms have been designed in Description Logic~\cite{DBLP:journals/jar/HustadtMS07,DBLP:journals/tocl/KrotzschRH13}, in particular, temporal description logics; see~\cite{LuWoZa-TIME-08,DBLP:journals/tocl/ArtaleKRZ14} and references therein. Horn fragments of modal and (metric) temporal logics have also been considered~\cite{del1987note,ChenLin93,DBLP:journals/tcs/ChenL94,nguyen2004complexity,artale2013complexity,DBLP:conf/aaai/BrandtKKRXZ17}.

In the context of the Halpern-Shoham logic, we observe first that any \HS-formula can be transformed to an equisatisfiable formula in \emph{clausal normal form}:
\begin{equation}\label{normal}
\varphi \ \  ::= \ \ \lambda \ \ \mid \ \ \neg \lambda  \ \ \mid \ \  \U (\neg \lambda_1 \lor \dots \lor \neg \lambda_n \lor
\lambda_{n+1} \lor \dots \lor \lambda_{n+m})  \ \ \mid \ \  \varphi_1 \land \varphi_2,
\end{equation}
where $\univ$ is the \emph{universal relation} (which can be expressed via the interval relations as
$\U\psi = \bigwedge_{\R}(\psi\wedge[\R]\psi\wedge[\bar{\R}]\psi)$),
and $\lambda$ and the $\lambda_i$ are (\emph{positive temporal}) \emph{literals} given by
\begin{equation}\label{literal}
\lambda \ ::= \ \ \top \ \ \mid \ \ \ \ \bot \ \ \mid \ \ p
\ \ \mid \ \ \langle \R \rangle \lambda \ \ \mid\ \ [\R]\lambda,
\end{equation}
with $\R$ being one of the interval relations and $p$ a \emph{propositional variable} and $[\R]$ the dual of $\langle \R \rangle$.
We now define the \emph{Horn fragment} $\HSh$ of \HS{} as comprising the formulas given by the grammar
\begin{equation}\label{horn}
\varphi \  ::=  \ \ \lambda
\ \ \mid \ \  \U (\lambda_1 \land \dots \land \lambda_{k} \to \lambda)
 \ \ \mid \ \  \varphi_1 \land \varphi_2.
\end{equation}
The conjuncts of the form $\lambda$ are called the \emph{initial conditions} of $\varphi$, and those of the form $\U(\lambda_1 \land \dots \land \lambda_{k} \to \lambda)$ the \emph{clauses} of $\varphi$.
We also consider the $\HShb$ fragment of $\HSh$, whose formulas do not contain occurrences of diamond operators $\langle \R\rangle$, and the $\HShd$ fragment whose formulas do not contain box operators $[\R]$.
We denote by $\HSc$ ($\HScb$ or $\HScd$) the fragment of $\HSh$ (respectively, $\HShb$ or $\HShd$) with only clauses of the form $\U (\lambda_1 \to \lambda_2)$ and $\U (\lambda_1 \land \lambda_2 \to \bot)$. We remind the reader that propositional Horn logic is $\PTime$-complete, while the (core) logic of binary Horn clauses is \NLogSpace-complete.

We illustrate the expressive power of the Horn fragments introduced above by a few examples describing constraints on a summer school timetable. The clause
\begin{equation*}
\U (\DbD \textit{MorningSession} \land  \textit{AdvancedCourse} \to \bot)
\end{equation*}
says that advanced courses cannot be given during the morning sessions defined by
\begin{equation*}
\U (\Dv \textit{LectureDay} \land \AD \textit{Lunch} \leftrightarrow \textit{MorningSession}).
\end{equation*}
The clause
\begin{equation*}
\U (\textit{teaches} \to \DB\textit{teaches})
\end{equation*}
claims that \textit{teaches} is {\em downward hereditary} (or {\em stative}) in the sense that if it holds in some interval, then it also holds in all of its sub-intervals. If, instead, we want to state that \textit{teaches} is {\em upward hereditary} (or {\em coalesced}) in the sense that \textit{teaches} holds in any interval covered by sub-intervals where it holds, then we can use the clause\footnote{Here we assume that the interval relations are reflexive; see Section~\ref{defs}.}
$$
\U \big( \DB (\OD \textit{teaches} \lor
  \DbD \textit{teaches}) \land
  \BD \textit{teaches} \land \Dh
  \textit{teaches} \to \textit{teaches}\big).
$$
By removing the last two conjuncts on the left-hand side of this clause, we make sure that \textit{teaches} is both upward and downward hereditary. For a discussion of these notions in temporal databases, consult~\cite{DBLP:conf/vldb/BohlenSS96,TerenzianiS:tkde04}. Note also that all of the above example clauses---apart from the implication $\leftarrow$ in the second one---are equisatisfiable to $\HShb$-formulas (see Section~\ref{defs} for details).

\begin{table}[t!]
\tbl{Horn and core $\HS$-satisfiability over various linear orders.\label{t:results}}{
\begin{tabular}{|l|c|c|}
\hline
&&\\
& Irreflexive semantics & Reflexive semantics \\[5pt]
\hline\hline
&  \multicolumn{2}{c|}{\ } \\
$\HSh$ & \multicolumn{2}{c|}{undecidable$^*$ (Thm.~\ref{t:undechorn})} \\[5pt]
\hline
&& \\
$\HSc$ & undecidable$^*$ (Thm.~\ref{t:undeccoreirrefl}) & \PSpace-hard$^*$ (Thm.~\ref{t:pspacecore})\\[5pt]
&& \framebox{decidable?}\\[5pt]
\hline
&  \multicolumn{2}{c|}{\ } \\
$\HShd$ & \multicolumn{2}{c|}{undecidable$^*$ (Thm.~\ref{t:undechorn})} \\[5pt]
\hline
&  \multicolumn{2}{c|}{\ } \\
$\HScd$ & \multicolumn{2}{c|}{\framebox{decidable?}} \\[5pt]
\hline
&& \\
& discrete: undecidable (Thm.~\ref{t:undechornboxirrefldisc}) & \\
$\HShb$ & & \PTime-complete (Thm.~\ref{t:inp})\\
& \hspace*{-.5cm}dense: \PTime-complete (Thm.~\ref{t:inp}) &\\[5pt]
\hline
&& \\
& discrete:  \PSpace-hard (Thm.~\ref{t:pspacecoreboxirrefldisc}) & \\[5pt]
$\HScb$ & \hspace*{-.2cm}\framebox{decidable?} & in \PTime\ (Thm.~\ref{t:inp})\\[7pt]
& \hspace*{-1.5cm}dense: in  \PTime\ (Thm.~\ref{t:inp})
&
\\[5pt]
\hline
\end{tabular}}
\hspace{2cm}\footnotesize $^*$actually holds for any class of linear orders containing unbounded orders.
\end{table}

\paragraph{Our contribution} In this article,
we investigate the satisfiability problem for the Horn fragments of \HS{} along two main axes. We consider:
\begin{itemize}
\item both the standard `irreflexive' semantics for \HS-formulas given by~Halpern and Shoham~\citeyear{HalpernS91} and its reflexive variant

\item over classes of discrete and dense linear orders (such as $(\Z,\le)$ and $(\mathbb R,\le)$), and general linear orders.
\end{itemize}
The obtained results are summarised in Table~\ref{t:results}. Most surprising is the computational behaviour of $\HShb$, which turns out to be undecidable over discrete orders under the irreflexive semantics (Theorem~\ref{t:undechornboxirrefldisc}), but becomes tractable under all other choices of semantics  (Theorem~\ref{t:inp}).
The tractability result, coupled with the ability of $\HShb$-formulas to express interesting temporal constraints, suggests that $\HShb$ can form a basis for  tractable interval temporal ontology languages that can be used for ontology-based data access over temporal databases or streamed data. Some preliminary steps in this direction have been made by Artale et al.~\citeyear{artale2015} and Kontchakov et al.~\citeyear{IJCAI16}. We briefly discuss applications of $\HShb$ for temporal ontology-based data access in Section~\ref{sec:app}.

On the other hand, the undecidability of $\HShb$ over discrete orders with the irreflexive semantics prompted us to investigate possible sources of high complexity.
\begin{itemize}
\item
What is the crucial difference between the irreflexive discrete and other semantic choices?
In
discrete models, there is a natural notion of `interval length'\!. With the irreflexive semantics, one can single out intervals of any `fixed' length
using very simple ($\HScb$) formulas: for example, $[\R]\bot$, where $\R$ is either $\eds$ or $\begins$, defines either intervals of length 0 (punctual intervals) or of length 1 (depending on whether one allows punctual intervals or not). Looking at $\HS$-models from the 2D perspective as in Fig.~\ref{f:2DAllen}, we see that intervals of the same fixed length  form a \emph{diagonal\/}.
 Such a `definable' diagonal might provide us with some kind of `horizontal' and `vertical' counting capabilities along the 2D grid, even though the horizontal and vertical `next-time operators' are not available in $\HS$.
It is a well-known fact about 2D \emph{modal product logics} that, if such a `unique controllable diagonal' is expressible in a logic, then the satisfiability problem for the logic is of high complexity~\cite{many_dimensional_modal_logics}. Here we show that $\HShb$ has sufficient counting power to make it undecidable (Theorem~\ref{t:undechornboxirrefldisc}), and that even the  seemingly very limited expressiveness of $\HScb$ is still enough to make it
$\PSpace$-hard (Theorem~ \ref{t:pspacecoreboxirrefldisc}).

\item
When $\D$-operators are available, even if the models are reflexive and/or dense,
one can generate a unique sequence of `diagonal-squares' (like on
a chessboard) and perform some horizontal and vertical counting on   it.
In particular, bimodal logics over products of (reflexive/irreflexive) linear orders \cite{undecidability_compass_logic,Reynolds01122001}
and also over products of various transitive (not necessarily linear) relations
\cite{gkwz05a} are all shown to be undecidable in this way.
It follows that full Boolean $\HS$-satisfiability with the reflexive semantics over any
unbounded timelines is undecidable.
Here we generalise this methodology and show that undecidability still holds even within the
$\HShd$-fragment (Theorem~\ref{t:undechorn}).
\item
We also analyse to what extent the above techniques can be applied within the \emph{core} fragments having $\D$-operators. We develop a few new `tricks'  that encode a certain degree of `Horn-ness'
to prove intractable lower bounds for $\HSc$-satisfiability: undecidability with the irreflexive semantics
(Theorem~\ref{t:undeccoreirrefl}) and $\PSpace$-hardness with the reflexive one (Theorem~\ref{t:pspacecore}).
\end{itemize}
%
The undecidability of $\HSh$ under the irreflexive semantics was established in a  conference paper by Bresolin et al.~\citeyear{DBLP:conf/jelia/BresolinMS14}, and the tractability of $\HShb$ over $(\Z,\le)$ under the reflexive semantics by Artale et al.~\citeyear{artale2015}.


\section{Semantics and notation}\label{defs}

$\HS$-formulas are interpreted over the set of intervals of any  linear order\footnote{Originally, Halpern and Shoham~\citeyear{HalpernS91} also consider more complex temporal structures based on partial orders with {\em linear intervals} such that, whenever $x \le y$, the closed interval $ \{z \in T \mid x \le z \le y\}$ is linearly ordered by $\le$. In particular, trees are temporal structures in this sense.} {$\T = ( T,\leq)$} (where $\le$ is a reflexive, transitive, antisymmetric and connected binary relation on $T$). As usual, we use $x<y$ as a shortcut for `$x\leq y$ and $x\ne y$'\!. The linear order $\T$ is
\begin{itemize}
\item \emph{dense} if, for any $x,y\in T$ with $x< y$, there exists $z$ such that $x< z< y$;

\item \emph{discrete} if every non-maximal $x\in T$ has an immediate $<$-successor, and every non-minimal $x \in T$ has an immediate $<$-predecessor.
\end{itemize}
Thus, the rationals $(\mathbb Q,\le)$ and reals $(\mathbb R,\le)$ are dense orders, while the integers \mbox{$(\mathbb Z,\le)$} and the natural numbers $(\mathbb N,\le)$ are  discrete. Any finite linear order is obviously discrete. We denote by $\Lin$ the class of all linear orders, by $\Fin$ the class of all finite linear orders, by $\Dis$ the class of all discrete linear orders, and by $\Den$ the class of all dense linear orders. We say that a linear order \emph{contains an infinite ascending} (\emph{descending}) \emph{chain} if it has a sequence of points $x_n$, $n < \omega$, such that $x_0<x_1<\dots <x_n<\dots$ (respectively, $x_0>x_1>\dots >x_n>\dots$).
Clearly, any infinite linear order contains an infinite ascending or an infinite descending chain.

Following Halpern and Shoham~\citeyear{HalpernS91}, by an \emph{interval in\/} $\T$  we mean any ordered pair
$\auf x,y\zu$ such that $x\le y$, and denote by $\int(\T)$ the set of all intervals in $\T$. Note that $\int(\T)$ contains all the  \emph{punctual intervals} of the form $\auf x,x\zu$, which is often referred to as the {\em non-strict} semantics. Under the {\em strict} semantics adopted by Allen~\citeyear{allen83}, punctual intervals are disallowed.
All
of our results hold for both semantics, with slight adjustments in the proofs in case of the strict semantics.
We define the interval relations over $\int(\T)$ in the same way as Halpern and Shoham~\citeyear{HalpernS91} by taking (see Fig.~\ref{f:2DAllen}):
\begin{itemize}
\item  $\auf x_1,y_1 \zu \after \auf x_2,y_2\zu$ iff
\footnote{It is to be noted that there exist slightly different versions of $\after$ and $\afterb$ in the
literature. All of our results hold with those versions as well, with slight adjustments in the proofs.}
$y_1=x_2$ and $x_2<y_2$;\hfill (After)

\item  $\auf x_1,y_1 \zu \begins \auf x_2,y_2 \zu$ iff $x_1=x_2$ and $y_2 < y_1$;\hfill (Begins)

\item  $\auf x_1,y_1 \zu \eds \auf x_2,y_2 \zu$ iff $x_1 < x_2$ and $y_1 = y_2$;\hfill (Ends)

\item  $\auf x_1,y_1 \zu \during \auf x_2,y_2 \zu$ iff $x_1 < x_2$ and $y_2 < y_1$;\hfill (During)

\item  $\auf x_1,y_1 \zu \later \auf x_2,y_2 \zu$ iff $y_1 < x_2$;\hfill (Later)

\item  $\auf x_1,y_1 \zu \overlaps \auf x_2,y_2 \zu$ iff $x_1 < x_2 < y_1 < y_2$;\hfill (Overlaps)

\item  $\auf x_1,y_1 \zu \afterb \auf x_2,y_2\zu$ iff  $y_2=x_1$ and $x_2<y_2$;

\item  $\auf x_1,y_1 \zu \beginsb \auf x_2,y_2 \zu$ iff $x_1=x_2$ and $y_1 < y_2$;

\item  $\auf x_1,y_1 \zu \edsb \auf x_2,y_2 \zu$ iff $x_2 < x_1$ and $y_1 = y_2$;

\item  $\auf x_1,y_1 \zu \duringb \auf x_2,y_2 \zu$ iff $x_2 < x_1$ and $y_1 < y_2$;

\item  $\auf x_1,y_1 \zu \laterb \auf x_2,y_2 \zu$ iff $y_2 < x_1$;

\item  $\auf x_1,y_1 \zu \overlapsb \auf x_2,y_2 \zu$ iff $x_2 < x_1 < y_2 < y_1$.

%
\begin{figure}[ht]
\begin{center}
\setlength{\unitlength}{.045cm}
\begin{picture}(120,150)(30,10)
\multiput(65,15)(0,2){68}{\circle*{.1}}
\multiput(95,15)(0,2){68}{\circle*{.1}}

\thicklines
\put(65,145){\line(1,0){30}}
\put(91,147){$i$}

\put(16,129){$i\,\after\,j$}
\put(95,130){\line(1,0){25}}
\put(116,133){$j$}
\put(16,119){$i\,\begins\,j$}
\put(65,120){\line(1,0){25}}
\put(86,123){$j$}
\put(16,109){$i\,\eds\,j$}
\put(70,110){\line(1,0){25}}
\put(91,113){$j$}
\put(16,99){$i\,\during\,j$}
\put(70,100){\line(1,0){15}}
\put(81,103){$j$}
\put(16,89){$i\,\later\,j$}
\put(105,90){\line(1,0){20}}
\put(121,93){$j$}
\put(16,79){$i\,\overlaps\,j$}
\put(75,80){\line(1,0){30}}
\put(101,83){$j$}

\put(16,69){$i\,\afterb\,j$}
\put(35,70){\line(1,0){30}}
\put(61,73){$j$}
\put(16,59){$i\,\beginsb\,j$}
\put(65,60){\line(1,0){37}}
\put(98,63){$j$}
\put(16,49){$i\,\edsb\,j$}
\put(45,50){\line(1,0){50}}
\put(91,53){$j$}
\put(16,39){$i\,\duringb\,j$}
\put(35,40){\line(1,0){75}}
\put(106,43){$j$}
\put(16,29){$i\,\laterb\,j$}
\put(35,30){\line(1,0){20}}
\put(51,33){$j$}
\put(16,19){$i\,\overlapsb\,j$}
\put(35,20){\line(1,0){50}}
\put(81,23){$j$}
\end{picture}
\hspace*{1cm}
\begin{picture}(120,138)(0,-10)
\thinlines
\put(-6,-6){\vector(1,0){136}}
\put(-6,-6){\vector(0,1){136}}
\put(114,-17){$( T,\leq)$}
\put(-12,134){$(T,\leq)$}
\put(1,1){\line(1,1){130}}

\thicklines
\put(1,45){\line(1,0){15}}
\put(25,45){\line(1,0){20}}
\put(1,95){\line(1,0){15}}
\put(25,95){\line(1,0){18}}
\put(45,95){\circle*{2.5}}
\put(47,95){\line(1,0){18}}
\put(73,95){\line(1,0){22}}
\put(45,45){\line(0,1){20}}
\put(45,75.5){\line(0,1){17}}
\put(45,98){\line(0,1){13}}
\put(45,121){\line(0,1){9}}
\put(95,95){\line(0,1){15}}
\put(95,121){\line(0,1){9}}

\put(105,120){$\later$}
\put(42,113){$\beginsb$}
\put(92,113){$\after$}
\put(18,113){$\duringb$}
\put(66,113){$\overlaps$}
\put(18,93){$\edsb$}
\put(66,93){$\eds$}
\put(18,68){$\overlapsb$}
\put(55,80){$\during$}
\put(42,68){$\begins$}
\put(18,43){$\afterb$}
\put(12,30){$\laterb$}
\end{picture}
\end{center}
\caption{The interval relations and their 2D representation.}\label{f:2DAllen}
\end{figure}
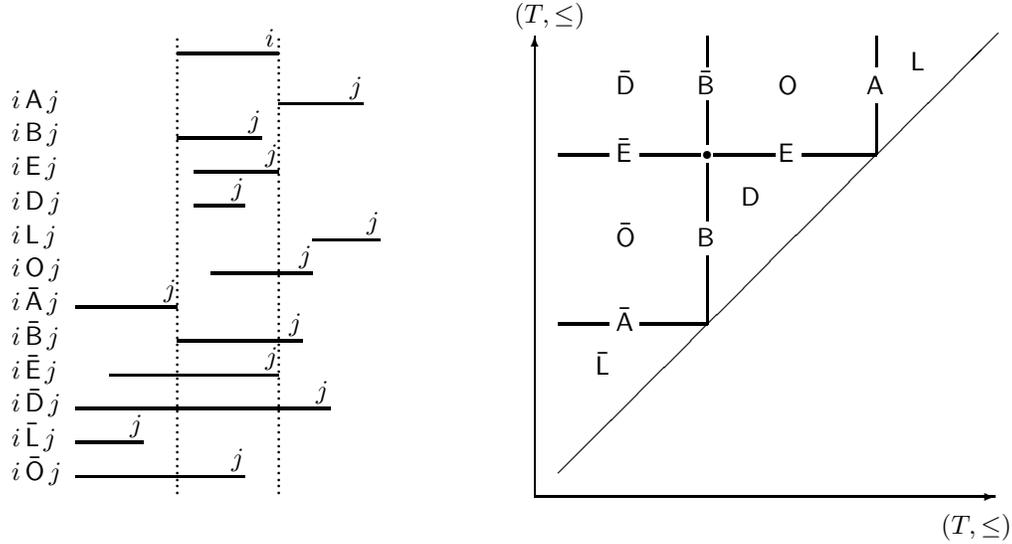

\end{itemize}
%
Observe that all of  these relations are irreflexive, so we refer to the definition above  as the \emph{irreflexive semantics}. As an alternative, we also consider the \emph{reflexive semantics}, which is obtained by replacing each $<$ with $\le$. We write $\T(\le)$ or $\T(<)$ to indicate that the semantics is reflexive or, respectively, irreflexive.
When formulating results where the choice of semantics for each interval relation does not matter, we use the term \emph{arbitrary semantics\/}.\!\footnote{It may be of interest to note that the query language SQL:2011 has seven interval temporal operators three of which are under the reflexive semantics and four under the irreflexive one~\cite{KulkarniM12}.}

As observed by Venema \shortcite{compass_logic}, if we represent intervals $\auf x,y\zu \in \int(\T)$ by points $(x,y)$ of the `north-western' subset of the two-dimensional Cartesian product \mbox{$T\times T$}, then $\int(\mathfrak T)$ together with the interval relations (under any semantics) forms a multimodal Kripke
frame (see Fig.~\ref{f:2DAllen}). We denote it by $\F_{\T}$ and call an $\HS$-\emph{frame\/}.%
\!\footnote{Note that if we consider $\T=(T,\le)$ as a unimodal Kripke frame, then $\bigl(\int(\T),\eds,\beginsb\bigr)$ with the reflexive semantics is an expanding subframe of the \emph{modal product frame} $\T\times\T$; see
\cite[Section~3.9]{many_dimensional_modal_logics}.}
Given a linear order $\T$, an \HS-\emph{model based on\/} $\T$ is a pair $\M=(\F_{\T}, \nu)$, where $\F_{\T}$ is an $\HS$-frame and $\nu$ a function from the set $\mathcal{P}$ of propositional variables to subsets of $\int(\T)$.
%
%
%
The \emph{truth-relation} $\M,\auf x,y\zu \models \varphi$, for an $\HSh$-formula $\varphi$ (read: $\varphi$ \emph{holds at} $\auf x,y\zu$ in $\M$), is defined inductively as follows, where $\R$ is any interval relation:
\begin{itemize}
\item
$\M,\auf x,y\zu\models \top$ and $\M,\auf x,y\zu\not\models \bot$, for any $\auf x,y\zu\in\int(\mathfrak T)$;
\item
$\M,\auf x,y\zu\models p$ iff $\auf x,y\zu\in\nu(p)$, for any $p \in \mathcal{P}$;
\item
$\M,\auf x,y\zu\models\langle \R\rangle \lambda$ iff there exists $\auf x',y'\zu$ such that
$\auf x,y\zu \R \auf x',y'\zu$ and $\M,\auf x',y'\zu\models\lambda$;
\item
$\M,\auf x,y\zu\models[\R] \lambda$ iff, for every $\auf x',y'\zu$ with
$\auf x,y\zu \R \auf x',y'\zu$, we have $\M,\auf x',y'\zu\models\lambda$;
\item
$\M,\auf x,y\zu\models[\univ](\lambda_1 \land \dots \land \lambda_{k} \to \lambda)$ iff,
for every $\auf x',y'\zu\in \int(\T)$ with $\M,\auf x',y'\zu\models\lambda_i$ for $i=1,\ldots,k$, we have $\M,\auf x',y'\zu\models\lambda$;
\item
$\M,\auf x,y\zu\models\varphi_1\land\varphi_2$ iff $\M,\auf x,y\zu\models\varphi_1$ and $\M,\auf x,y\zu\models\varphi_2$.
\end{itemize}
A model $\M$ based on $\T$ \emph{satisfies} $\varphi$ if $\M,\auf x,y\zu\models\varphi$, for some $\auf x,y\zu \in \int(\T)$.  
Given a class $\CC$ of linear orders, we say that a formula $\varphi$ is $\CC$-\emph{satisfiable} (respectively, $\CC(\le)$- or $\CC(<)$-\emph{satisfiable}) if it is satisfiable in an $\HS$-model based on some order from $\CC$ under the arbitrary (respectively, reflexive or  irreflexive) semantics.


To facilitate readability, we use the following \emph{syntactic sugar}, where $\psi = \lambda_1 \land \dots \land \lambda_{k}$:
\begin{itemize}
\item $\U(\psi \to\neg\lambda)$ as an abbreviation for
  $\U(\psi \land\lambda\to\bot)$;

\item $\U \bigl(\psi \to \lambda_1' \land \dots \land \lambda_{n}')$  as an abbreviation for
\[
\bigwedge_{i=1}^{n}\U \bigl(\psi \to \lambda_i');
\]
\item $\U \bigl(\psi \to [\R](\lambda_1' \land \dots \land \lambda_{n}' \to \lambda )\bigr)$ as an abbreviation for
\[
\U (\psi \to [\R] p)\ \land\
\U(p \land \lambda_1' \land \dots \land \lambda_{n}' \to \lambda),
\]
where $p$ is a fresh variable, and similarly for $\langle\R\rangle$ in place of $[\R]$.
\end{itemize}
Note also that $\U(\langle \R\rangle \lambda \land \psi \to \lambda')$ is equivalent to $\U(\lambda \to [\bar\R] (\psi
  \to \lambda'))$. This allows us to use $\langle \R \rangle$ on the left-hand side of the clauses in $\HShb$-formulas, and $[\R]$ on the right-hand side of the clauses in $\HShd$-formulas.



\section{Tractability of $\HShb$}\label{upperb}

Let $\T = (T,\le)$ be a linear order, $\auf a,b \zu \in \int(\T)$, and let $\varphi$ be an $\HShb$-formula. Suppose we want to check whether there exists a model $\M$ based on $\T$ such that $\M,\auf a,b \zu \models \varphi$ under the reflexive (or irreflexive) semantics, in which case we will say that $\varphi$ is \emph{$\auf a,b \zu$-satisfiable in $\T(\le)$} (respectively, $\T(<)$).
Let $\lhd \in \{\le,<\}$. We set
$$
\chase_\varphi =\{ \lambda@\auf a,b \zu \mid \lambda \text{ an initial condition of $\varphi$} \} \cup \{\top@\auf x,y \zu \mid \auf x,y \zu \in \int (\T) \}
$$
and denote by $\cl(\chase_\varphi)$ the result of applying non-recursively the following rules to $\chase_\varphi$, where $\R$ is any interval relation in $\T(\lhd)$:
\begin{itemize}
\item[(cl1)] if $[\R]\lambda @ \auf x,y\zu \in \chase_\varphi$, then we add to $\chase_\varphi$ all $\lambda @\auf x',y'\zu$ such that $\auf x',y'\zu\in \int(\T)$ and $\auf x,y\zu \R \auf x',y'\zu$;

\item[(cl2)] if $\lambda @ \auf x',y'\zu \in \chase_\varphi$ for all $\auf x',y'\zu\in \int(\T)$ such that $\auf x,y\zu \R \auf x',y'\zu$ and $[\R]\lambda$ occurs in $\varphi$, then we add $[\R]\lambda @ \auf x,y\zu$ to $\chase_\varphi$;

\item[(cl3)] if $\U (\lambda_1 \land \dots \land \lambda_{k} \to \lambda)$ is a clause of $\varphi$ and $\lambda_i @ \auf x,y\zu \in \chase_\varphi$, for $1 \le i \le k$, then we add $\lambda @\auf x,y\zu$ to $\chase_\varphi$.
\end{itemize}
Now, we set $\cl^0(\chase_\varphi) = \chase_\varphi$ and, for any successor ordinal $\alpha +1$ and limit ordinal $\beta$,
$$
\cl^{\alpha +1}(\chase_\varphi) = \cl(\cl^\alpha(\chase_\varphi)), \qquad \cl^\beta (\chase_\varphi) = \bigcup_{\alpha < \beta} \cl^{\alpha}(\chase_\varphi)\quad \text{and} \quad \cl^*(\chase_\varphi) = \bigcup_{\gamma \text{ an ordinal}} \cl^\gamma (\chase_\varphi).
$$
Define an \HS-model $\C_\varphi^{\auf a,b\zu} = (\F_{\T}, \nu)$ based on $\T(\lhd)$ by taking, for every variable $p$,
$$
\nu(p) = \{ \auf x,y\zu \mid p @ \auf x,y\zu \in \cl^* (\chase_\varphi)\}.
$$
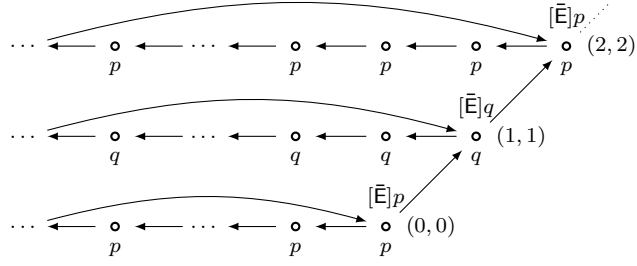
\begin{figure}
\begin{center}
\begin{tikzpicture}[scale=1.2,%
point/.style={draw, thick, circle, inner sep=1, outer
    sep=6},%
    >=latex,
  tornout/.style={very thin, 
  decorate,
    decoration={random steps, amplitude=1pt, segment length=3pt}},%
    segmentpattern/.style={pattern=dots, pattern color = gray}
  ]

\def\shift{.3}
\footnotesize

  \node[point] (p00) at (0,0) {};
  \node[point] (p-10) at (-1,0) {};
  \node (p-20) at (-2,0) {$\dots$};
  \node[point] (p-30) at (-3,0) {};
  \node (p-40) at (-4,0) {$\dots$};
  \node[point] (p11) at (1,1) {};
  \node[point] (p01) at (0,1) {};
  \node[point] (p-11) at (-1,1) {};
  \node (p-21) at (-2,1) {$\dots$};
  \node[point] (p-31) at (-3,1) {};
  \node (p-41) at (-4,1) {$\dots$};
  \node[point] (p22) at (2,2) {};
  \node[point] (p12) at (1,2) {};
  \node[point] (p02) at (0,2) {};
  \node[point] (p-12) at (-1,2) {};
  \node (p-22) at (-2,2) {$\dots$};
  \node[point] (p-32) at (-3,2) {};
  \node (p-42) at (-4,2) {$\dots$};
  \draw[->]  (p00) -- (p11);
  \draw[->] (p11) -- (p22);
  \node[xshift=0.6cm] at (p00) {$(0,0)$};
  \node[yshift=-0.3cm] at (p00) {$p$};
  \node[yshift=0.4cm] at (p00) {$[\edsb] p$};
  \node[yshift=-0.3cm] at (p-10) {$p$};
  \node[yshift=-0.3cm] at (p-30) {$p$};

  \draw[->] (p00) -- (p-10);
  \draw[->] (p-10) -- (p-20);
  \draw[->] (p-20) -- (p-30);
  \draw[->] (p-30) -- (p-40);
  \draw[->] (p-40) to[bend left=15] (p00);


  \node[xshift=0.6cm] at (p11) {$(1,1)$};
  \node[yshift=0.4cm] at (p11) {$[\edsb] q$};
  \node[yshift=-0.3cm] at (p11) {$q$};
  \node[yshift=-0.3cm] at (p01) {$q$};
  \node[yshift=-0.3cm] at (p-11) {$q$};
  \node[yshift=-0.3cm] at (p-31) {$q$};

  \draw[->] (p11) -- (p01);
  \draw[->] (p01) -- (p-11);
  \draw[->] (p-11) -- (p-21);
  \draw[->] (p-21) -- (p-31);
  \draw[->] (p-31) -- (p-41);
  \draw[->] (p-41) to[bend left=15] (p11);


  \node[xshift=0.6cm] at (p22) {$(2,2)$};
  \node[yshift=0.4cm] at (p22) {$[\edsb] p$};
  \node[yshift=-0.3cm] at (p22) {$p$};
  \node[yshift=-0.3cm] at (p12) {$p$};
  \node[yshift=-0.3cm] at (p02) {$p$};
  \node[yshift=-0.3cm] at (p-12) {$p$};
  \node[yshift=-0.3cm] at (p-32) {$p$};

  \draw[->] (p22) -- (p12);
  \draw[->] (p12) -- (p02);
  \draw[->] (p02) -- (p-12);
  \draw[->] (p-12) -- (p-22);
  \draw[->] (p-22) -- (p-32);
  \draw[->] (p-32) -- (p-42);
  \draw[->] (p-42) to[bend left=15] (p22);

 \draw[dotted] (p22) -- ($(p22)+(.5,.5)$);

\end{tikzpicture}
\caption{The sequence of the canonical model construction for $(\mathbb Z,\le)$.}
\label{can-mod-sequence}
\end{center}
\end{figure}
\begin{example}
Let $\T = (\mathbb Z,\le)$. The model $\C_\varphi^{\auf 0,0\zu}$ based on $\T(<)$ for the $\HShb$-formula
\begin{multline*}
\varphi = p \land \U ([\eds]p \land \auf \eds\zu \top \to p) \land \U ([\eds]q \land \auf \eds\zu \top \to q) \land{}\\ \U (\langle \edsb \rangle [\begins][\edsb]p \to q) \land \U (\langle \edsb \rangle [\begins] [\edsb]q \to p)
\end{multline*}
is shown in Fig.~\ref{can-mod-sequence}. Note that the construction of $\C_\varphi^{\auf 0,0\zu}$ requires $\omega^2$ applications of $\cl$.
\end{example}

\begin{theorem}\label{canonical}
An $\HShb$-formula $\varphi$ is $\auf a,b \zu$-satisfiable in $\T(\lhd)$ if and only if $\bot @ \auf x,y\zu \notin \cl^* (\chase_\varphi)$, for any $\auf x,y\zu$.
Furthermore, if some model $\M$ over $\T(\lhd)$ satisfies $\varphi$ at $\auf a,b \zu$, then $\C_\varphi^{\auf a,b \zu}, \auf a,b \zu \models \varphi$ and, for any $\auf x,y\zu\in \int(\T)$ and any variable $p$, $\C_\varphi^{\auf a,b \zu},\auf x,y\zu \models p$ implies $\M,\auf x,y\zu \models p$.
\end{theorem}
\begin{proof}
Suppose $\bot @ \auf x,y\zu \notin \cl^* (\chase_\varphi)$. It is easily shown by induction that we have $\lambda @ \auf x,y\zu \in \cl^*(\chase_\varphi)$ iff $\C_\varphi^{\auf a,b \zu}, \auf x,y\zu \models \lambda$.  It follows that $\C_\varphi^{\auf a,b \zu}, \auf a,b \zu \models \varphi$.
Suppose also that $\M, \auf a,b \zu \models \varphi$, for some model $\M$ over $\T(\lhd)$. Denote by $\chase$ the set of $\lambda @ \auf x,y \zu$ such that $\lambda$ occurs in $\varphi$, $\auf x,y\zu\in \int(\T)$ and $\M, \auf x,y \zu \models \lambda$. Clearly, $\chase$ is closed under the rules for $\cl$, and so  $\cl^*(\chase_\varphi) \subseteq \chase$. This observation also shows that if $\varphi$ is $\auf a,b \zu$-satisfiable in $\T(\lhd)$ then $\bot @ \auf x,y\zu \notin \cl^* (\chase_\varphi)$.
\end{proof}

If $\bot @ \auf x,y\zu \notin \cl^* (\chase_\varphi)$, we call $\C_\varphi^{\auf a,b \zu}$ the \emph{canonical} \emph{model of $\varphi$ based on $\T(\lhd)$}. Our next aim is to show that if (\emph{i}) $\T \in \Dis$ and $\lhd$ is $\le$, or (\emph{ii}) $\T \in \Den$ and $\lhd \in \{\le,<\}$, then there is a bounded-size multi-modal Kripke frame $\mathfrak Z^{\auf a,b \zu}$ with a set of worlds $Z$ and an accessibility relation $R$, for every interval relation $\R$, and a surjective map $f \colon \int(\T) \to Z$ such that the following conditions hold:
\begin{itemize}
\item[(p1)] if $\auf x,y \zu \R \auf x',y' \zu$ then $f(\auf x,y \zu) R f(\auf x',y' \zu)$;

\item[(p2)] if $z R z'$ then, for every $\auf x,y \zu \in f^{-1}(z)$, there is $\auf x',y' \zu \in f^{-1}(z')$ with $\auf x,y \zu \R \auf x',y' \zu$;

\item[(p3)] for any variable $p$ and any $z \in Z$, either $f^{-1}(z) \cap \nu(p) = \emptyset$ or $f^{-1}(z) \subseteq \nu(p)$.
\end{itemize}
In modal logic, a surjection respecting the first two properties is called a \emph{p-morphism} (or \emph{bounded morphism}) \emph{from $\F_\T$ to} $\mathfrak Z^{\auf a,b \zu}$ (see, e.g.,~\cite{Chagrov&Z97,GorankoOtto06}). It is well-known that if $f$ is a p-morphism from $\F_\T$ to $\mathfrak Z^{\auf a,b \zu}$ and $\varphi$ is  $f(\auf a,b \zu)$-satisfiable in $\mathfrak Z^{\auf a,b \zu}$ then $\varphi$ is $\auf a,b \zu$-satisfiable in $\T(\lhd)$. Moreover, if the third condition also holds and $\C_\varphi^{\auf a,b \zu}, \auf a,b \zu \models \varphi$, then $\varphi$ is  $f(\auf a,b \zu)$-satisfiable in $\mathfrak Z^{\auf a,b \zu}$. Indeed, in this case $f$ is a p-morphism from the canonical model $\C_\varphi^{\auf a,b \zu}$ onto the model $(\mathfrak Z^{\auf a,b \zu},\nu')$, where $\nu'(p) = \{z \mid f^{-1}(z) \subseteq \nu(p)\}$.

To construct $\mathfrak Z^{\auf a,b \zu}$ and $f$, we require a few definitions. If $a < b$, we denote by $\mathsf{sec}_{\T}(a,b)$ the set of \emph{non-empty} subsets of $T$ of the form $(-\infty,a)$, $[a,a]$, $(a,b)$, $[b,b]$ and $(b,\infty)$, where $(-\infty,a) = \{x \in T \mid x < a\}$ and $(b,\infty)= \{x \in T \mid x > b\}$. If $a=b$, then $\mathsf{sec}_{\T}(a,b)$ consists of non-empty sets of the form $(-\infty,a)$, $[a,a]$ and $(a,\infty)$. We call each $\sigma \in \mathsf{sec}_{\T}(a,b)$ an $(a,b)$-\emph{section} of $\T$. Clearly,  $\mathsf{sec}_{\T}(a,b)$ is a partition of $T$. Given $\sigma, \sigma' \in \mathsf{sec}_{\T}(a,b)$, we write $\sigma \preceq \sigma'$ if there exist $x\in\sigma$ and $x' \in \sigma'$ such that $\auf x,x' \zu \in \int(\T)$. The definition of $\mathfrak Z^{\auf a,b \zu}$ depends on the type of the linear order $\T$ and the semantics for the interval relations.

\smallskip

{\em Case $\T(\le)$, for $\T \in \Dis \cup \Den$.}
If $\T = (T,\le)$ is a linear order from $\Dis$ or $\Den$ and the semantics is reflexive, then we divide $\int(\T)$ into \emph{zones} of the form
\begin{itemize}
\item  $\zeta_{\sigma,\sigma'} = \{\auf x,x' \zu \in \int(\T) \mid \, x \in \sigma, \, x' \in \sigma' \}$, where $\sigma, \sigma' \in \mathsf{sec}_{\T}(a,b)$ and $\sigma \preceq \sigma'$.
\end{itemize}
For $a < b$ (or $a=b$), there are at most 15 (respectively, at most 6) disjoint non-empty zones covering $\int(\T)$; see Fig.~\ref{zones-ref}. These zones form the set $Z$ of worlds in the frame $\mathfrak Z^{\auf a,b \zu}$, and for any $\zeta,\zeta' \in Z$ and any interval relation $\R$, we set $\zeta R \zeta'$ iff there exist $\auf x,y \zu \in \zeta$ and $\auf x',y' \zu \in \zeta'$ such that $\auf x,y \zu \R \auf x',y' \zu$. Finally, we define a map $f \colon \int(\T) \to Z$ by taking $f(\auf x,y \zu) = \zeta$ iff $\auf x,y \zu \in \zeta$. By definition, $f$ is `onto' and satisfies (p1). Condition (p2) is checked by direct inspection of Fig.~\ref{zones-ref}, while condition (p3) is an immediate consequence of the following lemma:

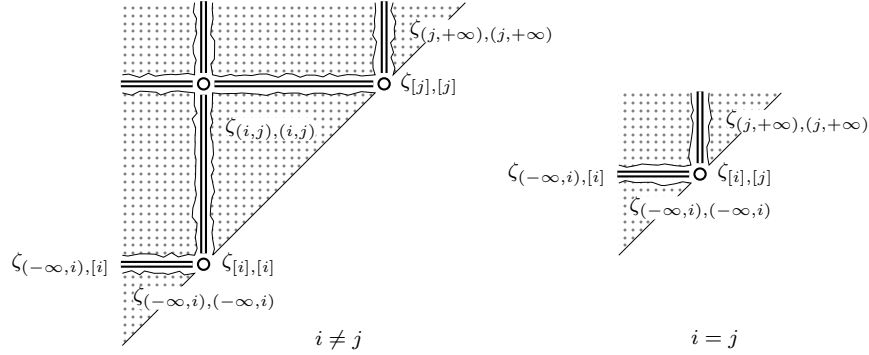
\begin{figure}
\begin{center}
\begin{tikzpicture}[scale=1.2,%
point/.style={draw, thick, circle, inner sep=1.5, outer
    sep=2},%
  tornout/.style={very thin, 
  decorate,
    decoration={random steps, amplitude=1pt, segment length=3pt}},%
    segmentpattern/.style={pattern=dots, pattern color = gray}
  ]

\def\shift{0.1}
\footnotesize

  \node[point] (p00) at (0,0) {};
  \node[point] (p0-1) at (0,-2) {};
  \node[point] (p10) at (2,0) {};

  \node (p-1-2) at (-1,-3) {};
  \node (p-1-1) at (-1,-2) {};
  \node (p-10) at (-1,0) {};
  \node (p-11) at (-1,1) {};
  \node (p01) at (0,1) {};
  \node (p11) at (2,1) {};
  \node (p21) at (3,1) {};

  \draw[thick, double distance=1] (p00) -- (p0-1) (p00) -- (p10);
  \draw[thick, double distance=1] (p0-1) -- (p-1-1) (p00) -- (p-10);
  \draw[thick, double distance=1] (p00) -- (p01) (p10) -- (p11);

  \draw[segmentpattern,tornout] ($(p0-1)+(\shift,\shift)$) --
  ($(p00)+(\shift,-\shift)$) -- ($(p10)-(\shift,\shift)$)
  decorate[decoration={amplitude=0}]{
  -- ($(p0-1)+(\shift,\shift)$)};

  \draw[segmentpattern,tornout] ($(p-1-1)+(\shift,-\shift)$) --
  ($(p0-1)-(\shift,\shift)$) %
  decorate[decoration={amplitude=0}] {-- ($(p-1-2)+(\shift,\shift)$)};

  \draw[segmentpattern,tornout] ($(p-1-1)+(\shift,\shift)$) --
  ($(p0-1)+(-\shift,\shift)$) -- ($(p00)-(\shift,\shift)$) --
  ($(p-10)+(\shift,-\shift)$);

  \draw[segmentpattern, tornout] ($(p-10)+(\shift,\shift)$) --
  ($(p00)+(-\shift,\shift)$) -- ($(p01)-(\shift,\shift)$);

  \draw[segmentpattern, draw=none]
  ($(p01)-(\shift,\shift)$) -- ($(p-11)+(\shift,-\shift)$)--
  ($(p-10)+(\shift,\shift)$);

  \draw[segmentpattern,tornout] ($(p01)+(\shift,-\shift)$) --
  ($(p00)+(\shift,\shift)$) -- ($(p10)+(-\shift,\shift)$) --
  ($(p11)+(-\shift,-\shift)$);


 \draw[segmentpattern,tornout] ($(p21)+(-\shift,-\shift)$) decorate[decoration={amplitude=0}] {--
  ($(p10)+(\shift,\shift)$)} -- ($(p11)+(\shift,-\shift)$);

  \node[xshift=0.6cm] at (p10) {$\zeta_{[j],[j]}$};

  \node[xshift=0.6cm] at (p0-1) {$\zeta_{[i],[i]}$};

  \node[fill=white,inner sep=0, rounded corners, shift={(.9cm,-.6cm)}]
  at (p00)  {$\zeta_{(i,j),(i,j)}$};

  \node[fill=white,inner sep=0, rounded corners, shift={(.1cm,-.5cm)}] at (p21) {$\zeta_{(j, +\infty),(j, +\infty)}$};

  \node[fill=white,inner sep=0, rounded corners, shift={(.15mm,-4.5mm)}] at (p0-1) {$\zeta_{(-\infty, i),(-\infty, i)}$};

  \node[xshift=-0.7cm] at (p-1-1) {$\zeta_{(-\infty,i),[i]}$};


\node[shift={(3,.2)}] at (p-1-2) {$i \neq j$};

  \begin{scope}[shift={(5.5cm, -1cm)}]

\node[point] (p00) at (0,0) {};

  \node (p-1-1) at (-1,-1) {};
  \node (p-10) at (-1,0) {};
  \node (p-11) at (-1,1) {};
  \node (p01) at (0,1) {};
  \node (p11) at (1,1) {};

  \draw[thick, double distance=1] (p00) -- (p-10) (p00) -- (p01);

  \draw[segmentpattern,tornout] ($(p11)+(-\shift,-\shift)$) decorate[decoration={amplitude=0}] {--
  ($(p00)+(\shift,\shift)$)} -- ($(p01)+(\shift,-\shift)$);

  \draw[segmentpattern, tornout] ($(p-10)+(\shift,\shift)$) --
  ($(p00)+(-\shift,\shift)$) -- ($(p01)-(\shift,\shift)$);

  \draw[segmentpattern, draw=none]
  ($(p01)-(\shift,\shift)$) -- ($(p-11)+(\shift,-\shift)$)--
  ($(p-10)+(\shift,\shift)$);

  \draw[segmentpattern,tornout] ($(p-10)+(\shift,-\shift)$) --
  ($(p00)-(\shift,\shift)$) %
  decorate[decoration={amplitude=0}] {-- ($(p-1-1)+(\shift,\shift)$)};


  \node[xshift=0.6cm] at (p00) {$\zeta_{[i],[j]}$};

  \node[fill=white,inner sep=0, rounded corners, shift={(.1cm,-.5cm)}] at (p11) {$\zeta_{(j, +\infty),(j, +\infty)}$};

  \node[fill=white,inner sep=0, rounded corners, shift={(.15mm,-4.5mm)}] at (p00) {$\zeta_{(-\infty, i),(-\infty, i)}$};

  \node[xshift=-0.7cm] at (p-10) {$\zeta_{(-\infty,i),[i]}$};
  \node[shift={(8,.2)}] at (p-1-2) {$i = j$};

  \end{scope}
\end{tikzpicture}
\caption{Zones in the canonical models over $\Dis(\le)$ and $\Den(\le)$.}
\label{zones-ref}
\end{center}
\end{figure}

\begin{lemma}\label{res:mod-monoton-ext1}
For any zone $\zeta$ and any literal $\lambda$ in $\varphi$, if $\C_\varphi^{\auf i, j \zu}, \auf x, y \zu \models \lambda$ for some $\auf x, y \zu \in \zeta$, then $\C_\varphi^{\auf i, j \zu}, \auf x, y \zu \models \lambda$ for all $\auf x, y \zu \in \zeta$.
\end{lemma}
\begin{proof}
It suffices to show that if $\lambda@\auf x,y\zu \in \cl^{\alpha +1}(\chase_\varphi)$ for some $\auf x,y\zu \in \zeta$, then $\lambda@\auf x',y'\zu \in \cl^{\alpha+1}(\chase_\varphi)$ for all $\auf x',y'\zu \in \zeta$, assuming that $\cl^\alpha(\chase_\varphi)$ satisfies this property, which is the case for $\alpha =0$.

Suppose $\auf x,y \zu \in \zeta$ and $\lambda@\auf x,y\zu \in \cl^{\alpha +1}(\chase_\varphi)$ is obtained by an application of  (cl1) to $[\R] \lambda \auf u,v \zu \in \cl^{\alpha}(\chase_\varphi)$ with $\auf u,v \zu R \auf x,y \zu$ and $\auf u,v \zu \in \zeta'$. Take any $\auf x',y' \zu \in \zeta$. By (p2), there is $\auf u',v' \zu \in \zeta'$ such that $\auf u',v' \zu R \auf x',y' \zu$. By our assumption, $[\R] \lambda \auf u',v' \zu \in \cl^{\alpha}(\chase_\varphi)$, and so an application of (c1) to it gives $\lambda@\auf x',y'\zu \in \cl^{\alpha +1}(\chase_\varphi)$.

Suppose next that $\auf x,y \zu \in \zeta$ and $[\R]\lambda@\auf x,y\zu \in \cl^{\alpha +1}(\chase_\varphi)$ is obtained by an application of (cl2). Then $\lambda \auf u,v \zu \in \cl^{\alpha}(\chase_\varphi)$ for all $\auf u,v \zu$ with $\auf x,y \zu \R \auf u,v \zu$. Take any $\auf x',y' \zu \in \zeta$. We show that $\lambda \auf u',v' \zu \in \cl^{\alpha}(\chase_\varphi)$ for every $\lambda \auf u',v' \zu$ with $\auf x',y' \zu \R \auf u',v' \zu$, from which $[\R]\lambda@\auf x',y'\zu \in \cl^{\alpha +1}(\chase_\varphi)$ will follow. Let $\auf u',v' \zu \in \zeta'$. By (p1), $\zeta R \zeta'$ and, by (p2), $\auf x,y\zu \R \auf u,v \zu$ for some $\auf u,v\zu \in \zeta'$ such that $\auf x,y \zu \R \auf u,v\zu$. Then $\lambda \auf u,v \zu \in \cl^{\alpha}(\chase_\varphi)$ and, by our assumption, $\lambda \auf u',v' \zu \in \cl^{\alpha}(\chase_\varphi)$.

The case of rule (cl3) is obvious.
\end{proof}

Note that Lemma~\ref{res:mod-monoton-ext1} does not hold for $\T(<)$. Indeed, we may have punctual intervals $\auf y,y\zu$ (for $y\notin\{a,b\}$) such that $\C_\varphi^{\auf a,b \zu}, \auf y, y \zu \models [\eds] \bot$ but $\C_\varphi^{\auf a,b \zu}, \auf x, y \zu \not\models [\eds] \bot$ for $x < y$, with $\auf x, y \zu$ from the same zone as $\auf y, y \zu$.

\smallskip

{\em Case $\T(<)$, for $\T \in \Den$.}
If $\T$ is a dense linear order and the semantics is irreflexive,
we divide $\int(\T)$ into zones of three types:
\begin{itemize}
\item $\zeta_{\sigma,\sigma'} = \{\auf x,x' \zu \in \int(\T) \mid \, x \in \sigma, \, x' \in \sigma' \}$, where $\sigma, \sigma' \in \mathsf{sec}_{\T}(a,b)$, $\sigma \preceq \sigma'$ and $\sigma \ne \sigma'$;

\item $\zeta_\sigma = \{ \auf x,x' \zu \in \int(\T) \mid  x, x' \in \sigma,\ x \neq x'\}$, where $\sigma \in \mathsf{sec}_{\T}(a,b)$;

\item $\zeta^\bullet_\sigma = \{ \auf x,x \zu \in \int(\T) \mid \, x \in \sigma\}$, where $\sigma \in \mathsf{sec}_{\T}(a,b)$.
\end{itemize}
Now, for $a < b$ (or $a=b$), we have at most 18 (respectively, at most 8) disjoint non-empty zones covering $\int(\T)$; see Fig.~\ref{areas}. It is again easy to see that the map $f \colon \int(\T) \to Z$ defined by taking $f(\auf x,y \zu) = \zeta$ iff $\auf x,y \zu \in \zeta$ satisfies (p1)--(p3). The fact that $\T$ is dense is required for (p2). For discrete $\T$, condition (p2) does not hold. For example, for $\T=(\Z,\lhd)$, $a=0$ and $b=3$, we have $\zeta_{(a,b)}^\bullet \edsb  \zeta_{(a,b), (a,b)}$ but for $\auf 2,2\zu \in \zeta_{(a,b)}^\bullet$ there is no $\auf x',y'\zu \in \zeta_{(a,b), (a,b)}$ such that $\auf 2,2\zu \edsb \auf x',y'\zu$ as shown in the picture below:\\[2pt]
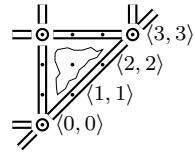
\centerline{
  \begin{tikzpicture}[scale=2,%
point/.style={fill, circle, inner sep=0.5, outer sep=0.5},%
    largepoint/.style={draw, thick, circle, inner sep=1.5, outer sep=2},%
  tornout/.style={very thin, 
  decorate, decoration={random steps, amplitude=1pt, segment length=3pt}},%
    segmentpattern/.style={pattern=dots, pattern color = gray}
  ]
\def\shift{0.1}
\footnotesize
 \node[largepoint] (p00) at (0,0) {};
 \node[point] at (p00) {};
\node[largepoint] (p30) at (.6,0) {};
 \node[point] at (p30) {};
\node[largepoint] (p03) at (0,-0.6) {};
 \node[point] at (p03) {};
\draw[thick, double distance=2] (p00)-- (p30);
\draw[thick, double distance=2] (p00) -- (p03);
\draw[thick, double distance=2] (p30) -- (p03);
\node[point] (p10) at (.2,0) {};
\node[point] (p20) at (.4,0) {};
\node[point] (p01) at (0,-0.2) {};
\node[point] (p11) at (.2,-0.2) {};
\node[point] (p21) at (.4,-0.2) {};
\node[point] (p02) at (0,-0.4) {};
\node[point] (p12) at (.2,-0.4) {};
\node[xshift=.5cm] at (p03) {$\auf 0,0 \zu$};
\node[xshift=.5cm] at (p12) {$\auf 1,1 \zu$};
\node[xshift=.5cm] at (p21) {$\auf 2,2 \zu$};
\node[xshift=.5cm] at (p30) {$\auf 3,3 \zu$};
\draw[tornout] ($(p00)+(\shift,-\shift)$) --
  ($(p30)+(-2*\shift,-\shift)$) %
  -- ($(p03)+(\shift,2*\shift)$) -- cycle;
 \draw[thick, double distance=2] (p00) -- ($(p00)+(0,0.2)$) (p00) -- ($(p00)+(-0.2,0)$);
  \draw[thick, double distance=2] (p30) -- ($(p30)+(0.15,0.15)$) (p30) -- ($(p30)+(0,0.2)$);
   \draw[thick, double distance=1.5] (p03) -- ($(p03)-(0.2,0)$) (p03) -- ($(p03)+(-0.15,-0.15)$);
\end{tikzpicture}
}

\smallskip

\begin{figure}
\begin{center}
 \begin{tikzpicture}[scale=1.2,%
point/.style={draw, thick, circle, inner sep=1.5, outer
    sep=2},%
  tornout/.style={very thin, 
  decorate,
    decoration={random steps, amplitude=1pt, segment length=3pt}},%
    segmentpattern/.style={pattern=dots, pattern color = gray}
  ]

\def\shift{0.1}
\footnotesize

  \node[point] (p00) at (0,0) {};
  \node[point] (p0-1) at (0,-2) {};
  \node[point] (p10) at (2,0) {};

  \node (p-1-2) at (-1,-3) {};
  \node (p-1-1) at (-1,-2) {};
  \node (p-10) at (-1,0) {};
  \node (p-11) at (-1,1) {};
  \node (p01) at (0,1) {};
  \node (p11) at (2,1) {};
  \node (p21) at (3,1) {};

  \draw[thick, double distance=1] (p00) -- (p0-1) (p00) -- (p10);
  \draw[thick, double distance=1] (p0-1) -- (p-1-1) (p00) -- (p-10);
  \draw[thick, double distance=1] (p00) -- (p01) (p10) -- (p11);
  \draw[thick, double distance=1] (p21) -- (p10) -- (p0-1) -- (p-1-2);

  \draw[segmentpattern,tornout] ($(p0-1)+(\shift,2*\shift)$) --
  ($(p00)+(\shift,-\shift)$) -- ($(p10)-(2*\shift,\shift)$)
  -- cycle;

  \draw[segmentpattern,tornout] ($(p-1-1)+(\shift,-\shift)$) --
  ($(p0-1)-(2*\shift,\shift)$) %
  -- ($(p-1-2)+(\shift,2*\shift)$);

  \draw[segmentpattern,tornout] ($(p-1-1)+(\shift,\shift)$) --
  ($(p0-1)+(-\shift,\shift)$) -- ($(p00)-(\shift,\shift)$) --
  ($(p-10)+(\shift,-\shift)$);

  \draw[segmentpattern, tornout] ($(p-10)+(\shift,\shift)$) --
  ($(p00)+(-\shift,\shift)$) -- ($(p01)-(\shift,\shift)$);

  \draw[segmentpattern, draw=none]
  ($(p01)-(\shift,\shift)$) -- ($(p-11)+(\shift,-\shift)$)--
  ($(p-10)+(\shift,\shift)$);

  \draw[segmentpattern,tornout] ($(p01)+(\shift,-\shift)$) --
  ($(p00)+(\shift,\shift)$) -- ($(p10)+(-\shift,\shift)$) --
  ($(p11)+(-\shift,-\shift)$);


 \draw[segmentpattern,tornout] ($(p21)+(-2*\shift,-\shift)$) --
  ($(p10)+(\shift,2*\shift)$) -- ($(p11)+(\shift,-\shift)$);

  \node[xshift=0.6cm] at (p10) {$\zeta_{[j],[j]}$};

  \node[xshift=0.6cm] at (p0-1) {$\zeta_{[i],[i]}$};

  \node[fill=white,inner sep=0, rounded corners, shift={(.9cm,-.6cm)}]
  at (p00)  {$\zeta_{(i,j)}$};

  \node[fill=white,inner sep=0, rounded corners, shift={(.5cm,-.0cm)}] at (p11) {$\zeta_{(j, \infty)}$};

  \node[fill=white,inner sep=0, rounded corners, shift={(-1mm,-4.5mm)}] at (p-1-1) {$\zeta_{(-\infty, i)}$};

  \node[xshift=-0.7cm] at (p-1-1) {$\zeta_{(-\infty,i),[i]}$};

  \node[shift={(.5cm,0cm)}]
  at ($0.5*(p10)+0.5*(p0-1)$)  {$\zeta_{(i,j)}^\bullet$};

  \node[shift={(.6cm,0cm)}] at ($0.5*(p10)+0.5*(p21)$) {$\zeta_{(j, \infty)}^\bullet$};

  \node[shift={(.7cm,0cm)}] at ($0.5*(p-1-2)+0.5*(p0-1)$) {$\zeta_{(-\infty, i)}^\bullet$};


\node[shift={(3,.2)}] at (p-1-2) {$i \neq j$};

  \begin{scope}[shift={(4.5cm, -1cm)}]

  \node[point] (p00) at (0,0) {};

  \node (p-1-1) at (-1,-1) {};
  \node (p-10) at (-1,0) {};
  \node (p-11) at (-1,1) {};
  \node (p01) at (0,1) {};
  \node (p11) at (1,1) {};

  \draw[thick, double distance=1] (p00) -- (p-10) (p00) -- (p01);
  \draw[thick, double distance=1] (p11) -- (p00) -- (p-1-1);

  \draw[segmentpattern,tornout] ($(p11)+(-2*\shift,-\shift)$) --
  ($(p00)+(\shift,2*\shift)$) -- ($(p01)+(\shift,-\shift)$);

  \draw[segmentpattern, tornout] ($(p-10)+(\shift,\shift)$) --
  ($(p00)+(-\shift,\shift)$) -- ($(p01)-(\shift,\shift)$);

  \draw[segmentpattern, draw=none]
  ($(p01)-(\shift,\shift)$) -- ($(p-11)+(\shift,-\shift)$)--
  ($(p-10)+(\shift,\shift)$);

  \draw[segmentpattern,tornout] ($(p-10)+(\shift,-\shift)$) --
  ($(p00)-(2*\shift,\shift)$) -- ($(p-1-1)+(\shift,2*\shift)$);


  \node[xshift=0.6cm] at (p00) {$\zeta_{[i],[j]}$};

  \node[fill=white,inner sep=0, rounded corners, shift={(.5cm,0cm)}] at (p01) {$\zeta_{(j, \infty)}$};

  \node[fill=white,inner sep=0, rounded corners, shift={(-1mm,-4.5mm)}] at (p-10) {$\zeta_{(-\infty, i)}$};

  \node[xshift=-0.7cm] at (p-10) {$\zeta_{(-\infty,i),[i]}$};

  \node[shift={(.6cm,0cm)}]
  at ($0.5*(p11)+0.5*(p00)$)  {$\zeta_{(j, \infty)}^\bullet$};

  \node[shift={(.7cm,0cm)}]
  at ($0.5*(p00)+0.5*(p-1-1)$)  {$\zeta_{(-\infty,i)}^\bullet$};

  \node[shift={(7,.2)}] at (p-1-2) {$i = j$};
  \end{scope}
\end{tikzpicture}
\caption{Zones in the canonical models over $\Den(<)$.}
\label{areas}
\end{center}
\end{figure}
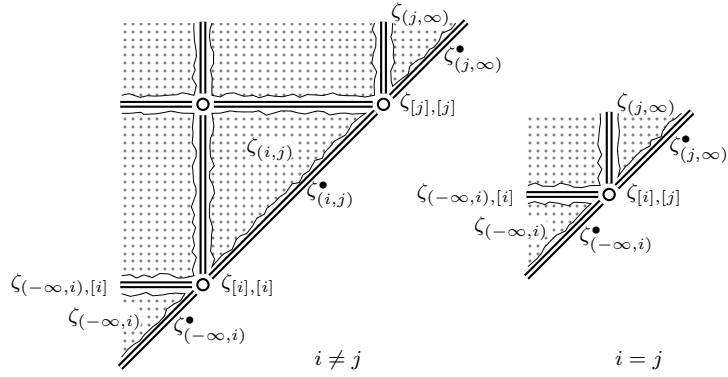



Thus, in both cases the constructed function $f \colon \int(\T) \to Z$ satisfies conditions (p1)--(p3), and so, using Theorem~\ref{canonical}, we obtain:

\begin{theorem}\label{zone-reduction}
Suppose $\T \in \Dis$ and $\lhd$ is $\le$, or $\T \in \Den$ and $\lhd \in \{\le,<\}$. Then an $\HShb$-formula $\varphi$ is $\auf a,b \zu$-satisfiable in $\T(\lhd)$ iff $\varphi$ is $f(\auf a,b \zu)$-satisfiable in $\mathfrak Z^{\auf a,b \zu}$.
\end{theorem}

To check whether $\varphi$ is $f(\auf a,b \zu)$-satisfiable in $\mathfrak Z^{\auf a,b \zu}$, we take the set
$$
\mathfrak U_\varphi =\{ \lambda@f(\auf a,b \zu) \mid \lambda \text{ an initial condition of $\varphi$} \} \cup \{\top@ \zeta \mid \zeta \in Z \}
$$
and apply to it the following obvious modifications of rules (cl1)--(cl3):
\begin{itemize}
\item if $[\R]\lambda @ \zeta \in \mathfrak U_\varphi$, then we add to $\mathfrak U_\varphi$ all $\lambda @\zeta'$ such that $\zeta R \zeta'$;

\item if $\lambda @ \zeta' \in \mathfrak U_\varphi$ for all $\zeta' \in Z$ with $\zeta R \zeta'$ and $[\R]\lambda$ occurs in $\varphi$, then we add $[\R]\lambda @ \zeta$ to $\mathfrak U_\varphi$;

\item if $\U (\lambda_1 \land \dots \land \lambda_{k} \to \lambda)$ occurs in $\varphi$ and $\lambda_i @ \zeta \in \mathfrak U_\varphi$,  $1 \le i \le k$, then add $\lambda @\zeta$ to $\mathfrak U_\varphi$.
\end{itemize}
It is readily seen that at most $|Z| \cdot |\varphi|$ applications are enough to construct a fixed point $\cl^*(\mathfrak U_\varphi)$. Similarly to Theorem~\ref{canonical}, we then show that $\varphi$ is $f(\auf a,b \zu)$-satisfiable in $\mathfrak Z^{\auf a,b \zu}$ iff $\cl^* (\mathfrak U_\varphi)$ does not contain $\bot @ f(\auf a,b \zu)$.

\begin{theorem}\label{t:inp}
Suppose $\Dis' \subseteq \Dis$ and $\Den' \subseteq \Den$ are non-empty.  Then $\Dis'(\le)$-, $\Den'(\le)$- and $\Den'(<)$-satisfiabily of $\HShb$-formulas are all $\PTime$-complete.
\end{theorem}
\begin{proof}
Observe first that, for each of $\Dis'(\le)$, $\Den'(\le)$, $\Den'(<)$, there are at most 8 pairwise non-isomorphic frames of the form $\mathfrak Z^{\auf a,b \zu}$. As we saw above, checking whether $\varphi$ is satisfiable in one of them can be done in polynomial time. It remains to apply Theorem~\ref{zone-reduction}. The matching lower bound holds already for propositional Horn formulas; see, e.g.,~\cite[Theorem 4.2]{dan01} and references therein.
\end{proof}

It is readily seen that, in fact, Theorem~\ref{t:inp} also holds for $\Lin'(\le)$, where $\Lin'$ is any non-empty subclass of $\Lin$.

\subsection{Ontology-based access to temporal data with extensions of $\HShb$}\label{sec:app}

We now briefly discuss how extensions of $\HShb$ can be used to facilitate access to temporal data; for more details and experiments consult~\cite{IJCAI16}.

\paragraph{Querying historical data} Suppose that a non-IT expert user  would like to query the historical data provided by the STOLE\footnote{For STOria LEgislativa della pubblica amministrazione italiana.} ontology that extracts facts about the Italian Public Administration from journal articles~\cite{DBLP:conf/rr/AdorniMPP15}.
The STOLE dataset, $\mathcal{D}$, contains facts about institutions, legal systems, events, and people such as:
\begin{align*}
&  \textit{LegalSystem}(\textit{regno\_di\_sardegna})@[1720,1861],\\
&  \textit{Institution}(\textit{consiglio\_di\_intendenza}) @ [1806,1865].
\end{align*}
The former one, for example, states that Regno di Sardegna was a legal system in the period between 1720 and 1861.
Suppose now that the user is searching for institutions founded during the Regno di Sardegna period. To simplify the user's task, we can create  an ontology, $\mathcal{O}$, with the single clause
$$
\U \forall x\, \big( \textit{Institution}(x) \land \BD \DbD \textit{LegalSystem}(\textit{regno\_di\_sardegna}) \rightarrow {}\textit{RdSInstitution}(x)\big).
$$
The user's query can now be very simple: $\avec{q}(x, t, s) = \textit{RdSInstitution}(x)@[t,s]$. However, the query-answering system has to find  \emph{certain answers} to the \emph{ontology-mediated query} $(\mathcal{O},\avec{q}(x, t, s))$ over $\mathcal{D}$, which are triples $(a,m,n)$ such that $\textit{RdSInstitution}(a)@[m,n]$ holds in all models of $\mathcal{O}$ and $\mathcal{D}$. As shown by Kontchakov et al.~\citeyear{IJCAI16}, this ontology-mediated query can be `rewritten' into a standard datalog query $(\Pi,G(x,t,s))$, where $\Pi$ is a datalog program $\Pi$ and $G(x,t,s)$ a goal, such that the certain answers to $(\mathcal{O},\avec{q}(x, t, s))$ over $\mathcal{D}$ coincide with the answers to $(\Pi,G(x,t,s))$ over $\mathcal{D}$.

The ontology language in this case is a straightforward datalog extension of $\HShb$. However, to represent temporal data, we require more complex initial conditions compared to $\HShb$, namely, facts of the form $P(a_1, \dots, a_l)@[n, m]$, where $\auf n, m \zu$ is an interval. The zonal representation of canonical models above can be extended to this case, but the number of zones will be quadratic in the number of the initial conditions.

\smallskip
We next show an application that requires a \emph{multi-dimensional} version of $\HShb$.

\paragraph{Querying sensor data} Consider a turbine monitoring system that is receiving from sensors a stream of data of the form $\mathit{Blade}(\mathit{id})@(\iota_1,\iota_2)$,
where $\mathit{id}$ is a turbine blade ID and $\iota_2$ is the temperature range over $(\mathbb R,<)$ observed during the time interval $\iota_1$ over $(\mathbb Z,\le)$. Suppose also that the user wants to find the blades and time intervals where the temperature was rising. Thinking of a pair $\avec{\iota}= (\iota_1,\iota_2)$ as a rectangle in the two-dimensional space $(\mathbb Z,\le) \times (\mathbb R, <)$ and using the operators $\langle \R \rangle_\ell$ in dimension $\ell \in \{1,2\}$ coordinate-wise (that is,  $\avec{\iota} \R_\ell \avec{\iota}'$ iff $\iota_\ell \R \iota'_\ell$ and $\iota_i = \iota'_i$, for $i \ne \ell$), we can define rectangles with rising temperature by the clause
\begin{equation*}
\U \forall x \, \big(\AbD_1  \ObD_2 \mathit{BladeTemp}(x) \land  \AD_1 \OD_2 \mathit{BladeTemp}(x) \to \mathit{TempRise}(x) \big)
\end{equation*}
saying that the temperature of a blade $x$ is rising over a rectangle $(\iota_1,\iota_2)$ if $\mathit{BladeTemp}(x)$ holds at some rectangles $(\iota^{\scriptscriptstyle-}_1,\iota^{\scriptscriptstyle-}_2)$ and $(\iota^{\scriptscriptstyle+}_1,\iota^{\scriptscriptstyle+}_2)$ located as shown in Fig.~\ref{f:risingrect}.
%
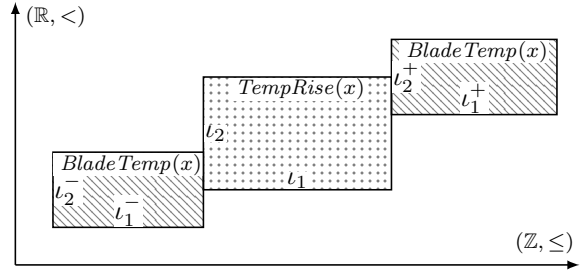
\begin{figure}[ht]
\begin{center}
\begin{tikzpicture}[>=latex,semithick,
    segmentpattern/.style={pattern=north west lines, pattern color = gray}
    ]

  \coordinate (z) at (0.5,0.5) {};

  \draw[->] (z) -- (8, 0.5);
  \draw[->] (z) -- (0.5, 4);

  \node at (7.5,0.8) {\footnotesize $(\mathbb{Z},\leq)$};
  \node at (1,3.8) {\footnotesize $(\mathbb{R},<)$};

  \coordinate (psw) at (1,1) {};
  \coordinate (pse) at (3,1) {};
  \coordinate (pne) at (3,2) {};
  \coordinate (pnw) at (1,2) {};

  \draw[ pattern=north west lines, pattern color = gray] (psw) -- (pse) -- (pne) -- (pnw) -- cycle;
  \path (pse) -- (psw) node[fill=white, inner sep=0.3, above, midway] {$\iota_1^-$};
  \path (psw) -- (pnw) node[fill=white, inner sep=0.3, right, midway] {$\iota_2^-$};
  \path (pnw) -- (pne) node[fill=white, inner sep=0.3, below, pos=0.5] {\footnotesize$\mathit{BladeTemp}(x)$};


 \coordinate (qsw) at (3,1.5) {};
  \coordinate (qse) at (5.5,1.5) {};
  \coordinate (qne) at (5.5,3) {};
  \coordinate (qnw) at (3,3) {};

  \draw[ pattern=dots, pattern color = gray] (qsw) -- (qse) -- (qne) -- (qnw) -- cycle;
  \path (qse) -- (qsw) node[fill=white, inner sep=0.3, above, midway] {$\iota_1$};
  \path (qsw) -- (qnw) node[fill=white, inner sep=0.3, right, midway] {$\iota_2$};
  \path (qnw) -- (qne) node[fill=white, inner sep=0.3, below, pos=0.5] {\footnotesize$\mathit{TempRise}(x)$};


  \coordinate (rsw) at (5.5,2.5) {};
  \coordinate (rse) at (7.7,2.5) {};
  \coordinate (rne) at (7.7,3.5) {};
  \coordinate (rnw) at (5.5,3.5) {};

    \draw[ pattern=north west lines, pattern color = gray] (rsw) -- (rse) -- (rne) -- (rnw) -- cycle;

  \path (rse) -- (rsw) node[fill=white, inner sep=0.3, above, midway] {$\iota_1^+$};
  \path (rsw) -- (rnw) node[fill=white, inner sep=0.3, right, midway] {$\iota_2^+$};
  \path (rnw) -- (rne) node[fill=white, inner sep=0.3, below, pos=0.5] {\footnotesize$\mathit{BladeTemp}(x)$};


\end{tikzpicture}
\caption{Rectangles with rising temperature.}\label{f:risingrect}
\end{center}
\end{figure}

Note that relation algebras over (hyper)rectangles are well-known in temporal and spatial knowledge representation: the rectangle/block algebra $\mathsf{RA}$~\cite{DBLP:journals/logcom/BalbianiCC02} that extends Allen's interval algebra; see also~\cite{DBLP:journals/amai/NavarreteMSV13,DBLP:journals/jair/CohnLLR14,DBLP:conf/kr/ZhangR14} and references therein.
This
multi-dimensional $\HShb$ is capable of expressing rules such as `if $A$ holds at $\avec{\iota}$ and $A'$ at $\avec{\iota}'$, then $B$ holds at the intersection $\avec{\kappa}$ of $\avec{\iota}$ and $\avec{\iota}'$ (or at the smallest rectangle $\avec{\kappa}$ covering $\avec{\iota}$ and $\avec{\iota}'$)' as shown in Fig.~\ref{f:rectops}.
\begin{figure}[ht]
\begin{center}
\begin{tikzpicture}
  \coordinate (psw) at (1,1) {};
  \coordinate (pse) at (3,1) {};
  \coordinate (pne) at (3,2) {};
  \coordinate (pnw) at (1,2) {};

  \coordinate (qsw) at (2,0.5) {};
  \coordinate (qse) at (4,0.5) {};
  \coordinate (qne) at (4,1.5) {};
  \coordinate (qnw) at (2,1.5) {};

  \draw[ pattern=north west lines, pattern color = gray] (psw) -- (psw-|qnw) -- (qnw) -- (qnw-|pne) -- (pne) -- (pnw) -- cycle;

  \path (pse) -- (psw) node[fill=white, inner sep=0.3, above, pos = .95] {$\avec{\iota}$};
  \path (pnw) -- (pne) node[fill=white, inner sep=0.3, below, pos=0.5] {$A$};

  \draw[ pattern=north east lines, pattern color = gray] (qsw) -- (qse) -- (qne) -- (pne|-qnw) -- (pse) -- (pse-|qnw) -- cycle;

  \path (qse) -- (qsw) node[fill=white, inner sep=0.3, above, pos = .92] {$\avec{\iota'}$};
  \path (qnw) -- (qne) node[fill=white, inner sep=0.3, below, pos=0.8] {$A'$};

  \draw[ pattern=dots, pattern color = gray] (psw-|qsw) -- (qnw) -- (qnw-|pse) -- (pse) -- cycle;

  \path (psw-|qsw) -- (pse) node[fill=white, inner sep=0.3, above, pos = .15] {$\avec{\kappa}$};
  \path (qnw) -- (qnw-|pse) node[fill=white, inner sep=0.3, below, pos=0.7] {$B$};

  \begin{scope}[xshift = 5cm]

   \coordinate (psw) at (1,1) {};
  \coordinate (pse) at (3,1) {};
  \coordinate (pne) at (3,2) {};
  \coordinate (pnw) at (1,2) {};

  \coordinate (qsw) at (2,0.5) {};
  \coordinate (qse) at (4,0.5) {};
  \coordinate (qne) at (4,1.5) {};
  \coordinate (qnw) at (2,1.5) {};

  \draw[ pattern=north west lines, pattern color = gray] (psw) -- (pnw) -- (pne) -- (pse) -- cycle;

  \path (pse) -- (psw) node[fill=white, inner sep=0.3, above, pos = .95] {$\avec{\iota}$};
  \path (pnw) -- (pne) node[fill=white, inner sep=0.3, below, pos=0.5] {$A$};

  \draw[ pattern=north east lines, pattern color = gray] (qsw) -- (qnw) -- (qne) -- (qse) -- cycle;

  \path (qse) -- (qsw) node[fill=white, inner sep=0.3, above, pos = .92] {$\avec{\iota'}$};
  \path (qnw) -- (qne) node[fill=white, inner sep=0.3, below, pos=0.8] {$A'$};

  \draw[ pattern=dots, pattern color = gray] (psw|-qsw) -- (psw) -- (psw-|qsw) -- (qsw) -- cycle;

  \path (psw|-qsw) -- (qsw) node[fill=white, inner sep=0.3, above, pos = .2] {$\avec{\kappa}$};

  \draw[pattern=dots, pattern color = gray] (qnw-|pse) -- (pne) -- (pnw-|qne) -- (qne) -- cycle;

  \path (pne) -- (pne-|qne) node[fill=white, inner sep=0.3, below, pos=0.7] {$B$};


  \end{scope}

\end{tikzpicture}
\end{center}
\caption{Expressing simple rules in multi-dimensional $\HShb$.}\label{f:rectops}
\end{figure}

Answering ontology-mediated queries with ontologies in the datalog extension of multi-dimensional $\HShb$ is \PTime-complete for data complexity and can also be done via rewriting into standard datalog queries over the data. The reasonable scalability of this approach was shown experimentally by Kontchakov et al.~\citeyear{IJCAI16} for both one- and two-dimensional cases using standard off-the-shelf datalog tools.


\section{Lower bounds}\label{lb}

In this section, we show that tractability results such as Theorem~\ref{t:inp} are not possible when some kind of `controlled infinity' becomes expressible in the formalism.

\subsection{Methodology}

When
simulating complex problems in $\HS$-models, we always begin by singling out those intervals---call them \emph{units\/}---that are used in the simulation.
It should be clear that if an $\HS$-fragment is capable of
\begin{itemize}
\item[(\emph{i})] forcing an $\omega$-type infinite (or unbounded finite) sequence of units, and

\item[(\emph{ii})] passing polynomial-size information from one unit to the next,
\end{itemize}
then it is $\PSpace$-hard (because polynomial space bounded Turing machine computations can be encoded). It is readily seen that $\HSh$ can easily do both (\emph{i}) and (\emph{ii}).
We show that, in certain situations, Horn clauses can be encoded by means of core clauses, which gives (\emph{i}) and (\emph{ii}) already in the core fragments. In particular, this is the case:
\begin{itemize}
\item  for $\HSc$ over any class  of unbounded timelines under arbitrary semantics (Theorem~\ref{t:pspacecore}), and even
\item  for $\HScb$ over any class of unbounded discrete timelines under the irreflexive semantics (Theorem~\ref{t:pspacecoreboxirrefldisc}).
\end{itemize}
Further, if a fragment is expressive enough to
\begin{itemize}
\item[(\emph{iii})] force an $\omega\times\omega$-like grid-structure of units, and
\item[(\emph{iv})] pass (polynomial-size) information from each unit representing some grid-point to the unit representing its right- and up-neighbours in the grid,
\end{itemize}
then it becomes possible to encode undecidable problems such as $\omega\times\omega$-tilings, Turing
or counter machine computations.
We show this to be the case for the following fragments:
\begin{itemize}
\item
$\HShd$ over any class  of unbounded timelines under arbitrary semantics (Theorem~\ref{t:undechorn}),
\item
$\HSc$ over any class  of unbounded timelines under the irreflexive semantics (Theorem~\ref{t:undeccoreirrefl}), and
\item  $\HShb$ over any class  of unbounded  discrete timelines under the irreflexive semantics (Theorem~\ref{t:undechornboxirrefldisc}).
\end{itemize}
%
%

Although $\HS$-models are always grid-like by definition, it is not straightforward to achieve (\emph{iii})--(\emph{iv}) in them.
Even if we consider the irreflexive semantics and discrete underlying linear orders, $\HS$ does not provide us with horizontal and vertical next-time operators.
The undecidability proofs for (Boolean) $\HS$-satisfiability given by  Halpern and Shoham~\citeyear{HalpernS91} and Marx and Reynolds~\citeyear{undecidability_compass_logic} (for irreflexive semantics), by Reynolds and Zakharyaschev~\citeyear{Reynolds01122001} and Gabbay et al.~\citeyear{many_dimensional_modal_logics}
(for arbitrary semantics), and by Bresolin et al.~\citeyear{lpar08} (for the $\begins\eds$, $\beginsb\eds$ and $\beginsb\edsb$ fragments with irreflexive semantics) all employ the following solution to this problem:
\begin{itemize}
\item[(\emph{v})] Instead of using a grid-like subset of an $\HS$-model as units representing grid-locations, we use some Cantor-style enumeration of either the whole \mbox{$\omega\times\omega$}-grid or its north-western octant $\nwplane$ (see Fig.~\ref{f:enumtwo}), and then force a \emph{unique} infinite (or unbounded finite) sequence of units representing this enumeration (or an unbounded finite prefix of it).

\item[(\emph{vi})] Then we use some `up- and right-pointers' in the model to access the unit representing the grid-location immediately above and to the right of the current one.
\end{itemize}
\begin{figure}[t!]
\setlength{\unitlength}{0.1cm}
\begin{center}
\begin{picture}(60,53)(-5,0)
\put(10,4){\circle*{1}}
\multiput(10,14)(10,0){2}{\circle*{1}}
\multiput(10,24)(10,0){3}{\circle*{1}}
\multiput(10,34)(10,0){4}{\circle*{1}}
\multiput(10,44)(10,0){5}{\circle*{1}}
\thicklines
\put(10,4){\line(0,1){10}}
\put(10,14){\line(1,0){10}}
\put(10,24){\line(1,-1){10}}
\put(10,24){\line(1,0){20}}
\put(10,34){\line(2,-1){20}}
\put(10,34){\line(1,0){30}}
\put(10,44){\line(3,-1){30}}
\put(10,44){\line(1,0){33}}
\put(44.5,43.9){$\ldots$}

\put(11,3){$0$}
\put(11,11){$1$}
\put(21,11){$2$}
\put(10,25){$3$}
\put(20,25){$4$}
\put(30,20){$5$}
\put(10,35){$6$}
\put(20,35){$7$}
\put(30,30){$8$}
\put(40,30){$9$}

\put(1,3){$( 0,0)$}
\put(1,13){$( 0,1)$}
\put(1,23){$( 0,2)$}
\put(1,33){$( 0,3)$}
\put(1,43){$( 0,4)$}

\put(16,46){$( 1,4)$}
\put(26,46){$( 2,4)$}
\put(36,46){$( 3,4)$}

\put(7,51){wall}
\put(9,48){${}_\downarrow$}
\put(50,50){diagonal}
\put(52,46.5){${}_\swarrow$}
\put(-12,13){line$_1\to$}
\put(-12,23){line$_2\to$}
\put(-12,33){line$_3\to$}
\put(-12,43){line$_4\to$}
\end{picture}
\end{center}
\caption{An enumeration of the $\nwplane$-grid.}\label{f:enumtwo}
\end{figure}

\noindent Here, we follow a similar approach. The proofs of Theorems~\ref{t:undechorn}--\ref{t:undechornboxirrefldisc} differ in how (\emph{v}) and (\emph{vi}) are achieved by the capabilities of the different formalisms.
\begin{itemize}
\item In the proof of Theorem~\ref{t:undechorn}, the encoding of the $\omega\times\omega$-grid resembles that of~\cite{undecidability_compass_logic,Reynolds01122001,many_dimensional_modal_logics} for modal products of linear orders, and~\cite{gkwz05a} for modal products of various transitive (not necessarily linear) relations, regardless whether the relations are irreflexive or reflexive.
In particular, in the reflexive semantics the uniqueness constraints in (\emph{v}) are usually not satisfiable, so instead it is forced that all points encoding the same unit behave in the same way.
It turns out that, with some additional `tricks'\!, this technique is applicable to $\HShd$-formulas.

\item It is not clear whether the above method can be applied to the case of $\HSc$. In the proof of Theorem~\ref{t:undeccoreirrefl}, we achieve (for the irreflexive semantics) (\emph{v}) and (\emph{vi}) in a different way, similar to that of~\cite{HalpernS91}.

\item Both techniques above make an essential use of $\auf\R\zu$-operators. In order to achieve (\emph{v}) and (\emph{vi}) using $\HShb$-formulas with the irreflexive semantics and discrete linear orders, in the proof of Theorem~\ref{t:undechornboxirrefldisc} we provide a completely different encoding the $\nwplane$-grid.
\end{itemize}
%


\subsection{Turing machines}\label{tm}

We begin by fixing the notation and terminology regarding Turing machines.
A {\em single-tape right-infinite deterministic Turing Machine} ({\em TM}, for short) is a tuple
$\A=( Q,\Sigma,q_0,q_f,\delta_{\A})$, where $Q$
is a finite set of {\em states} containing, in particular, the
  {\em initial state\/} $q_0$ and the {\em halt state\/} $q_f$,
$\Sigma$ is the {\em tape alphabet\/}  (with a distinguished {\em blank} symbol $\sqcup\in \Sigma$), and
 $\delta_{\A}$ is the {\em transition function\/}, where we use the symbol $\pounds\notin\Sigma$ to mark
the leftmost cell of the tape:
\[
\delta_{\A}:(Q-\{q_f\})\times (\Sigma\cup\{\pounds\})\to Q\times(\Sigma\cup\{
  \tml ,\tmr \}).
\]
The transition function transforms each pair of the form $(q,s)$ into one of the following
pairs:
\begin{itemize}
\item $(q',s')$ (write $s'$ and change the state to $q'$);
\item $(q',\tml)$ (move one cell left and change the state to $q'$);
\item $( q', \tmr)$ (move one cell right and change the state to $q'$),
\end{itemize}
where $\tml$ and $\tmr$ are fresh symbols. We assume that if $s=\pounds$ (i.e.,
the leftmost cell of the tape is active) then
$\delta_{\A}(q,s) = ( q',\tmr)$ (that is, having reached
the leftmost cell, the machine always moves to the right).  We set $\textit{size}(\A)=|Q\cup\Sigma\cup\delta_{\A}|$. {\em Configurations} of $\A$ are infinite
sequences  of the form
$$
C=( s_0, s_1,\ldots, s_i,\ldots, s_n,\sqcup,\ldots),
$$
where either $s_0=\pounds$ and all $s_1,\ldots,s_n$ save one, say $s_i$, are in $\Sigma$, while $s_i$
belongs to $Q\times \Sigma$ and represents the \emph{active cell} and the
\emph{current state}, or $s_0=( q,\pounds)$ for some $q\in Q$ ($s_0$ is the active cell), and all
 $s_1,\ldots,s_n$ are in $\Sigma$. In both cases, all cells of the tape located to the right of $s_n$ contain $\sqcup$.
We assume that the machine always starts with the empty tape (all cells
of which are blank), and so the {\em initial configuration} is represented
by the sequence
\[
C_0=\bigl(( q_0,\pounds),\sqcup,\sqcup,\ldots\bigr).
\]
We denote by $(\CC_n \mid n<H)$ the unique sequence of subsequent configurations
of $\A$ starting with the empty tape ---the unique \emph{computation of} $\A$ \emph{with empty input}---where
\[
H=\left\{
\begin{array}{ll}
n+1, & \mbox{$n$ is the smallest number with $( q_f,s)$ occurring in $\CC_n$ for some $s$},\\
\omega, & \mbox{otherwise}.
\end{array}
\right.
\]
If $H<\omega$, we say that $\A$ \emph{halts with empty input\/}, and call $\CC_{H-1}$ the \emph{halting configuration of} $\A$.
If $H=\omega$, we say that $\A$ \emph{diverges with empty input\/}.
We denote by $\CC_n(m)$ the $m$th symbol in $\CC_n$.

In our lower bound proofs, we use the following Turing machine problems~\cite{Moret1998}:

\medskip
\noindent
\underline{{\sc Halting:}}\ \ ($\Sigma_1^0$-hard)\\[3pt]
Given a Turing machine $\A$, does it halt with empty input?

\medskip
\noindent
\underline{{\sc Non-halting:}}\ \ ($\Pi_1^0$-hard)\\[3pt]
Given a Turing machine $\A$, does it diverge with empty input?

\medskip
\noindent
\underline{{\sc PSpace-bound halting:}}\ \ ($\PSpace$-hard)\\[3pt]
Given a Turing machine $\A$ whose computation with empty input uses at most $\textit{poly}\bigl(\textit{size}(\A)\bigr)$ tape cells for some polynomial function $\textit{poly}()$, does $\A$ halt on empty input?

\medskip
\noindent
\underline{{\sc PSpace-bound non-halting:}}\ \ ($\PSpace$-hard)\\[3pt]
Given a Turing machine $\A$ whose computation with empty input uses at most $\textit{poly}\bigl(\textit{size}(\A)\bigr)$ tape cells for some polynomial function $\textit{poly}()$, does $\A$ diverge on empty input?


\subsection{\PSpace-hardness of core fragments}\label{pspace}

As we have already observed, proving \PSpace-hardness in the case of $\HSh$ is relatively easy. In order to do this in the case
$\HSc$, we use the following \emph{binary implication trick} to capture at least some of the Horn features in $\HSc$.
For any literals $\lambda_1$, $\lambda_2$, and $\lambda$, we define the
formula $\bigl[\lambda_1\land\lambda_2\imph\lambda\bigr]$ as the conjunction of
\begin{align}
\label{vsees}
& \U(\lambda_1\to\AD\mu_1)\land\U(\lambda_2\to\AD\mu_2),\\
\label{lower}
& \U(\mu_2\to\neg\Dv\mu_1),\\
\label{vall}
& \U(\mu_1\to\mu\land\Bv\mu)\land\U(\mu_2\to\BB\mu),\\
\label{hall}
& \U(\AB\mu\to\lambda),
\end{align}
where $\mu_1$, $\mu_2$, and $\mu$ are fresh variables
(the $H$ in $\Rightarrow_H$ stands for `horizontal').
The following claim holds for arbitrary
semantics:

\begin{claim}\label{c:bicomp}
Suppose $\M$ is an $\HS$-model based on some linear order $\T$ and satisfying $\bigl[\lambda_1\land\lambda_2\imph\lambda\bigr]$.
For all $y$ in $\T$, if there exist $x_1,x_2\le y$ such that $\M,\auf x_1,y\zu\models\lambda_1$ and $\M,\auf x_2,y\zu\models\lambda_2$, then
$\M,\auf x,y\zu\models\lambda$ for all $x\le y$.
\end{claim}

\begin{proof}
Suppose $\M,\auf x_1,y\zu\models\lambda_1$ and $\M,\auf x_2,y\zu\models\lambda_2$. Take some $x\le y$.
By \eqref{vsees}, there exist $z_1,z_2\geq y$ such that
$\M,\auf y,z_1\zu\models\mu_1$ and $\M,\auf y,z_2\zu\models\mu_2$. Then $z_1\leq z_2$
by \eqref{lower}. So $\M,\auf y,z\zu\models\mu$ for all $z\geq y$ by \eqref{vall}, and therefore
$\M,\auf x,y\zu\models\lambda$ by \eqref{hall}.
\end{proof}

\emph{Soundness\/}:
Observe that in order to satisfy $\bigl[\lambda_1\land\lambda_2\imph\lambda\bigr]$ the following are necessary:
\begin{itemize}
\item
$\lambda$ is \emph{horizontally stable\/}: for every $y$, we have $\M,\auf x,y\zu\models\lambda$ iff $\M,\auf x',y\zu\models\lambda$ for all $x'$;
\item
if $\M,\auf x',y\zu\not\models\lambda$ (and so $\M,\auf x,y\zu\not\models\lambda$ for all $x$) and $\M,\auf x'',y\zu\models\lambda_1$ for some $x',x''$, then $\M,\auf x,y\zu\not\models\lambda_2$ should hold for all $x$.
\end{itemize}

We use the binary implication trick to prove the following:

\begin{theorem}\label{t:pspacecore}\textbf{\textup{(}$\HSc$, arbitrary semantics\textup{)}}\\
\textup{(}i\textup{)}
For any class $\CC$ of linear orders containing an infinite order,
$\CC$-satisfiability of $\HSc$-formulas is $\PSpace$-hard.
\textup{(}ii\textup{)}
$\Fin$-satisfiability of $\HSc$-formulas is $\PSpace$-hard.
\end{theorem}

\begin{proof}
(\emph{i}) We reduce  {\sc PSpace-bound non-halting} to $\CC$-satisfiability.
%
%
%
%
Let $\A$ be a Turing machine whose computation on empty input uses $<\textit{poly}\bigl(\textit{size}(\A)\bigr)$ tape cells for some polynomial function $\textit{poly}()$, and
let $N=\textit{poly}\bigl(\textit{size}(\A)\bigr)$.
Then we may assume that each configuration $\CC$ of $\A$
is not infinite but of length $N$, and $\A$ never visits the last cell of any configuration.
Let $\GA=\Sigma\cup\{\pounds\}\cup \bigl(Q\times(\Sigma\cup\{\pounds\})\bigr)$.
%
For each $i<N$ and $z\in\GA$, we introduce two propositional variables:
$\cellvar_i^z$ (to encode that `the content of the $i$th cell is $z$') and its `copy' $\overline{\cellvar}_i^z$.

%
%
%
%
Then we can express the uniqueness of cell-contents by
\begin{equation}
\label{celluniq}
\bigwedge_{i<N}\bigwedge_{z\ne z'\in\GA}\U(\cellvar_i^z\to\neg\cellvar_i^{z'}),
\end{equation}
and initialise the computation by
\begin{equation}
\label{cellinit}
\cellvar_0^{(q_0,\pounds)}\land\bigwedge_{0<i<N}\cellvar_i^\sqcup.
\end{equation}
%
%
%
Now we pass information from one configuration to the next, using the `copy' variables and the
`binary implication trick':
\begin{align}
\label{nextunit}
& \U\bigl(\cellvar_i^{(q,s)}\to\AD\overline{\cellvar}_i^{(q,s)}\bigr),\qquad\mbox{for $i<N$, $(q,s)\in (Q-\{q_f\})\times(\Sigma\cup\{\pounds\})$},\\
\nonumber
& \bigl[\cellvar_i^{(q,s)}\land\cellvar_j^z\imph\AD\overline{\cellvar}_j^z\bigr],\\
\label{nextbarvar}
& \hspace*{3cm}
 \mbox{for $i,j<N$, $(q,s)\in (Q-\{q_f\})\times(\Sigma\cup\{\pounds\})$, $z\in \Sigma\cup\{\pounds\}$},\\
 & \U\bigl(\overline{\cellvar}_i^{(q,s)}\to\neg\BD\overline{\cellvar}_j^z\bigr).
\end{align}
We can force that all $\overline{\cellvar}_i^{(q,s)}$-intervals are different (meaning none of them is punctual)
 by the conjunction of, say,
\begin{align}
\label{cellunit}
& \U\bigl(\overline{\cellvar}_i^{(q,s)}\to\plane\bigr),\qquad\mbox{for $i<N$, $(q,s)\in Q\times(\Sigma\cup\{\pounds\})$},\\
\label{unitnotdiag}
& \U (\plane\to\neg\DB\plane).
\end{align}
%
Finally, we can ensure that the information passed in fact encodes the computation steps of $\A$ by the
following formulas. For all $(q,s)\in (Q-\{q_f\})\times(\Sigma\cup\{\pounds\})$ and
$z\in\Sigma\cup\{\pounds\}$,
\begin{itemize}
\item
if $\delta_{\A}(q,s)=(q',s')$, then take the conjunction of
\begin{align}
& \U\bigl(\overline{\cellvar}_i^{(q,s)}\to\cellvar_i^{(q',s')}\bigr),\qquad\mbox{for $i<N$},\\
\label{headstays}
& \bigl[\overline{\cellvar}_i^{(q,s)}\land\Dv\overline{\cellvar}_j^z\imph\cellvar_j^z\bigr],\qquad\mbox{for $i,j<N$, $j\ne i$};
\end{align}

\item
if $\delta_{\A}(q,s)=(q',\tmr)$, then take the conjunction of
\begin{align}
\label{headcell}
& \U\bigl(\overline{\cellvar}_i^{(q,s)}\to\cellvar_i^s\bigr),\qquad\mbox{for $i<N-1$},\\
\label{headright}
& \bigl[\overline{\cellvar}_i^{(q,s)}\land\Dv\overline{\cellvar}_{i+1}^z\imph\cellvar_{i+1}^{(q',z)}\bigr],\qquad\mbox{for $i<N-1$},\\
& \bigl[\overline{\cellvar}_i^{(q,s)}\land\Dv\overline{\cellvar}_j^z\imph\cellvar_j^z\bigr],\qquad\mbox{for $i<N-1$, $j<N$, $j\ne i,\,i+1$};
\end{align}

\item
if $\delta_{\A}(q,s)=(q',\tml)$, then take the conjunction of \eqref{headcell} for $0<i<N$ and
\begin{align}
& \bigl[\overline{\cellvar}_i^{(q,s)}\land\Dv\overline{\cellvar}_{i-1}^z\imph\cellvar_{i-1}^{(q',z)}\bigr],\qquad\mbox{for $0<i<N$},\\
\label{pspacelast}
& \bigl[\overline{\cellvar}_i^{(q,s)}\land\Dv\overline{\cellvar}_j^z\imph\cellvar_j^z\bigr],\qquad\mbox{for $0<i<N$, $j<N$, $j\ne i,\,i-1$}.
\end{align}
\end{itemize}
Finally, we force non-halting with
\begin{equation}
\label{pspacenonhalting}
\U \bigl(\cellvar_i^{(q_f,s)}\to\bot\bigr),\qquad\mbox{for $i<N$, $s\in\Sigma\cup\{\pounds\}$}.
\end{equation}

\begin{claim}
Let $\ftmpspace$ be the conjunction of \eqref {celluniq}--\eqref{pspacenonhalting}. If $\ftmpspace$ is satisfiable in an $\HS$-model, then $\A$ diverges with empty input.
\end{claim}

\begin{proof}
Take any $\HS$-model $\M$ based on a linear order $\T$. Suppose $\M,\auf r,r'\zu\models\ftmpspace$. Then it is not hard to show by induction on $n$ that
there exists an infinite sequence $u_0\le u_1 < u_2 <\dots < u_n<\dots$ of points in $\T$ such that
$u_0=r$, $u_1=r'$, and for all $n<\omega$, the interval $\auf u_n,u_{n+1}\zu$ `represents' the
$n$th configuration $\CC_n$ in the infinite computation of $\A$ with empty input in the following sense:
\[
\M,\auf u_n,u_{n+1}\zu\models\cellvar_i^z\quad\mbox{iff}\quad
\CC_n(i)=z,
\]
for all $i<N$ and $z\in \GA$.
\end{proof}

On the other hand, if $\A$ diverges on empty input, then take some linear order $\T$ containing an infinite ascending chain $t_0< t_1 <\dots$ and define an $\HS$-model
$\M=(\F_{\T},\nu)$ by taking, for all $i<N$ and $z\in \GA$,
\begin{align*}
\nu(\plane) & = \{\auf t_{2n},t_{2n+2}\zu\mid n<\omega\},\\
\nu(\cellvar_i^z) & = \{\auf x,t_{2n+2}\zu\mid \CC_n(i)=z,\ n<\omega,\ x\le t_{2n+2}\},\\
\nu(\overline{\cellvar}_i^z) & = \left\{
\begin{array}{ll}
\{\auf t_{2n+2},t_{2n+5}\zu\mid \CC_n(i)=z,\ n<\omega\}, & \mbox{if $z\in \Sigma\cup\{\pounds\}$},\\[3pt]
\{\auf t_{2n+2},t_{2n+4}\zu\mid \CC_n(i)=z,\ n<\omega\}, & \mbox{if $z\in Q\times(\Sigma\cup\{\pounds\})$}.
\end{array}
\right.
\end{align*}
We
claim that it is possible to evaluate the fresh auxiliary variables in the binary trick formulas
\eqref{nextbarvar}, \eqref{headstays}, \eqref{headright}--\eqref{pspacelast}
so that $\M,\auf t_0,t_2\zu\models\ftmpspace$ with arbitrary semantics.
Indeed, for example, fix some
$(q,s)\in (Q-\{q_f\})\times(\Sigma\cup\{\pounds\})$ with $\delta_{\A}(q,s)=(q',\tmr)$,
$z\in\Sigma\cup\{\pounds\}$, and $i<N-1$,
and consider the corresponding instance of conjunct \eqref{headright}:
\[
\bigl[\overline{\cellvar}_i^{(q,s)}\land\Dv\overline{\cellvar}_{i+1}^z\imph\cellvar_{i+1}^{(q',z)}\bigr].
\]
If we take
\begin{align*}
\nu(\mu_1) & = \{\auf t_{2n+4},t_{2n+6}\zu\mid \CC_n(i)=(q,s),\ n<\omega\},\\
\nu(\mu_2) & = \{\auf x,t_{2n+7}\zu\mid t_{2n+2}\leq x\leq t_{2n+5},\ \CC_n(i+1)=z,\  n<\omega\},
\end{align*}
then it is easy to check that
\[
\U(\overline{\cellvar}_i^{(q,s)}\to\AD\mu_1)\land\U(\Dv\overline{\cellvar}_{i+1}^z\to\AD\mu_2)
\qquad\mbox{and}\qquad
\U(\mu_2\to\neg\Dv\mu_1)
\]
hold in $\M$ (at all points); see Fig.~\ref{f:pspacesound}.
Further, if we take
\begin{multline*}
\nu(\mu) = \{\auf t_{2n+4},y\zu\mid y\geq t_{2n+6},\ \CC_n(i)=(q,s),\ n<\omega\}\ \cup\\
\{\auf x,y\zu\mid t_{2n+2}\leq x\leq t_{2n+5},\ x\leq y\leq t_{2n+7},\ \CC_n(i+1)=z,\  n<\omega\},
\end{multline*}
then it is straightforward to see that  $\U(\mu_1\to\mu\land\Bv\mu)\land\U(\mu_2\to\BB\mu)$
also holds in $\M$.
\begin{figure}[ht]
\begin{center}
 \begin{tikzpicture}[xscale=1.3,yscale=1.3]
\draw[thin,->] (4,2) -- (12,2);
\draw[thick] (5.5,2.1) -- (5.5,1.9) node[below]{$t_{2n+2}$};
\draw[thick] (6.5,2.1) -- (6.5,1.9) node[below]{$t_{2n+3}$};
\draw[thick] (7.5,2.1) -- (7.5,1.9) node[below]{$t_{2n+4}$};
\draw[thick] (8.5,2.1) -- (8.5,1.9) node[below]{$t_{2n+5}$};
\draw[thick] (9.5,2.1) -- (9.5,1.9) node[below]{$t_{2n+6}$};
\draw[thick] (10.5,2.1) -- (10.5,1.9) node[below]{$t_{2n+7}$};
\draw[thin,->] (3,2.5) -- (3,10.5);
\draw[thick] (2.9,4) -- (3.1,4) node[left]{$t_{2n+2}\ \ $};
\draw[thick] (2.9,5) -- (3.1,5) node[left]{$t_{2n+3}\ \ $};
\draw[thick] (2.9,6) -- (3.1,6) node[left]{$t_{2n+4}\ \ $};
\draw[thick] (2.9,7) -- (3.1,7) node[left]{$t_{2n+5}\ \ $};
\draw[thick] (2.9,8) -- (3.1,8) node[left]{$t_{2n+6}\ \ $};
\draw[thick] (2.9,9) -- (3.1,9) node[left]{$t_{2n+7}\ \ $};
\draw[thin] (4.5,3) -- (12,10.5);

 \draw[very thick, double distance=1] (3.5,4) -- (5.5,4);
 \draw (4.4,4.3) node{$\cellvar_i^{(q,s)},\ \cellvar_{i+1}^z$};

\draw[thick, pattern=dots, pattern color = gray] (5.5,9) -- (5.5,4) -- (8.5,7)  -- (8.5,9) -- cycle;
 \draw[thick, dotted] (7.5,9) -- (7.5,10.5);
   \draw (7.7,10) node{$\mu$};

\draw (6.65,8.15) node[rounded corners, inner sep=0.05cm, fill=white] {\Large $\mu$};
 \draw (7.7,7.8) node[rounded corners, inner sep=0.05cm, fill=white]{$\mu_1$};

  \draw[very thick, double distance=1] (3.5,6) -- (7.5,6);
  \draw (4.1,6.3) node{$\cellvar_{i+1}^{(q',z)}$};

  \draw[very thick, double distance=1] (5.5,9) -- (8.5,9);
   \draw (6,9.2) node{$\mu_2$};

  \draw [fill] (5.5,6) circle [radius=0.08];
    \draw (5.1,5.7) node{$\overline{\cellvar}_i^{(q,s)}$};

    \draw [fill] (5.5,7) circle [radius=0.08];
     \draw (5.1,7.3) node{$\overline{\cellvar}_{i+1}^z$};

 \draw [fill] (7.5,8) circle [radius=0.08];

 \draw (10.2,4.5) node{$\delta_{\A}(q,s)=(q',\tmr)$};
 \draw (10,4) node{$\CC_n(i)=(q,s)$};
 \draw (10,3.5) node {$\CC_n(i+1)=z$};
 \draw (10.4,3) node {$\CC_{n+1}(i+1)=(q',z)$};

 \end{tikzpicture}
\caption{Satisfying $\bigl[\overline{\cellvar}_i^{(q,s)}\land\Dv\overline{\cellvar}_{i+1}^z\imph\cellvar_{i+1}^{(q',z)}\bigr]$ in $\M$.}\label{f:pspacesound}
\end{center}
\end{figure}
Finally, we claim that $\U(\AB\mu\to\cellvar_{i+1}^{(q',z)})$
%
%
holds in $\M$ as well.

Indeed, suppose $\M,\auf x,y\zu\models\AB\mu$ for some $x\leq y$.
Then there exist $y_1$, $y_2$ such that $\M,\auf y,y_1\zu\models\mu_1$ and $\M,\auf y,y_2\zu\models\mu_2$.
Thus, there is $n<\omega$ such that $y=t_{2n+4}$ and $\CC_n(i)=(q,s)$,
and  there is $m<\omega$ such that $t_{2m+2}\leq y\leq t_{2m+5}$ and $\CC_m(i+1)=z$.
It follows that either $m=n+1$ or $m=n$. If $m=n+1$ were the case, then both
$\CC_n(i)=(q,s)$ and $\CC_{n+1}(i+1)=z$ would hold, which is not possible when the head moves to the right.
So $m=n$, and we have
$\CC_n(i)=(q,s)$ and $\CC_n(i+1)=z$. Therefore, $\CC_{n+1}(i+1)=(q',z)$, and so
$\M,\auf x,t_{2n+4}\zu\models\cellvar_{i+1}^{(q',z)}$, as required.
Checking the other conjuncts in $\ftmpspace$ is similar and left to the reader.

The case when $\T$ contains an infinite descending chain requires `symmetrical versions' of the used formulas and is also left to the reader.


\smallskip
(\emph{ii})  In the finite case, we reduce  {\sc PSpace-bound halting} to $\Fin$-satisfiability. To achieve this, we just omit
the conjunct \eqref{pspacenonhalting} from $\ftmpspace$. Now,
\eqref{nextunit}
together with the finiteness of the models force the computation to reach the halting state.
\end{proof}

\begin{theorem}\label{t:pspacecoreboxirrefldisc}\textbf{\textup{(}$\HScb$, discrete orders, irreflexive semantics\textup{)}}\\
\textup{(}i\textup{)}
For any class $\Disinf$ of discrete linear orders containing an infinite order, \mbox{$\Disinf(<)$}-satisfiability of $\HScb$-formulas is $\PSpace$-hard.
\textup{(}ii\textup{)}
$\Fin(<)$-satisfiability of $\HScb$-formulas is $\PSpace$-hard.
\end{theorem}

\begin{proof}
(\emph{i}) We again reduce  {\sc PSpace-bound non-halting} to the satisfiability problem.
%
Take any $\HS$-model $\M$ based on a discrete linear order $\T$, and consider the irreflexive semantics of the interval relations.
In this case, we can single out  $\plane$-intervals with the formula
%
\begin{equation}
\label{diaggen}
 \U(\plane\to\Bh\bot)\land \U(\Bh\bot\to\plane).
\end{equation}
It should be clear that if $\M\models$ \eqref{diaggen} then,
for all $\auf x,x'\zu$  in $\int(\T)$, we have $\M,\auf x,x'\zu\models\plane$ iff $x=x'$.
%
Further, it is easy to pass information from one $\plane$-interval to the next, as we have a
`next-time operator w.r.t.' the above  $\plane$-sequence. Namely,
\[
\U (\BB\lambda\to\Bh\lambda')
\]
forces $\lambda'$ to be true at a $\plane$-interval, whenever $\lambda$ is true at the previous one.

To replace the binary implication trick with one using only $\HScb$-formulas,
we employ the following \emph{binary implication trick for the diagonal}.
For any literals $\lambda_1$, $\lambda_2$ and $\lambda$, we define the
formula $\bigl[\lambda_1\land\lambda_2\imphd\lambda\bigr]$ as the conjunction of
\begin{align*}
& \U(\lambda_1\to\AB\mu),\\
& \U(\lambda_2\to\AB\EbB\mu),\\
& \U(\AB\AbB\mu\to\lambda),
\end{align*}
where $\mu$ is a fresh variable.
Then we clearly have the following:

\begin{claim}
Suppose $\M$ satisfies $\bigl[\lambda_1\land\lambda_2\imphd\lambda\bigr]$. If $\M,\auf u_n,u_n\zu\models\lambda_1\land\lambda_2$ then $\M,\auf x,u_{n+1}\zu\models\lambda$ for all $x\le u_{n+1}$.
\end{claim}

\emph{Soundness\/}:
Observe again that to satisfy $\bigl[\lambda_1\land\lambda_2\imphd\lambda\bigr]$ it is necessary that $\lambda$ is horizontally stable in the model.

\smallskip
Now suppose that
$\A$ is a Turing machine whose computation with empty input uses $<\textit{poly}\bigl(\textit{size}(\A)\bigr)$ tape cells for some polynomial function $\textit{poly}()$, and let $\ftmpspaced$
be the conjunction of $\plane$, \eqref{celluniq}, \eqref{cellinit}, \eqref{pspacenonhalting}, \eqref{diaggen}, 
and the following formulas,
for all $(q,s)\in (Q-\{q_f\})\times(\Sigma\cup\{\pounds\})$ and $z\in\Sigma\cup\{\pounds\}$:
%
\[
\U\bigl(\cellvar_i^{(q,s)}\to\neg\Bv\bot\bigr),\qquad\mbox{for $i<N$,}
\]
and
\begin{itemize}
\item
if $\delta_{\A}(q,s)=(q',s')$, then
%
\begin{align}
\nonumber
& \U\bigl(\BB\cellvar_i^{(q,s)}\to\Bh\cellvar_i^{(q',s')}\bigr),\qquad\mbox{for $i<N$},\\
\label{headstaysdiscr}
& \bigl[\cellvar_i^{(q,s)}\land\cellvar_j^z\imphd\cellvar_j^z\bigr],\qquad\mbox{for $i,j<N$, $j\ne i$};
\end{align}

\item
if $\delta_{\A}(q,s)=(q',\tmr)$, then
%
\begin{align}
\label{headcelld}
& \U\bigl(\BB\cellvar_i^{(q,s)}\to\Bh\cellvar_i^s\bigr),\qquad\mbox{for $i<N-1$},\\
\nonumber
& \bigl[\cellvar_i^{(q,s)}\land\cellvar_{i+1}^z\imphd\cellvar_{i+1}^{(q',z)}\bigr],\qquad\mbox{for $i<N-1$},\\
\nonumber
& \bigl[\cellvar_i^{(q,s)}\land\cellvar_j^z\imphd\cellvar_j^z\bigr],\qquad\mbox{for $i<N-1$, $j<N$, $j\ne i,\,i+1$};
\end{align}

\item
if $\delta_{\A}(q,s)=(q',\tml)$, then
\eqref{headcelld} for $0<i<N$ and
\begin{align*}
& \bigl[\cellvar_i^{(q,s)}\land\cellvar_{i-1}^z\imphd\cellvar_{i-1}^{(q',z)}\bigr],\qquad\mbox{for $0<i<N$},\\
& \bigl[\cellvar_i^{(q,s)}\land\cellvar_j^z\imphd\cellvar_j^z\bigr],\qquad\mbox{for $0<i<N$, $j<N$, $j\ne i,\,i-1$}.
\end{align*}
\end{itemize}
%

\begin{claim}
If $\ftmpspaced$ is satisfiable in
an $\HS$-model based on a discrete linear order, then $\A$ diverges with empty input.
\end{claim}

\begin{proof}
Take any $\HS$-model $\M$ based on a discrete linear order $\T$, and suppose $\M,\auf r,r'\zu\models\ftmpspaced$ with the irreflexive semantics. Then it is not hard to show by induction on $n$ that
there exists an infinite sequence $r=r'=u_0< u_1 < u_2 <\dots < u_n<\dots$ of subsequent
points in $\T$ such that for all $n<\omega$, the interval $\auf u_{n},u_{n}\zu$ `represents' the
$n$th configuration $\CC_n$ in the infinite computation of $\A$ with empty input in the following sense:
\[
\M,\auf u_{n},u_{n}\zu\models\cellvar_i^z\quad\mbox{iff}\quad
\CC_n(i)=z,
\]
for all $i<N$ and $z\in \GA$.
\end{proof}

On the other hand, if $\A$ diverges on empty input, then take some discrete linear order $\T$ containing an infinite ascending chain $t_0< t_1 <\dots$ of subsequent points. Define an $\HS$-model
$\M=(\F_{\T},\nu)$ by taking, for all $i<N$ and $z\in \GA$,
\begin{align*}
\nu(\plane) & = \{\auf x,x\zu\mid x\mbox{ in }\T\},\\
\nu(\cellvar_i^z) & = \{\auf x,t_{n}\zu\mid x\le t_{n},\ \CC_n(i)=z,\ n<\omega\}.
\end{align*}
%
%
%
%
%
%
We claim that it is possible to evaluate the fresh auxiliary variable in each binary trick formula
so that $\M,\auf t_0,t_0\zu\models\ftmpspaced$ with the irreflexive semantics.
Indeed, for example, fix some
$(q,s)\in (Q-\{q_f\})\times(\Sigma\cup\{\pounds\})$ with $\delta_{\A}(q,s)=(q',s')$,
$z\in\Sigma\cup\{\pounds\}$, and $i,j<N$, $j\ne i$,
and consider the corresponding instance of conjunct \eqref{headstaysdiscr}:
\[
\bigl[\cellvar_i^{(q,s)}\land\cellvar_j^z\imphd\cellvar_j^z\bigr].
\]
Take
\[
\nu(\mu)= \{\auf x,y\zu\mid x\leq t_{n},\ y\geq t_{n+1},\ \CC_n(i)=(q,s),\ \CC_n(j)=z,\ n<\omega\}.
\]
Then it is straightforward to see that $\U(\cellvar_i^{(q,s)}\to\AB\mu)$ and
$\U(\cellvar_j^z\to\AB\EbB\mu)$ both hold in $\M$. We claim that $\U(\AB\AbB\mu\to\cellvar_j^z)$
%
%
holds in $\M$ as well. Indeed, suppose $\M,\auf x,y\zu\models\AB\AbB\mu$ for some $x\leq y$.
Then $y=t_{n+1}$ for some $n<\omega$ such that $\CC_n(i)=(q,s)$ and $\CC_n(j)=z$.
Thus, $\CC_{n+1}(j)=z$, and so $\M,\auf x,t_{n+1}\zu\models\cellvar_j^z$, as required.
Checking the other conjuncts in $\ftmpspaced$ is similar and left to the reader.

The case when $\T$ contains an infinite descending chain of immediate predecessor points requires `symmetrical versions' of the used formulas and is also left to the reader.

(\emph{ii}) We reduce  {\sc PSpace-bound halting} to $\Fin(<)$-satisfiability. To achieve this,
we omit the conjunct \eqref{pspacenonhalting} from $\ftmpspaced$ above
in order to force the computation to reach the halting state.
\end{proof}


\subsection{Undecidability}\label{undec}

In our undecidability proofs, we `represent' Turing machine computations
on the $\nwplane$-grid as follows.
Given any Turing machine $\A$, observe that for any computation of $\A$
in the $n$th step the head can never move further than the $n$th cell. If $\A$ starts with empty input, this
means that $\CC_n(m)=\sqcup$ for all $n<H$ and $n<m<\omega$. Because of this
we may actually assume that $\CC_n$ is not of infinite length but of \emph{finite} length $n+2$.
(Thus, $\CC_0=\bigl((q_0,\pounds),\sqcup\bigr)$ and $\A$ never
visits the last cell of any $\CC_n$, so it is always $\sqcup$.)
%
So we can place the subsequent finite configurations of the computation
on the subsequent horizontal \emph{lines} of the $\nwplane$-grid,
continuously one after another (until we reach $\mathcal{C}_{H-2}$, if $H<\omega$),
as depicted in Fig.~\ref{f:lines}.
\begin{figure}[t]
\setlength{\unitlength}{0.1cm}
\begin{center}
\begin{picture}(60,43)(-18,4)
\put(10,4){\circle*{1}}
\multiput(10,14)(10,0){2}{\circle*{1}}
\multiput(10,24)(10,0){3}{\circle*{1}}
\multiput(10,34)(10,0){4}{\circle*{1}}
\multiput(10,44)(10,0){5}{\circle*{1}}
\thicklines
\put(10,4){\line(0,1){10}}
\put(10,14){\line(1,0){10}}
\put(10,24){\line(1,-1){10}}
\put(10,24){\line(1,0){20}}
\put(10,34){\line(2,-1){20}}
\put(10,34){\line(1,0){30}}
\put(10,44){\line(3,-1){30}}
\put(10,44){\line(1,0){33}}
\put(44.5,43.9){$\ldots$}

\put(-12,13){{\small $\CC_0$ on line$_1\to$}}
\put(-12,23){{\small $\CC_1$ on line$_2\to$}}
\put(-12,33){{\small $\CC_2$ on line$_3\to$}}
\put(-12,43){{\small $\CC_3$ on line$_4\to$}}
\end{picture}
\end{center}
\caption{Placing the computation of $\A$ on the
$\nwplane$-grid.}\label{f:lines}
\label{F1}
\end{figure}

Observe also that only the active cell and its neighbours can be changed by the transition to the next
configuration, while all other cells remain the same. So instead of using the transition function $\delta_{\A}$,
we can have the same information in the form of a `triples to cells' function $\tau_{\A}$ defined as follows.
Let $\GA=\Sigma\cup\{\pounds\}\cup\bigl(Q\times(\Sigma\cup\{\pounds\})\bigr)$ and
let $\WA\subseteq \GA\times\GA\times\GA$ consist of those triples that can occur as three
subsequent cells in the continuous enumeration of the configurations of the computation, that is, let
\begin{align*}
\WA= \; & \bigl((Q^-\times\Sigma)\times\Sigma\times\Sigma\bigr)\cup
  \bigl(\Sigma\times(Q^-\times\Sigma)\times\Sigma\bigr)\cup
   \bigl(\Sigma\times\Sigma\times(Q^-\times\Sigma)\bigr)\,\cup\\
   & \ \bigl(\LEnd\times\Sigma\times\Sigma\bigr)\cup
  \bigl(\{\sqcup\}\times\LEnd\times\Sigma\bigr)\cup
   \bigl(\Sigma\times\{\sqcup\}\times\LEnd\bigr)\cup
   \bigl\{\bigl(\auf q_0,\pounds\zu,\sqcup,\pounds\bigr)\bigr\},
\end{align*}
where $Q^-=Q-\{q_f\}$ and $\LEnd=\{\pounds\}\cup\bigl(Q^-\times\{\pounds\}\bigr)$.
We define a function
$\tau_{\A}:\WA\to \GA$ by taking, for all $(x,y,z)\in \WA$,
\begin{equation}\label{tmfv}
\tau_{\A}(x,y,z)=\left\{
\begin{array}{ll}
( q',y), & \mbox{if either $x\in (Q-\{q_f\})\times(\Sigma\cup\{\pounds\})$ and $\delta_{\A}(x)=(q',\tmr)$},\\
& \hspace*{.4cm} \mbox{or $z\in (Q-\{q_f\})\times\Sigma$ and $\delta_\mathcal{A}(z)=(q',\tml)$},\\
( q',y'), & \mbox{if $y\in (Q-\{q_f\})\times\Sigma$ and $\delta_{\A}(y)=( q',y')$},\\
y' & \mbox{if $y=( q,y')$ and $\delta_{\A}(y)=( q',{\sf M})$ for ${\sf M}=\tml,\tmr$},\\
y, & \mbox{otherwise}.
\end{array}
\right.
\end{equation}
Then it is easy to see that $\tau_{\A}$ indeed determines the computation of $\A$, that is,
for all $0<n<H$, $\CC_n(n+1)=\sqcup$ and for all $m\leq n$,
\[
\CC_n(m)=\left\{
\begin{array}{ll}
\tau_\mathcal{A}\bigl(\sqcup,\CC_{n-1}(0),\CC_{n-1}(1)\bigr),
& \mbox{if $m=0$},\\[3pt]
\tau_\mathcal{A}\bigl(\CC_{n-1}(m-1),\CC_{n-1}(m),\CC_{n-1}(m+1)\bigr),
& \mbox{if $0<m< n$},\\[3pt]
\tau_\mathcal{A}\bigl(\CC_{n-1}(n-1),\sqcup,\CC_{n}(0)\bigr),
& \mbox{if $m=n$}.
\end{array}
\right.
\]

%

\begin{theorem}\label{t:undechorn}\textbf{\textup{(}$\HShd$, arbitrary semantics\textup{)}}\\
\textup{(}i\textup{)}
For any class $\CC$ of linear orders containing an infinite order,
$\CC$-satisfiability of $\HShd$-formulas is undecidable.
\textup{(}ii\textup{)}
$\Fin$-satisfiability of $\HShd$-formulas is undecidable.
\end{theorem}

\begin{proof}
(\emph{i}) We reduce  {\sc non-halting} to $\CC$-satisfiability.
We discuss only the case when $\CC$ contains some linear order $\T$ having an infinite ascending chain.
(The case when $\T$ contains an infinite descending chain requires `symmetrical versions' of the used formulas and it is left to the reader.)
%

To make the main ideas more transparent, first we assume the \emph{irreflexive semantics}
for the interval relations, and then we show how to modify the proof for \emph{arbitrary}
semantics.
Take any $\HS$-model $\M$ based on some linear order $\T$.
We begin with forcing a \emph{unique} infinite $\plane$-sequence in $\M$, using the conjunction of
\eqref{unitnotdiag} and
\begin{align}
\label{diaggenwithA}
& \plane\land\U(\plane\to\AD\plane),\\
\label{diaguniq}
& \U(\plane\to\neg\Dh\plane\land\neg\Dv\plane\land\neg\DD\plane\land\neg\OD\plane).
\end{align}
Then it is straightforward to show the following:

\begin{claim}\label{c:grid}
Let $\fdiag$ be the conjunction of \eqref{unitnotdiag}, \eqref{diaggenwithA} and \eqref{diaguniq}, and
suppose that $\M,\auf r,r'\zu\models\fdiag$. Then there is an infinite sequence
$u_0<u_1<\ldots<u_n<\ldots$
of points in $\T$ such that for all $r\leq x$ and all $r'\leq x'$,
we have $\M,\auf x,x'\zu\models\plane$ iff $x=u_n$ and $x'=u_{n+1}$ for some $n<\omega$.
%
%
\end{claim}

Next, we use this $\plane$-sequence to encode the enumeration of the $\nwplane$-grid depicted in Fig.~\ref{f:enumtwo}.
Observe that for this particular enumeration the right-neighbour of a grid-location is the next one in the enumeration. As we generated our $\plane$-sequence with \eqref{diaggenwithA}, we have access from one $\plane$-interval
to the next by the $\after$ interval relation.
So, to encode the $\nwplane$-grid, it is enough to use  `up-pointers'.
We force the proper placement of `up-pointers' in a particular way, by using
the following properties of this enumeration:
\begin{enumerate}
\item[(a.1)]
$0$ is on the diagonal, and $\mathsf{up\_neighbour\_of}(0)=1$.
\item[(a.2)]
If $n$ is on the diagonal, then $\mathsf{up\_neighbour\_of}(n)+1$ is on the diagonal, for every $n<B$.
\item[(a.3)]
If $n$ is the up-neighbour of some location, then $n$ is not on the diagonal, for every $n<B$.
\item[(a.4)]
If $n$ is not on the diagonal, then $\mathsf{up\_neighbour\_of}(n+1)=\mathsf{up\_neighbour\_of}(n)+1$,
for every $n+1<B$.
\item[(a.5)]
If $n$ is on the diagonal, then $\mathsf{up\_neighbour\_of}(n+1)=\mathsf{up\_neighbour\_of}(n)+2$,
for every $n+1<B$.
\end{enumerate}
\begin{claim}
Properties {\rm (a.1)--(a.5)} uniquely determine%
\footnote{Among those that contain the enumeration of the diagonal locations as $( 0,0),\ldots,( 1,1),\ldots,( 2,2),\ldots$}
%
the enumeration in Fig.~\ref{f:enumtwo}.
\end{claim}

\begin{proof}
We prove by induction on $n<B$ that for every $k\leq n$,
\begin{itemize}
\item[(\emph{i})]
$k=\langle x,y)$ is like it should be in Fig.~\ref{f:enumtwo}.
\item[(\emph{ii})]
$k$ is on the diagonal iff $k=( x,x)$ for some $x$.
\end{itemize}
Indeed, for $n=1$ (\emph{i}) follows from (a.1), and (b) follows from (a.3). Now suppose inductively that
(\emph{i})--(\emph{ii}) hold for all $k\leq n$ for some $0<n<B$, and let $n+1<B$. There are three cases.

If $n$ is on the diagonal, then by (ii),
$n=( x,x)$ for some $x>0$. Let \mbox{$m=( x-1,x-1)$.} Then $m<n$ by (i) and so by (ii),
$m$ is on the diagonal. So by (a.5), $n+1=\mathsf{up\_neighbour\_of}(m+1)$, proving (\emph{i}).
Now (\emph{ii}) follows from (a.3).

If $n$ is not on the diagonal and $n=( x,y)$ for some $y$ and $x<y-1$, then let $m=( x,y-1)$.
Then $m<n$ by (\emph{i}) and so by (\emph{ii}), $m$ is not on the diagonal.
So by (a.4), $n+1=\mathsf{up\_neighbour\_of}(m+1)$, proving (\emph{i}). Now (\emph{ii}) follows from (a.3).

If $n$ is not on the diagonal and $n=( y-1,y)$ for some $y$, then let $m=( y-1,y-1)$.
Then $m<n$ by (i) and so by (ii), $m$ is on the diagonal.
By (a.2), $n+1$ is on the diagonal, so it should be the next `unused' diagonal location, which is
$( y,y)$, proving both (\emph{i}) and (\emph{ii}).
\end{proof}

Next,
given a unique infinite $\plane$-sequence $\mathcal{U}=\bigl(\auf u_n,u_{n+1}\zu \mid n<\omega\bigr)$
as in Claim~\ref{c:grid} above,
we express `horizontal' and `vertical next-time' in $\M$ `\emph{with respect to}\ $\mathcal{U}$'. Given
literals $\lambda_1$ and $\lambda_2$,
let $\succh[\lambda_1,\lambda_2]$ denote the conjunction of
\begin{align}
\label{uniqhnext}
& \U(\lambda_1\to\neg\Dh\lambda_1)\land \U(\lambda_2\to\neg\Dh\lambda_2),\\
\nonumber
& \U (\lambda_1\to\Dh\lambda_2),\\
\nonumber
& \U \bigl(\lambda_1\to\Bh(\Dh\lambda_2\to\neg\BD\plane)\bigr),
\end{align}
and similarly,
let $\succv[\lambda_1,\lambda_2]$ denote the conjunction of
\begin{align}
\label{uniqvnext}
& \U(\lambda_1\to\neg\Dv\lambda_1)\land \U(\lambda_2\to\neg\Dv\lambda_2),\\
\nonumber
& \U (\lambda_1\to\Dv\lambda_2),\\
\nonumber
& \U \bigl(\lambda_1\to\Bv(\Dv\lambda_2\to\neg\Dh\plane)\bigr).
\end{align}

It is straightforward to show the following:

\begin{claim}
Suppose $\M,\auf u_m,u_n\zu\models\lambda_1$ for some $m,n<\omega$.
\begin{itemize}
\item
Suppose $\M$ satisfies $\succh[\lambda_1,\lambda_2]$. Then, for all $x$,
 $\M,\auf x,u_n\zu\models\lambda_2$ iff $x=u_{m+1}$, and
 $\M,\auf x,u_n\zu\models\lambda_1$ iff $x=u_{m}$.
 \item
Suppose $\M$ satisfies $\succv[\lambda_1,\lambda_2]$. Then, for all $y$,
 $\M,\auf u_m,y\zu\models\lambda_2$ iff $y=u_{n+1}$,
 and
 $\M,\auf u_m,y\zu\models\lambda_1$ iff $y=u_{n}$.
 \end{itemize}
\end{claim}

Now we can encode (a.1)--(a.5) as follows.
We use a propositional variable $\uvar$ to mark up-pointers,
 variables $\diag$  and $\notdiag$ to mark those respective $\plane$-points that are on the diagonal and not on
 the diagonal,
and further fresh variables $\now$, $\uvaru$, $\uvarr$, $\uvarn$
(see Fig.~\ref{f:diagtwo} for the intended placement of the variables).
\begin{figure}[ht]
\begin{center}
\setlength{\unitlength}{.045cm}
\begin{picture}(215,195)(0,3)
\thicklines
\put(35,25){\circle{4}}
\multiput(65,55)(15,15){2}{\circle{4}}
\multiput(110,100)(15,15){3}{\circle{4}}
\multiput(170,160)(15,15){3}{\circle{4}}

\put(140,70){\circle*{4}}
\put(145,69){$=\plane\land\diag$}
\put(140,55){\circle{4}}
\put(145,54){$=\plane\land\notdiag$}

\thinlines
\put(20,10){\line(1,0){185}}
\put(35,10){\line(1,1){175}}
\multiput(35,10)(15,15){12}{\line(0,1){13}}
\multiput(37,25)(15,15){11}{\line(1,0){13}}
\multiput(35,9)(15,0){12}{\line(0,1){2}}
\put(20,10){\line(0,1){185}}
\multiput(19,40)(0,15){11}{\line(1,0){2}}
\put(20,25){\circle*{2}}
\multiput(20,10)(30,30){2}{\circle*{4}}
\put(95,85){\circle*{4}}
\put(155,145){\circle*{4}}

\put(9,2){$\now$}
\put(10,23){$\uvar$}
\multiput(35,55)(15,15){2}{\circle*{2}}
\multiput(37,57)(15,15){2}{$\uvar$}
\multiput(65,100)(15,15){3}{\circle*{2}}
\multiput(67,102)(15,15){3}{$\uvar$}
\multiput(110,160)(15,15){3}{\circle*{2}}
\multiput(112,162)(15,15){3}{$\uvar$}
\multiput(213,189)(1.5,1.5){3}{\circle*{.5}}

\put(34,3){$u_0$}
\put(49,3){$u_1$}
\put(64,3){$u_2$}
\put(79,3){$u_3$}
\put(94,3){$u_4$}
\put(109,3){$u_5$}
\put(124,3){$u_6$}
\put(139,3){$u_7$}
\put(154,3){$u_8$}
\put(169,3){$u_9$}
\put(184,3){$u_{10}$}
\put(199,3){$u_{11}$}

\put(10,39){$u_2$}
\put(10,54){$u_3$}
\put(10,69){$u_4$}
\put(10,84){$u_5$}
\put(10,99){$u_6$}
\put(10,114){$u_7$}
\put(10,129){$u_8$}
\put(10,144){$u_9$}
\put(8,159){$u_{10}$}
\put(8,174){$u_{11}$}
\put(8,189){$u_{12}$}

\end{picture}
\hspace*{.5cm}
\begin{picture}(60,70)(0,-30)
\thinlines
\put(15,5){\line(1,0){35}}
\put(5,15){\line(0,1){45}}
\multiput(8,5)(2,0){3}{\circle*{.5}}
\multiput(53,5)(2,0){3}{\circle*{.5}}
\multiput(5,8)(0,2){3}{\circle*{.5}}
\multiput(5,63)(0,2){3}{\circle*{.5}}
\multiput(25,4)(15,0){2}{\line(0,1){2}}
\multiput(4,25)(0,15){3}{\line(1,0){2}}
\put(25,25){\circle*{2}}
\multiput(25,40)(15,0){2}{\circle*{2}}
\put(40,55){\circle*{2}}

\put(19,18){$\uvar$}
\put(20,33){$\uvaru$}
\put(40,58){$\uvarn$}
\put(40,33){$\uvarr$}

\put(22,-1){$u_n$}
\put(37,-1){$u_{n+1}$}
\put(-7,22){$u_m$}
\put(-15,38){$u_{m+1}$}
\put(-15,53){$u_{m+2}$}

\end{picture}

\end{center}
\caption{Encoding the $\nwplane$-grid in an $\HS$-model: version 1.\label{f:diagtwo}}
\end{figure}
%
%
%
%
%
%
Then we express (a.1) by the conjunction of
\begin{align}
\label{init1}
& \plane\land\diag\land\now,\\
& \succv[\now,\uvar],
\end{align}
(a.2) by  the conjunction of
\begin{align}
& \succv[\uvar,\uvaru],\\
& \U\Bigl(\plane\land\diag\to\Bv\bigl(\uvaru\to\Bh(\plane\to\diag)\bigr)\Bigr),
\end{align}
(a.3) by   the conjunction of
\begin{align}
& \U \bigl(\uvar\to\Bh(\plane\to\notdiag)\bigr),\\
& \U(\diag\land\notdiag\to\bot),
\end{align}
(a.4) by
\begin{equation}
 \U\Bigl(\plane\land\notdiag\to\Bv\bigl(\uvaru\to\Bh(\uvarr\to\uvar)\bigr)\Bigr),
\end{equation}
and (a.5) by the conjunction of
\begin{align}
& \succh[\uvaru,\uvarr],\\
& \succv[\uvarr,\uvarn],
\end{align}

\begin{align}
\label{lastr}
& \U\Bigl[\plane\land\diag\to\Bv\Bigl(\uvaru\to\Bh\bigl(\uvarr\to\Bv(\uvarn\to\uvar)\bigr)\Bigr)\Bigr].
\end{align}
%
%
%

It is not hard to show the following:

\begin{claim}\label{c:gridok}
Suppose $\M,\auf r,r'\zu\models\fdiag\land\fgrid$, where $\fgrid$ is the conjunction of  \eqref{init1}--\eqref{lastr}. Then $\now$, $\diag$, $\notdiag$ and $\uvar$ are properly placed
\textup{(}see Fig.~\ref{f:diagtwo}\textup{)}.
\end{claim}

Given a Turing machine $\A$, we will use the function $\tau_{\A}$ (defined in \eqref{tmfv}) to
force a diverging computation of $\A$ with empty input as follows.
We introduce (with a slight abuse of notation) a propositional variable $x$ for each $x\in\GA$. Then we formulate general constraints as
\begin{align}
\label{tmongrid}
& \U(x\to\plane),\qquad\mbox{for $x\in\GA$},\\
\label{tmuniq}
& \U(x\to\neg y),\qquad\mbox{for $x\ne y,\ x,y\in\GA$},
\end{align}
and then force the computation steps  by the conjunction of
\begin{align}
& \AD(q_0,\pounds),\\
\label{tmdiag}
& \U(\diag\to\sqcup),\\
\label{transH}
& \U\Bigl(y\land\AD z\land\AbD x\to \Bv\bigl(\uvar\to\Bh(\plane\to \tau_\mathcal{A}(x,y,z))\bigr)\Bigr),\qquad\mbox{for $(x,y,z)\in \WA$}.
\end{align}
Finally, we force non-halting with
\begin{equation}
\label{nonhalting}
\U \bigl((q_f,s)\to\bot\bigr),\qquad\mbox{for $s\in\Sigma\cup\{\pounds\}$}.
\end{equation}
Using Claims~\ref{c:grid}--\ref{c:gridok}, now it is straightforward to prove the following:

\begin{claim}\label{c:tm}
Let $\ftm$ be the conjunction of $\fdiag$, $\fgrid$ and \eqref{tmongrid}--\eqref{nonhalting}. If $\ftm$ is satisfiable
in an $\HS$-model, then $\A$ diverges with empty input.
\end{claim}

On the other hand, Fig.~\ref{f:diagtwo}  shows how to satisfy $\fdiag\land\fgrid$
 (using the irreflexive semantics)  in an $\HS$-model
that is based on some linear order $\T$ having an infinite ascending chain
\mbox{$u_0<u_1<\ldots$}. If $\A$ diverges with empty input, then we can add, for all $x\in\GA$,
\begin{multline}
\label{tmsound}
\nu(x)=\{\auf u_{n-1},u_{n}\zu\mid \mbox{$n>0$, $\CC_j(i)=x$}\\
\mbox{and the $n$th point in the grid-enumeration is $(i,j+1)$}\}
\end{multline}
to obtain an $\HS$-model $\M=(\F_{\T},\nu)$ satisfying \eqref{tmongrid}--\eqref{nonhalting} as well.


\smallskip

Next, we show how to modify the formula $\ftm$ above in order to be satisfiable with arbitrary semantics of the interval relations. `Uniqueness forcing' constraints like \eqref{diaguniq}, \eqref{uniqhnext}, and  \eqref{uniqvnext} above are clearly not satisfiable with the reflexive semantics.
Expanding on an idea of \cite{Spaan93},
\cite{Reynolds01122001,many_dimensional_modal_logics,gkwz05a}, we use the following \emph{chessboard trick} to solve this problem and
kind of `discretise' the $\HS$-model.
Take two fresh propositional variables $\Htick$ and $\Vtick$, and make the $\HS$-model $\M$ `chessboard-like' by the formula
\begin{equation}
\label{stripe}
\U (\Htick\to\Bv\Htick)\land \U(\Vtick\to\Bh\Vtick).
\end{equation}
However, to make it a real chessboard, we also need to have `cover' by these variables and their negations, that is,  for every interval in $\M$,
$\Htick\lor\neg\Htick$ and $\Vtick\lor\neg\Vtick$ should hold.
In order to express these by $\HShd$-formulas, we use the following
\emph{cover trick} of \cite[p.~11]{DBLP:conf/time/ArtaleKLWZ07}.
%
%
For any literals $\lambda$ and $\overline{\lambda}$, let
$\hcoverf[\lambda,\overline{\lambda}]$ denote the conjunction of
\begin{align}
\label{coverd}
& \U\bigl(\top\to \Dv(\mpred\land\Dh\xpred\land\Dh\ypred)\bigr),\\
\nonumber
& \U\bigl(\xpred\land\ypred\to\bot\bigr),\\
\nonumber
& \U\Bigl(\Dv\bigl(\mpred\land\Dh(\ypred\land\Dh\xpred)\bigr)\to\lambda\Bigr),\\
\nonumber
& \U\Bigl(\Dv\bigl(\mpred\land\Dh(\xpred\land\Dh\ypred)\bigr)\to\overline{\lambda}\Bigr),\\
\nonumber
& \U\bigl(\lambda\land\overline{\lambda}\to\bot\bigr),
\end{align}
where $\mpred$, $\xpred$, and $\ypred$ are fresh variables.
%

\emph{Soundness\/}:
Observe that $\hcoverf[\lambda,\overline{\lambda}]$ forces the model to be infinite. Also, it always implies that both $\lambda$ and $\overline{\lambda}$ are vertically stable,
that is,
\[
\U \bigl(\lambda\to\Bv\lambda\bigr)\land\U\bigl(\overline{\lambda}\to\Bv\overline{\lambda}\bigr).
\]
holds.
We can define $\vcoverf[\lambda,\overline{\lambda}]$ similarly, for horizontally stable $\lambda$ and $\overline{\lambda}$.
Now we take fresh variables $\nHtick$ and $\nVtick$, and define
$\cbr$ by taking
\begin{equation}
\cbr:=\ \
\hcoverf[\Htick,\nHtick]\land\vcoverf[\Vtick,\nVtick].
\end{equation}
%
Then \eqref{stripe} and the similar formula for $\nHtick$ and $\nVtick$ follow.
%
%
%
Suppose that $\M$ is an $\HS$-model based on some linear order $\T=(T,\le)$ satisfying $\cbr$.
We define two new binary relations
$\RhM$ and $\RvM$ on $T$ by taking, for all $u,v\in T$,
\begin{align*}
u \RhM v\quad\mbox{iff}\quad & \exists
z\ \Bigl(u\leq z\leq v\mbox{ and }\\
& \forall y\ \bigl(\mbox{if $\auf z,y\zu$ is in $\M$, then } \bigl(\M,\auf u,y\zu\models \Htick\ \leftrightarrow\ \M,\auf z,y\zu\models\neg \Htick)\bigr)\Bigr);\\
%
u \RvM v\quad\mbox{iff}\quad & \exists
z\ \Bigl(u\leq z\leq v\mbox{ and }\\
& \forall x\ \bigl(\mbox{if $\auf x,u\zu$ is in $\M$, then } \bigl(\M,\auf x,u\zu\models \Vtick\ \leftrightarrow\ \M,\auf x,z\zu\models\neg \Vtick)\bigr)\Bigr).
\end{align*}
Then it is straightforward to check that both $\RhM$ and $\RvM$ imply $\leq$, and both are transitive
and irreflexive. (They are not necessarily linear orders.)
%
We call a non-empty subset $I\subseteq T$ a \emph{horizontal} $\M$-\emph{interval} (shortly, an \emph{h-interval\/}), if $I$ is maximal with the
following two properties:
\begin{itemize}
\item
for all $x,y,z\in T$, if $x\leq y\leq z$ and $x,z\in I$ then $y\in I$;
\item
either $\M,\auf x,y\zu\models\Htick$,
for all $x\in I$ and $y\in T$ such that $\auf x,y\zu$ is in $\M$,
or $\M,\auf x,y\zu\models\neg\Htick$,
for all $x\in I$ and $y\in T$ such that $\auf x,y\zu$ is in $\M$.
\end{itemize}
%
For any $x\in T$, let $\hinte(x)$ denote the unique h-interval $I$ with
$x\in I$.
We define \emph{v-intervals} and $\vinte(x)$ similarly, using $\RvM$.
A set $S$ of the form $S=I\times J$ for some h-interval $I$ and v-interval $J$ is called
a \emph{square}.
For any $\auf x,y\zu$ in $\M$, let $\sqe(x,y)$ denote the unique square $S$ with
$\auf x,y\zu\in S$.

Now we define horizontal and vertical successor squares.
Given propositional variables $\pred$ and $\qpred$,
let $\hsuc[\pred,\qpred]$ be the conjunction of
\begin{align}
\label{hsuc1}
& \U\bigl(\pred\land\Htick\to\Dh(\qpred\land\nHtick)\bigr),\\
\nonumber
& \U(\pred\land\overline{\pred}\to\bot),\\
\label{hsuc3}
& \U(\pred\land\Htick\to\Bh\predp),\\
\nonumber
& \U\bigl(\predp\land\nHtick\to (\overline{\pred}\land\Bh\overline{\pred})\bigr),\\
%
\nonumber
& \U(\qpred\land\overline{\qpred}\to\bot),\\
\nonumber
& \U(\qpred\land\nHtick\to\Bh\qpredp),\\
\nonumber
& \U\bigl(\qpredp\land\Htick\to (\overline{\qpred}\land\Bh\overline{\qpred})\bigr),\\
\label{hsuc8}
& \U\bigl(\predp\land\Htick\land\Dh(\qpred\land\nHtick)\to\pred\bigr),\\
\label{hsuc9}
& \U(\predp\land\nHtick\land\Dh\qpred\to\qpred),\\
\label{hsucten}
& \U(\qpred\land \pred\to\bot),\\
\label{onemore}
& \U(\qpred\land\Dh\pred\to\bot)
\end{align}
plus similar formulas for the `$\pred\land\nHtick$' case (here $\overline{\pred}$, $\overline{\qpred}$, $\predp$ and $\qpredp$ are
fresh variables). One can define $\vsuc[\pred,\qpred]$ similarly.
%
Finally, we let
\[
\ffill[\pred]=\ \ \hsuc[\pred_l,\pred]\land\hsuc[\pred,\pred_r]\land\vsuc[\pred_d,\pred]\land
\vsuc[\pred,\pred_u],
\]
where $\pred_l$, $\pred_r$, $\pred_d$, and $\pred_u$ are fresh variables.

\begin{claim}\label{c:hsuc}
Suppose $\M$ satisfies $\cbr$ and $\hsuc[\pred,\qpred]$.
Then the following hold, for all $x$, $y$, $z$, $w$:
%
\begin{itemize}
\item[{\rm (\emph{i})}]
If $\M,\auf x,y\zu\models\pred$, then
there is $v$ such that $x\RhM v$ and $\M,\auf v,y\zu\models\qpred$.
\item[{\rm (\emph{ii})}]
If $\M,\auf x,y\zu\models\pred$ and $x\RhM z$, then $\M,\auf z,y\zu\not\models\pred$.
\item[{\rm (\emph{iii})}]
If $\M,\auf x,y\zu\models\qpred$ and $x\RhM z$, then $\M,\auf z,y\zu\not\models\qpred$.
\item[{\rm (\emph{iv})}]
If $\M,\auf x,y\zu\models\pred$, $z\in\hinte(x)$, $x\leq z$, then $\M,\auf z,y\zu\models\pred$.
\item[{\rm (\emph{v})}]
If $\M,\auf x,y\zu\models\pred$, $\M,\auf z,y\zu\models\qpred$,
$w\in\hinte(z)$ and $w\leq z$, then $\M,\auf w,y\zu\models\qpred$.
\item[{\rm (\emph{vi})}]
If $\M,\auf x,y\zu\models\pred$ and $\M,\auf z,y\zu\models\qpred$, then $x\RhM z$ and there is no $t$ with
$x\RhM t\RhM z$.
\end{itemize}
Similar statements hold if $\M$ satisfies $\vsuc[\pred,\qpred]$. Therefore,
\begin{itemize}
\item[{\rm (\emph{vii})}]
if $\M$ satisfies $\ffill[\pred]$ and $\M,\auf x,y\zu\models\pred$ then
$\M,\auf x',y'\zu\models\pred$ for all $\auf x',y'\zu\in\sqe(x,y)$.
\end{itemize}
\end{claim}

\begin{proof}
It is mostly straightforward. We show the trickiest case,
(\emph{vi}) We have $x\leq z$ by \eqref{hsucten}.
Suppose, say, that $\M,\auf x,y\zu\models\Htick$. By (\emph{i}), there is
$v$ such that $x\RhM v$ and $\M,\auf v,y\zu\models\qpred$, and so
$\M,\auf v,y\zu\models\nHtick$. Then $z\in\hinte(v)$ follows by (iii), and so $x\RhM z$.
Now let $t$ be such that $x\leq t\leq z$.
If $\M,\auf t,y\zu\models\Htick$, then
$\M,\auf t,y\zu\models\pred$ by \eqref{hsuc3} and \eqref{hsuc8}, and so $t\in\hinte(x)$ by (\emph{ii}).
If $\M,\auf t,y\zu\models\nHtick$, then
$\M,\auf t,y\zu\models\qpred$ by \eqref{hsuc3} and \eqref{hsuc9}, and so $t\in\hinte(z)$ by (\emph{iii}).
\end{proof}

\emph{Soundness\/}:
If $\M$ satisfies $\ffill[\pred]$ then $\pred$ must be both `horizontally and vertically square-unique' in the following sense: if $\M,\auf x,y\zu\models\pred$ and $\M,\auf x',y'\zu\models\pred$ for some $x\RhM x'$ and $y\RvM y'$, then $\sqe(x,y)=\sqe(x',y')$ must follow.


Now, using this `chessboard trick', we can modify the formula $\ftm$ above
for any semantical choice of the interval relations.
To begin with,
 instead of using $\fdiag$, we force a unique infinite sequence of $\plane$-squares by introducing a fresh variable $\ndiag$, and taking the conjunction $\fdiagr$ of the following formulas:
\begin{align}
\nonumber
& \cbr\land\ffill[\plane]\land\ffill[\ndiag],\\
\label{gridh}
& \plane\land\hsuc[\plane,\ndiag],\\
\nonumber
& \vsuc[\ndiag,\plane].
\end{align}
Then we have the following generalisation of Claim~\ref{c:grid}:

\begin{claim}\label{c:gridall}
Suppose $\M,\auf r,r'\zu\models\fdiagr$.
Then there exist infinite sequences
\mbox{$( x_n \mid n<\omega)$} and $( y_n \mid n<\omega)$
of points in $\T$ such that the following hold:
%
\begin{itemize}
\item[{\rm (\emph{i})}]
$r=x_0\RhM x_1\RhM \ldots \RhM a_n\RhM \ldots$ and $r'=y_0\RvM y_1\RvM \ldots \RvM y_n\RvM \ldots$.
\item[{\rm (\emph{ii})}]
There is no $x$ with $x_n\RhM x\RhM x_{n+1}$ and
there is no $y$ with $y_n\RvM y\RhM y_{n+1}$, for any $n<\omega$.
\item[{\rm (\emph{iii})}]
For all $x,y$,
$\M,\auf x,y\zu\models\plane$ iff $\auf x,y\zu\in\sqe(x_n,y_n)$ for some $n<\omega$.
\end{itemize}

\end{claim}

In order to show the soundness of $\fdiagr$, let $\T=(T,\le)$ be a linear order containing an infinite ascending chain $u_0<u_1<\ldots$.
%

\begin{claim}\label{c:soundgridr}
$\fdiagr$ is satisfiable in an \HS-model based on $\T$ under arbitrary semantics.
\end{claim}

\begin{proof}
For each $n<\omega$, we let
\[
U_n=\{x\in T \mid u_n\leq x < u_{n+1}\}.
\]
It is straightforward to check that the following $\HS$-model $\M=(\F_{\T},\nu)$ satisfies $\hcoverf[\Htick,\nHtick]$:
\begin{align*}
\nu(\Htick) = & \{\auf x,y\zu\in \int(\mathfrak T) \mid x\in U_n,\ n\mbox{ is even}\},\\
\nu(\nHtick) =&  \{\auf x,y\zu\in \int(\mathfrak T) \mid x\in U_n,\ n\mbox{ is odd}\},\\
\nu({\sf M}_{\Htick})= & \{\auf x,y\zu\in \int(\mathfrak T) \mid x\in U_m,\ y\in U_n,\ \mbox{both $m,n$ are even, or both $m,n$ are odd}\},\\
 \nu({\sf X}_{\Htick})= &\{\auf x,y\zu\in \int(\mathfrak T) \mid x\in U_n,\ y\in U_{n+1}\cup U_{n+2},\ n\mbox{ is even}\},\\
\nu({\sf Y}_{\Htick})= &\{\auf x,y\zu\in \int(\mathfrak T) \mid x\in U_n,\ y\in U_{n+1}\cup U_{n+2},\ n\mbox{ is odd}\}.
\end{align*}
$\vcoverf[\Vtick,\nVtick]$ can be satisfied similarly. The rest is obvious.
\end{proof}

Next, consider the formula $\fgrid$ defined in Claim~\ref{c:gridok}.
Let $\fgridr$ be obtained from $\fgrid$ by replacing each occurrence of  $\succh$ by $\hsuc$
and each occurrence of  $\succv$ by $\vsuc$, and adding the conjuncts
%
%
$\ffill[\pred]$ for $\pred\in\{\now,\plane,\diag,\notdiag,\uvar,\uvaru,\uvarr,\uvarn\}$.
%
Using Claim~\ref{c:gridall},
it is straightforward to show that we have the analogue of Claim~\ref{c:gridok} for squares.

Finally, given a Turing machine $\A$, let $\ftmr$ be the conjunction of
of $\fdiagr$, $\fgridr$, \eqref{tmongrid}--\eqref{nonhalting},
and $\ffill[x]$ for each $x\in \GA$. Then we have:

\begin{claim}\label{c:tmr}
If $\ftmr$ is satisfiable in an $\HS$-model, then
$\A$ diverges with empty input.
\end{claim}

On the other hand, using Fig.~\ref{f:diagtwo}, Claim~\ref{c:soundgridr} and \eqref{tmsound} it is easy to show
how to satisfy $\ftmr$ in an $\HS$-model that is based on some linear order $\T$ having an infinite ascending chain \mbox{$u_0<u_1<\ldots$}, regardless which semantics of the interval relations is considered.


\smallskip

(\emph{ii})  We reduce  `halting' to $\Fin$-satisfiability. We show how to modify the formula
$\ftmr$ above to achieve this. To begin with,
`generating' conjuncts like \eqref{coverd} and its `vertical' version in $\cbr$,
and  \eqref{hsuc1} and its $\nHtick$ version in $\hsuc[\plane,\ndiag]$ of \eqref{gridh}
are not satisfiable in $\HS$-models based on finite orders. In order to obtain a finitely satisfiable version,
we introduce a fresh variable $\evar$, replace \eqref{nonhalting} with the conjunction of
\begin{align}
\label{endunit}
& \U(\evar\to\plane),\\
\label{notfinalstate}
& \U(\evar\land x\to\bot),\qquad\mbox{for $x\in \Sigma\cup\{\pounds\}\cup\bigl(Q^-\times(\Sigma\cup\{\pounds\})\bigr)$},
\end{align}
then replace conjunct \eqref{coverd} in $\hcoverf[\lambda,\overline{\lambda}]$
with the conjunction of
\[
 \U\bigl(\auf\R\zu\evar\to \Dv(\mpred\land\Dh\xpred\land\Dh\ypred)\bigr),
 \qquad\mbox{for $\R\in\{\after,\beginsb,\duringb,\later,\overlaps\}$},
 \]
(and similarly in $\vcoverf[\lambda,\overline{\lambda}]$), and
then replace conjunct \eqref{hsuc1} in
$\hsuc[\plane,\ndiag]$ with the conjunction of
%
\[
 \U\bigl(\auf\R\zu\evar\land\plane\land\Htick\to\Dh(\ndiag\land\nHtick)\bigr),
  \qquad\mbox{for $\R\in\{\after,\beginsb,\duringb,\later,\overlaps\}$}
\]
(and do similarly for the `$\nHtick$-version', and for the `generating' conjuncts in $\vsuc[\ndiag,\plane]$).
\end{proof}


\begin{theorem}\label{t:undeccoreirrefl}
\textbf{\textup{(}$\HSc$, irreflexive semantics\textup{)}}\\
\textup{(}i\textup{)}
For any class $\CC$ of linear orders containing an infinite order,
$\CC(<)$-satisfiability of $\HSc$-formulas is undecidable.
\textup{(}ii\textup{)}
$\Fin(<)$-satisfiability of $\HSc$-formulas is undecidable.
\end{theorem}

\begin{proof}
(\emph{i})
We reduce  {\sc non-halting} to $\CC(<)$-satisfiability.
Given an $\HS$-model $\M$ based on some linear order $\T$, observe that the formula $\fdiag$ (defined in
Claim~\ref{c:grid}) that forces a unique infinite $\plane$-sequence
$\bigl(\auf u_n,u_{n+1}\zu \mid n<\omega\bigr)$ in $\M$ is within $\HSc$.
However, the formula $\fgrid$
(defined in Claim~\ref{c:gridok})
we used in the proof of Theorem~\ref{t:undechorn} to encode the $\nwplane$-grid in $\M$  with the help of properly placed $\uvar$-pointers  contains several seemingly `non-$\HSc$-able' conjuncts. In order to fix this,
below we will force the proper placement of $\uvar$-pointers in a different way.

Consider again the enumeration of $\nwplane$ in Fig.~\ref{f:enumtwo}.
Observe that the enumerated points can be organized in (horizontal) \emph{lines}:
$\mathsf{line}_1=( 1,2)$,
$\mathsf{line}_2=(3,4,5)$,
$\mathsf{line}_3=(6,7,8,9)$,
 and so on.
%
 %
 Consider the following properties of this enumeration (different from the ones listed as (a.1)--(a.5) in
 the proof of Theorem~\ref{t:undechorn} above):
\begin{enumerate}
\item[(b.1)]
$\mathsf{start\_of}(\mathsf{line}_{1})=1$,
and $\mathsf{up\_neighbour\_of}(0)=1$.
\item[(b.2)]
$\mathsf{start\_of}(\mathsf{line}_{i+1})=\mathsf{end\_of}(\mathsf{line}_{i})+1$, for all $i>0$.
\item[(b.3)]
Every line starts with some $n$ on the wall and ends with some $m$ on the diagonal.
\item[(b.4)]
If $n$ is in $\mathsf{line}_i$, then $\mathsf{up\_neighbour\_of}(n)$ is in $\mathsf{line}_{i+1}$, for all $i$.
\item[(b.5)]
For every $m,n$, if $m<n$ then $\mathsf{up\_neighbour\_of}(m)<\mathsf{up\_neighbour\_of}(n)$.
\item[(b.6)]
For every $n>0$ on the wall, there is $m$ with $\mathsf{up\_neighbour\_of}(m)=n$.
\item[(b.7)]
For every $n$, if $n$ is neither on the wall nor on the diagonal, then
there is $m$ with $\mathsf{up\_neighbour\_of}(m)=n$.
\end{enumerate}
Observe that (b.1) and (b.2) imply that every $n$ in the enumeration belongs to $\mathsf{line}_i$ for
some $i$. Also, by (b.2) and (b.3), for every $i$
there is a unique $m$ in $\mathsf{line}_i$ that is on
the diagonal (its last according to the enumeration). As $\mathsf{up\_neighbour\_of}$ is an injective
function, by (b.4) we have that
\[
\mbox{number of points in $\mathsf{line}_i\leq$ number of points in $\mathsf{line}_{i+1}$}.
\]
Further, by (b.4), (b.6) and (b.7),
\[
\mbox{number of non-diagonal points in $\mathsf{line}_{i+1}\leq$ number of points in $\mathsf{line}_{i}$}.
\]
Therefore,
%
\[
\mathsf{length\_of}(\mathsf{line}_{i+1})=\mathsf{length\_of}(\mathsf{line}_i)+1\mbox{ for all $i$.}
\]
%
Finally, by (b.4) and (b.5) we obtain that $\mathsf{line}_i$ is what it should be in Fig.~\ref{f:enumtwo},
and so we have:

\begin{claim}\label{c:nwenum}
Properties {\rm (b.1)--(b.7)} uniquely determine%
\footnote{among those that contain the enumeration of the diagonal locations as $(0,0),\ldots,( 1,1),\ldots,(2,2),\ldots$}
%
the enumeration in Fig.~\ref{f:enumtwo}.
\end{claim}

Given a unique infinite $\plane$-sequence $\mathcal{U}=\bigl(\auf u_n,u_{n+1}\zu \mid n<\omega\bigr)$
in $\M$ as in Claim~\ref{c:grid} above,
we now encode (b.1)--(b.7) as follows.
In addition to $\uvar$, $\diag$, and $\now$,
we will also use a variable $\wvar$ to mark those $\plane$-points that are on the wall,
and a variable $\lvar$ to mark lines in the following sense:
$\M,\auf x,y\zu\models\lvar$ iff $x=u_m$, $y=u_n$ and $(m+1,\ldots,n)$ is a line
(see Fig.~\ref{f:diag} for the intended placement of the variables).
\begin{figure}[ht]
\begin{center}
\setlength{\unitlength}{.05cm}
\begin{picture}(215,195)(0,5)
\put(150,70){\circle*{4}}
\put(155,69){$=\plane$}

\put(20,10){\line(1,0){185}}
\put(35,10){\line(1,1){175}}
\multiput(35,10)(15,15){12}{\line(0,1){15}}
\multiput(35,25)(15,15){11}{\line(1,0){15}}
\multiput(35,9)(15,0){12}{\line(0,1){2}}
\put(20,10){\line(0,1){185}}
\multiput(19,40)(0,15){11}{\line(1,0){2}}
\put(20,25){\circle*{2}}
\multiput(20,10)(15,15){13}{\circle*{4}}

\put(5,10){$\diag$}
\put(5,4){$\now$}
\put(11,23){$\uvar$}
\multiput(35,55)(15,15){2}{\circle*{2}}
\multiput(37,57)(15,15){2}{$\uvar$}
\put(35,40){\circle*{2}}
\put(29,42.5){$\lvar$}
\multiput(65,100)(15,15){3}{\circle*{2}}
\multiput(67,102)(15,15){3}{$\uvar$}
\put(65,85){\circle*{2}}
\put(59,87){$\lvar$}
\put(36,28){$\wvar$}
\put(51,43){$\diag$}
\put(66,58){$\wvar$}
\put(96,88){$\diag$}
\put(111,103){$\wvar$}
\put(156,148){$\diag$}
\put(171,163){$\wvar$}
\put(110,145){\circle*{2}}
\put(105,147){$\lvar$}
\multiput(110,160)(15,15){3}{\circle*{2}}
\multiput(112,162)(15,15){3}{$\uvar$}
\multiput(213,189)(1.5,1.5){3}{\circle*{.5}}

\put(34,3){$u_0$}
\put(49,3){$u_1$}
\put(64,3){$u_2$}
\put(79,3){$u_3$}
\put(94,3){$u_4$}
\put(109,3){$u_5$}
\put(124,3){$u_6$}
\put(139,3){$u_7$}
\put(154,3){$u_8$}
\put(169,3){$u_9$}
\put(184,3){$u_{10}$}
\put(199,3){$u_{11}$}

\put(10,39){$u_2$}
\put(10,54){$u_3$}
\put(10,69){$u_4$}
\put(10,84){$u_5$}
\put(10,99){$u_6$}
\put(10,114){$u_7$}
\put(10,129){$u_8$}
\put(10,144){$u_9$}
\put(8,159){$u_{10}$}
\put(8,174){$u_{11}$}
\put(8,189){$u_{12}$}

\end{picture}
\end{center}
\caption{Encoding the $\nwplane$-grid in an $\HS$-model: version 2.\label{f:diag}}
\end{figure}

To begin with, we express that $\mathsf{up\_neighbour\_of}$ is an injective function by
%
\begin{equation}
\label{initone}
\U(\uvar\to\neg\Dh\uvar\land\neg\Dv\uvar),
\end{equation}
then we express (b.1) by the conjunction of
\begin{align}
& \now\land\AD\lvar,\\
& \U (\uvar\to\neg\DD\now),
\end{align}
%
(b.2) by
\begin{equation}\label{existsline}
\U(\lvar\to\AD\lvar),\\
\end{equation}
(b.3) by the conjunction of
%
\begin{align}
& \U(\wvar\to\plane),\\
& \U(\diag\to\plane),\\
\label{existswall}
& \U(\lvar\to\Dh\diag\land\BD\wvar),
\end{align}
(b.4) by  the conjunction of
%
\begin{align}
\label{gridtoup}
& \U (\plane\to\Dv\uvar),\\
\label{upatunit}
& \U(\uvar\to\Dh\plane\land\BD\plane),\\
& \U(\uvar\to\neg\Dv\lvar\land\neg\DD\lvar),
\end{align}
(b.5) by
\begin{equation}
\U (\uvar\to\neg\DD\uvar),
\end{equation}
(b.6) by
\begin{equation}\label{upatwall}
\U(\wvar\to\EbD\uvar).
\end{equation}
Finally, we can express (b.7) by
\begin{equation}\label{coremidline}
\bigl[\DbD\lvar\land\plane\imph\AD\AbD\uvar\bigr],
\end{equation}
using the `binary implication trick' introduced in Section~\ref{pspace}.
%

%
%
%


Now it is not hard to show the following:

\begin{claim}
Suppose $\M,\auf r,r'\zu\models\fdiag\land\fgridcore$, where $\fgridcore$ is a conjunction of \eqref{initone}--\eqref{coremidline}.  Then $\now$, $\wvar$, $\diag$, $\lvar$, and $\uvar$ are properly placed \textup{(}see Fig.~\ref{f:diag}\textup{)}.
\end{claim}

On the other hand, using Fig.~\ref{f:diag} it is not hard to see that $\fgridcore$ is satisfiable
(using the irreflexive semantics) in an $\HS$-model that is based on some linear order $\T$ having an infinite ascending chain $u_0<u_1<\ldots$. In particular, conjunct
 \eqref{coremidline} is satisfiable because of the following:
$\AD\AbD\uvar$ is clearly horizontally stable, and it is easy to check that
for every $x,n$ with $\M,\auf x,u_n\zu\models\neg \AD\AbD\uvar$, we have
$\M,\auf x,u_n\zu\models\neg\DbD\lvar$.

Given a Turing machine $\A$, consider the conjuncts \eqref{tmongrid}--\eqref{nonhalting} above, and
observe that the only non-$\HSc$ conjuncts among them
are \eqref{transH} for $(x,y,z)\in \WA$. In order to replace these with $\HSc$-formulas we introduce
the following fresh propositional variables:
\begin{itemize}
\item
$(y,z)$ and $\overline{( y,z)}$, for all $y,z\in\GA$, and
\item
$( x,y,z)$ and $\overline{( x,y,z)}$, for all $(x,y,z)\in \WA$.
\end{itemize}
Then we again use the `binary implication trick' of Section~\ref{pspace}
(and its `vertical' version),
and take the conjunction of the
following formulas, for all $y,z\in\GA$ and all $( x,y,z)\in \WA$:
\begin{align*}
& \bigl[\AbD y\land z\impv\overline{( y,z)}\bigr],\\
& \U\bigl(\overline{( y,z)}\to\AbD( y,z)\bigr),\\
& \U\bigl(( y,z)\to\plane\bigr),\\
& \bigl[\AD (y,z)\land x\imph\overline{( x,y,z)}\bigr],\\
& \U\bigl(\overline{(x,y,z)}\to\AD( x,y,z)\bigr),\\
& \U\bigl(( x,y,z)\to\uvar\land\Dh \tau_{\A}(x,y,z)\bigr).
\end{align*}
Fig.~\ref{f:Tm} shows the intended meaning of these formulas, and also how to satisfy them
in the $\HS$-model $\M$ defined in \eqref{tmsound}.
\begin{figure}[ht]
\setlength{\unitlength}{0.07cm}
\begin{center}
\begin{picture}(80,90)(0,3)
\thinlines
\put(0,0){\line(1,1){80}}
\put(0,25){\line(1,0){25}}
\multiput(25,35)(10,10){2}{\line(1,0){10}}
\put(25,70){\line(1,0){45}}
\multiput(25,70)(10,10){2}{\line(0,-1){45}}
\put(15,25){\line(0,-1){10}}
\put(60,70){\line(0,-1){10}}

\multiput(15,25)(10,10){3}{\circle*{3}}
\put(60,70){\circle*{3}}
\put(25,70){\circle*{2}}
\put(11,27){$x$}
\put(12,36){$( y,z)$}
\put(26.5,31.5){$y$}
\put(31,46){$z$}
\put(24,73){$\uvar$}
\put(7,68){$( x,y,z)$}
\put(34,83){$\downarrow$}
\put(30,87){$\overline{( y,z)}$}
\put(27,24){$\leftarrow$}
\put(32,24){$\overline{( x,y,z)}$}
\put(50,73){$\tau_{\A}(x,y,z)$}
\put(60,15){\circle*{3}}
\put(63,14){$=\plane$}

\end{picture}
\end{center}
\caption{Encoding formula \eqref{transH} in $\HSc$.}\label{f:Tm}
\end{figure}

(\emph{ii}) We reduce  {\sc halting} to $\Fin(<)$-satisfiability.
In order to achieve this, we introduce a fresh variable $\evar$,
 replace \eqref{nonhalting} with the conjunction of \eqref{endunit} and \eqref{notfinalstate},
 and replace the `generating' conjunct
\eqref{diaggenwithA} of $\fdiag$ with
\begin{equation}\label{fingridcore}
\plane\,\land\bigl[\LD\evar\land\plane\imph\AD\plane\bigr],
\end{equation}
using the binary implication trick.
%
%
\end{proof}


\begin{theorem}\label{t:undechornboxirrefldisc}
\textbf{\textup{(}$\HShb$, discrete orders, irreflexive semantics\textup{)}}\\
\textup{(}i\textup{)}
For any class $\Disinf$ of discrete linear orders containing an infinite order,
$\Disinf(<)$-satisfiability of $\HShb$-formulas is undecidable.
\textup{(}ii\textup{)}
$\Fin(<)$-satisfiability of $\HShb$-formulas is undecidable.
\end{theorem}

\begin{proof}
(\emph{i}) We again reduce  `non-halting' to satisfiability, modifying the techniques
employed in the proofs of Theorems~\ref{t:undechorn} and \ref{t:undeccoreirrefl}.
In both of these proofs, `positive' $\auf\R\zu$-operators are used for two purposes.
First, they help to `generate' an infinite $\plane$-sequence; see formula \eqref{diaggenwithA}.
Second, they help to `generate' appropriate pointers for the encoding of the
$\nwplane$-grid via the enumeration in Fig.~\ref{f:enumtwo}; see formulas
$\succh$, $\succv$, \eqref{existsline}, \eqref{existswall}--\eqref{upatunit}, \eqref{upatwall}
and \eqref{coremidline}.
Below, we show how to `mimic' these features within $\HShb$.
Recall that formulas of the form $\U(\varphi\to\neg[\R]\bot)$ are within $\HShb$.

Take any $\HS$-model $\M$ based on a discrete linear order $\T$, and consider the irreflexive semantics of the interval relations.
%
In case of these semantical choices,
we can single out  $\plane$-intervals as follows. Let  $\fdiagb$ denote the formula
%
\[
\U\bigl(\plane\to\neg\Bh\bot\land\Bh\Bh\bot\bigr)\land\U\bigl(\Dh\Bh\bot\land\Bh\Bh\bot\to \plane\bigr).
\]
It is not hard to see that if $\M$ satisfies $\fdiagb$ then,
for all $\auf x,x'\zu$  in $\int(\T)$, we have $\M,\auf x,x'\zu\models\plane$ iff
$x'$ is an immediate successor of $x$ in $\T$.
(Note that
this is not the same $\plane$-sequence as in the proof of Theorem~\ref{t:pspacecoreboxirrefldisc}.)
This $\plane$-sequence has the useful property of having access to the `next' and `previous' $\plane$-intervals with the $\after$ and $\afterb$ interval relations, respectively.
%
%

The following \emph{nw-next trick} will also be essential.
 For any finite conjunction $\varphi$ of literals and any literal $\lambda$,
we define the formula $\bigl[\varphi\nwnxt\lambda\bigr]$ as the conjunction of
\begin{align*}
& \U(\varphi\to \lambda_\downarrow\land\BB\lambda_\downarrow\land\Bv\lambda_\uparrow),\\
& \U(\lambda_\uparrow\land\BB\lambda_\downarrow\to\lambda_\ast),\\
& \U(\lambda_\ast\to\lambda_\to\land\Bh\lambda_\to\land\EbB\lambda_\leftarrow),\\
& \U(\lambda_\leftarrow\land\Bh\lambda_\to\to\lambda),
\end{align*}
where $\lambda_\downarrow$, $\lambda_\uparrow$, $\lambda_\to$, $\lambda_\leftarrow$ and $\lambda_\ast$ are fresh
variables.
Now suppose $u_0<u_1<\ldots<u_n<\ldots$ is an infinite sequence of subsequent points in $\T$.
(We will `force' its existence with the formula \eqref{infunit} below.)
It is easy to see the following:

\begin{claim}
If $\M\models\bigl[\varphi\nwnxt\lambda\bigr]$ and $\M,\auf u_i,u_j\zu\models\varphi$, then
$\M,\auf u_{i-1},u_{j+1}\zu\models\lambda$.
\end{claim}

\emph{Soundness\/}:
Observe that in order to satisfy $\bigl[\varphi\nwnxt\lambda\bigr]$ there are certain restrictions on
$\varphi$ and $\lambda$. For example, there is no problem
whenever they are both `horizontally and vertically unique in $\M$' in the following sense:
If $\M,\auf x,y\zu\models\varphi$ then $\M,\auf x',y\zu\not\models\varphi$ and
$\M,\auf x,y'\zu\not\models\varphi$ for any $x'\ne x$, $y'\ne y$ (and similarly for $\lambda$).

\smallskip
Next, we force the proper placement of line- and up-pointers of the $\nwplane$-grid in Fig.~\ref{f:enumtwo}
 in a novel way, different from the ones in the proofs of Theorems~\ref{t:undechorn} and \ref{t:undeccoreirrefl}.
In representing this enumeration by our $\plane$-sequence, each line will be
followed by a `mirror'-unit, then by a `mirrored copy' of the next line with its locations listed in reverse order,
and then by a proper listing of the next line's locations.
In order to achieve this, we introduce the following fresh propositional variables:
\begin{itemize}
\item
$\gridp$, $\wvar$ and $\diag$ (to mark those $\plane$-intervals that represent line-locations and the respective wall- and diagonal-ends of each line);
\item
$\gridc$ (to mark $\plane$-intervals representing the mirror-copies of proper line locations);
\item
$\uvar$ and $\mvar$ (to mark pointers helping to access the up-neighbour of each location);
\item
$\gridm$, $\mvarlast$ and $\uvarlast$ (to mark the beginning and end of each `north-west going' $\mvar$- and $\uvar$-sequence, respectively).
\end{itemize}
See Fig.~\ref{f:newnw} for the intended placement of these variables, and for an example of how to
access, say, grid-location $(1,4)$ from $(1,3)$, and $(1,3)$ from $(1,2)$ with the help of
up- and mirror-pointers.

We force the proper placement of these variables
by the conjunction  $\fgridb$ of the following formulas:
%
\begin{align*}
& \init\land\bigl[\init\nwnxt\uvarlast\bigr],\\
& \U(\init\to \plane\land\wvar),\\
& \U(\plane\land\EbD\uvarlast\to\diag),\\
& \U\bigl(\diag\to\AB(\plane\to\gridm)\bigr),\\
& \bigl[\gridm\nwnxt\mvar\bigr],\\
& \bigl[\wvar\nwnxt\uvar\bigr],\\
& \U(\plane\land\EbD\uvar\to\gridp),\\
& \bigl[\mvar\land\BD\gridp\nwnxt\mvar\bigr],\\
& \U(\mvar\land\BD\wvar\to\mvarlast),\\
& \U(\plane\land\EbD\mvar\to\gridc),\\
& \U(\plane\land\EbD\mvarlast\to\wvar),\\
& \bigl[\uvar\land\BD\gridc\nwnxt\uvar\bigr],\\
& \U(\uvar\land\BD\gridm\to\uvarlast).
\end{align*}

\begin{figure}[ht]
\begin{center}
\setlength{\unitlength}{.08cm}
\begin{picture}(160,160)(-8,5)
\thinlines
\put(10,7){\vector(1,0){130}}
\put(0,12){\vector(0,1){152}}
\multiput(20,6)(5,0){24}{\line(0,1){2}}
\multiput(-1,20)(0,5){29}{\line(1,0){2}}
\put(18,4){${}_{u_0}$}
\put(23,4){${}_{u_1}$}
\put(28,4){${}_{u_2}$}
\put(33,4){$\ldots$}
\put(-6,20){${}_{u_1}$}
\put(-6,25){${}_{u_2}$}
\multiput(-4,30)(0,2){3}{\circle*{.5}}

\put(20,20){\circle*{1}}
\put(22,17){$\wvar$}
\put(12,17){$\init$}
\put(15,25){\circle{1.5}}

\put(25,25){\circle*{3}}
\put(27,22){$\diag$}
\put(29,29){$\ast$}
\put(30,30){\circle{3}}
\multiput(24,34)(-5,5){2}{$\ast$}

\multiput(35,35)(5,5){2}{\circle{3}}
\put(42,37){$\wvar$}
\multiput(35,45)(-5,5){2}{\circle{1.5}}
\put(46,42){\scriptsize $(1,2)$}
\multiput(45,48)(0,1){16}{\circle*{.5}}
\multiput(47,65)(1,0){16}{\circle*{.5}}
\multiput(65,68)(0,1){6}{\circle*{.5}}
\multiput(67,75)(1,0){6}{\circle*{.5}}
\multiput(75,78)(0,1){26}{\circle*{.5}}
\multiput(77,105)(1,0){26}{\circle*{.5}}
\multiput(105,108)(0,1){6}{\circle*{.5}}
\multiput(107,115)(1,0){6}{\circle*{.5}}

\multiput(45,45)(5,5){2}{\circle*{3}}
\put(52,47){$\diag$}
\put(54,54){$\ast$}
\put(55,55){\circle{3}}
\multiput(49,59)(-5,5){3}{$\ast$}

\multiput(60,60)(5,5){3}{\circle{3}}
\put(72,67){$\wvar$}
\multiput(65,75)(-5,5){3}{\circle{1.5}}
\put(76,72){\scriptsize $(1,3)$}

\multiput(75,75)(5,5){3}{\circle*{3}}
\put(87,82){$\diag$}
\put(89,89){$\ast$}
\put(90,90){\circle{3}}
\multiput(84,94)(-5,5){4}{$\ast$}

\multiput(95,95)(5,5){4}{\circle{3}}
\put(112,107){$\wvar$}
\multiput(105,115)(-5,5){4}{\circle{1.5}}
\put(116,112){\scriptsize $(1,4)$}

\multiput(115,115)(5,5){4}{\circle*{3}}
\put(132,127){$\diag$}
\put(134,134){$\ast$}
\put(135,135){\circle{3}}
\multiput(129,139)(-5,5){5}{$\ast$}

\put(3,27){$\uvarlast$}
\put(23,52){$\uvarlast$}
\put(47,87){$\uvarlast$}
\put(83,132){$\uvarlast$}

\put(5,42){$\mvarlast$}
\put(30,72){$\mvarlast$}
\put(60,112){$\mvarlast$}
\put(100,162){$\mvarlast$}

\put(142,142){$\swarrow$}
\put(144,147){$\plane$}

\thinlines
\put(30,10){\line(1,1){5}}
\multiput(30,10)(5,5){2}{\line(-1,1){4}}
\put(34,10){line$_1$}

\put(50,30){\line(1,1){10}}
\multiput(50,30)(10,10){2}{\line(-1,1){4}}
\put(56,32){line$_2$}

\put(80,60){\line(1,1){15}}
\multiput(80,60)(15,15){2}{\line(-1,1){4}}
\put(88,64){line$_3$}

\put(120,100){\line(1,1){20}}
\multiput(120,100)(20,20){2}{\line(-1,1){4}}
\put(130,106){line$_4$}

\put(120,50){\circle*{3}$=\gridp$}
\put(120,43){\circle{3}$=\gridc$}
\put(120,36){\circle{3}$=\gridm$}
\put(119,34.9){$\ast$}
\put(120,29){\circle{1.5}$\ =\uvar$}
\put(119,21){$\ast=\mvar$}
\end{picture}
\end{center}
\caption{Encoding the $\nwplane$-grid in an $\HS$-model: version 3.}\label{f:newnw}
\end{figure}

Then it is not hard to show the following:

\begin{claim}\label{c:mirrrorgrid}
If $\M,\auf u_0,u_1\zu\models\fdiagb\land\fgridb$, then all variables are placed as in Fig.~\ref{f:newnw}.
\end{claim}

Finally, given a Turing machine $\A$,
we again place the subsequent configurations of its computation with empty input on the subsequent
lines of the $\nwplane$-grid (see Fig.~\ref{f:lines}), using the function $\tau_{\A}$ defined in
\eqref{tmfv}. We define the formula $\ftmb$ as follows.
First, we ensure that there are infinitely many $\plane$-intervals with
\begin{equation}\label{infunit}
 \U (\plane\land x\to\neg\Bv\bot),
 \qquad\mbox{for $x\in \Sigma\cup\{\pounds\}\cup\bigl((Q-\{q_f\})\times(\Sigma\cup\{\pounds\})\bigr)$}.
\end{equation}
Next,
we take the general constraints \eqref{tmongrid} and \eqref{tmuniq},
then initialize the computation with
\[
\U \bigl(\init\to( q_0,\pounds)\bigr),
\]
and then force the computation steps with the conjunction of \eqref{tmdiag} and
\begin{align*}
& \U(\gridm\to\pounds),\\
& \U\Bigl(\gridp\land y\land\AD z\land\AbD x\to \Bv\bigl(\mvar\to\Bh(\plane\to \tau_\mathcal{A}(x,y,z))\bigr)\Bigr),\\
& \hspace*{10.5cm} \mbox{for $( x,y,z)\in \WA$},\\
& \U\Bigl(\wvar\land y\land\AD z\to \Bv\bigl(\mvar\to\Bh(\plane\to \tau_\mathcal{A}(\sqcup,y,z))\bigr)\Bigr),\quad \mbox{for $(\sqcup,y,z)\in \WA$},\\
& \U\Bigl(\gridc\land\Dv\uvar\land x\to \Bv\bigl(\uvar\to\Bh(\plane\to x)\bigr)\Bigr),\quad \mbox{for $x\in \GA$}.
\end{align*}
Then we force non-halting with \eqref{nonhalting}.
Using Claim~\ref{c:mirrrorgrid}, now it is straightforward to prove the following:

\begin{claim}\label{c:tmb}
If $\ftmb$ is satisfiable in an $\HS$-model based on a discrete linear order, then $\A$ diverges with empty input.
\end{claim}

On the other hand, using Fig.~\ref{f:newnw} it is not hard to see that $\fdiagb\land\fgridb$ is satisfiable
(using the irreflexive semantics) in an $\HS$-model that is based on some discrete linear order $\T$ having an infinite ascending chain $u_0<u_1<\ldots$ of subsequent points.
 If $\A$ diverges with empty input, then it is not hard to modify the $\HS$-model $\M$ given in
 \eqref{tmsound} to obtain a model satisfying $\ftmb$.
 The case when $\T$ contains an infinite descending chain of immediate predecessor points requires `symmetrical versions' of the used formulas and is left to the reader.

\smallskip
(\emph{ii}) We reduce  `halting' to $\Fin(<)$-satisfiability.
In order to achieve this, we omit \eqref{nonhalting}.
This completes the proof of the theorem.
\end{proof}


\section{Conclusions and open problems}

Our motivation for introducing the Horn fragments of \HS{} and investigating  their computational behaviour comes from two sources. The first one is applications for ontology-based access to temporal data, where an ontology provides definitions of complex temporal predicates that can be employed in user queries. Atemporal ontology-based data access (OBDA)~\cite{PLCD*08} with Horn description logics and profiles of \OWL{} is now  paving its way to industry~\cite{DBLP:conf/semweb/KharlamovHJLLPR15}, supported by OBDA systems such as Stardog~\cite{Perez-Urbina12}, Ultrawrap~\cite{DBLP:conf/semweb/SequedaAM14}, and the Optique platform~\cite{optique,ISWC13,DBLP:conf/semweb/KontchakovRRXZ14}. However, OBDA ontology languages were not designed for applications with \emph{temporal} data (sensor measurements, historical records, video or audio annotations, etc.). That the datalog extension of (multi-dimensional) $\HShb$ is sufficiently expressive for defining useful temporal predicates over historical and sensor data was shown by Kontchakov et al.~\citeyear{IJCAI16}, who also demonstrated experimentally the efficiency of $\HShb$ for query answering. We briefly discussed these applications in Section~\ref{sec:app}. (Other temporal ontology languages have been developed based on Horn fragments of the linear temporal logic LTL~\cite{DBLP:conf/ijcai/ArtaleKKRWZ15,DBLP:conf/ijcai/Gutierrez-Basulto15,DBLP:conf/ijcai/Gutierrez-Basulto16},  computational tree logic CTL~\cite{DBLP:conf/kr/Gutierrez-BasultoJ014}, and metric temporal logic MTL~\cite{DBLP:conf/ecai/Gutierrez-Basulto16,DBLP:conf/aaai/BrandtKKRXZ17}.)

Our second motivation originates in multi-dimensional modal logic~\cite{many_dimensional_modal_logics,Kurucz07}. Its formalisms try to capture the interactions
between modal operators representing time, space, knowledge, actions, etc., and are closely connected not only to $\HS$ but also to finite variable fragments of various kinds of predicate logics
(as first-order quantifiers can be regarded as propositional modal operators over interacting universal
relations).
While the satisfiability problem of the two-variable fragment of classical predicate logic is \NExpTime-complete
\cite{GKV:decptv},
taming even two-dimensional propositional modal logics over interacting transitive but not equivalence relations
 by designing their interesting fragments turned out to be a difficult task.
 Introducing syntactical restrictions (guards, monodicity) \cite{Hodkinson06,Deg&Fisher&Konev06,Hodkinsonetal00,HodkinsonWZ02,Hodkinsonetal03}
 and/or modifying the semantics by allowing various subsets of product-like domains
 \cite{GabelaiaKKWZ05,gkwz06,hk15}
 or restricting the available valuations
 \cite{GollerJL15}
 might result in decidable logics that are still very complex, ranging from \ExpSpace\ to non-primitive recursive.
In this context, it would be interesting to see whether Horn fragments of multi-dimensional modal
formalisms exhibit more acceptable computational properties. Here, we make a step in this direction.

This paper has launched an investigation of Horn fragments of the Halpern-Shoham interval temporal logic $\HS$, which provides a powerful framework for temporal representation and reasoning on the one hand, but is notorious for its nasty computational behaviour on the other. We classified the Horn fragments of $\HS$ along the four axes:
\begin{itemize}
\item the type of interval modal operators available in the fragment: boxes $[\R]$ or diamonds $\langle \R \rangle$, or both;

\item the type of the underlying timelines: discrete or dense linear orders;

\item the type of semantics for the interval relations: reflexive or irreflexive; and

\item the number of literals in Horn clauses: two in the \emph{core}  fragment or more.
\end{itemize}
Both positive and negative results were obtained. The most unexpected negative results are the undecidability of (\emph{i}) $\HSc$ with both box and diamond operators under the irreflexive semantics, and of (\emph{ii}) $\HShb$ over discrete orders under the irreflexive semantics. Compared with (\emph{i}) and (\emph{ii}), the ubiquitous undecidability of $\HShd$ might look like a natural feature. Fortunately, we have also managed to identify a `chink in $\HS$'s armour' by proving that $\HShb$ turns out to be \emph{tractable} (\PTime-complete) over both discrete and dense orders under the reflexive semantics and over dense orders under the irreflexive semantics. First applications of the $\HShb$ fragment to  ontology-based data access over temporal databases or streamed data have been found by Kontchakov et al.~\citeyear{IJCAI16}.

Recently, Wa{\l}\c{e}ga~\citeyear{Prz2017} has considered a hybrid version of $\HShb$ (with nominals and the @-operator) and proved that it is NP-complete over  discrete and dense orders under the reflexive semantics and over dense orders under the irreflexive semantics.

In order to prove the undecidability results mentioned above as well as \PSpace-hardness of $\HSc$ under the reflexive semantics and of $\HScb$ over discrete orders under the irreflexive semantics, we developed a powerful toolkit that utilises the 2D character of $\HS$ and builds on various techniques and tricks from many-dimensional modal logic. However, we still do not completely understand the computational properties of the core fragment of $\HS$, leaving the following questions open:

\begin{question}
Are $\HSc$ and $\HScd$ decidable over any unbounded class of timelines under the reflexive semantics? What is the computational complexity?
\end{question}

\begin{question}
Is $\HScb$ decidable over any unbounded class of discrete timelines under the irreflexive semantics? What is the computational complexity?
\end{question}

In our Horn-$\HS$ logics, we did not restrict the set of available interval relations, which used to be one of the ways of obtaining decidable fragments. Classifying Horn fragments of $\HS$ along this axis can be an interesting direction for further research in the area. Syntactically, all of our Horn-$\HS$ logics are different. However, we do not know whether they are distinct in terms of their expressive power. Establishing an expressivity hierarchy of Horn fragments of $\HS$ (taking into account different semantical choices) can also be an interesting research question.


\begin{acks}
We are grateful to Przemys{\l}aw Wa{\l}\c{e}ga
for spotting and correcting a mistake in the preliminary version of the proof of
Theorem~\ref{t:pspacecoreboxirrefldisc}.
Thanks are also due to the anonymous reviewers for their useful suggestions.
\end{acks}

\bibliographystyle{acmsmall}

\end{document}